%% file: umlaut.tex
\documentclass[11pt,aps,pra,notitlepage,tightenlines,nofootinbib,superscriptaddress]{revtex4-2}
\pdfoutput=1

\input{marcos}

\newtheorem{conv}[thm]{Convention}

\newcommand{\opns}[1]{\ensuremath{E_0^{\rm NS}(#1)}} 
\newcommand{\opom}[1]{\ensuremath{E_0^{\Omega}(#1)}} 

\newcommand{\errns}[1]{\ensuremath{E^{\rm NS}(0^+,#1)}} 
\newcommand{\errom}[1]{\ensuremath{E^{\Omega}(0^+,#1)}} 
\newcommand{\errpl}[1]{\ensuremath{E^{\emptyset}(0^+,#1)}}          

\newcommand{\errL}[1]{\ensuremath{E^{\emptyset}_L(0^+,#1)}}

\usepackage{tikz}

\begin{document}

\title{Umlaut information}

\author{Filippo Girardi}
\email{filippo.girardi@sns.it}
\affiliation{Scuola Normale Superiore, Piazza dei Cavalieri 7, 56126 Pisa, Italy}
\affiliation{QuSoft, Science Park 123, 1098 XG Amsterdam, The Netherlands}
\affiliation{Korteweg--de Vries Institute for Mathematics, University of Amsterdam, Science Park 105-107, 1098 XG Amsterdam, The Netherlands}

\author{Aadil Oufkir}
\affiliation{Institute for Quantum Information, RWTH Aachen University,  Germany}

    \author{Bartosz Regula}
\affiliation{Mathematical Quantum Information RIKEN Hakubi Research Team, RIKEN Pioneering Research Institute (PRI) and RIKEN Center for Quantum Computing (RQC), Wako, Saitama 351-0198, Japan}

\author{Marco~Tomamichel}
\affiliation{Centre for Quantum Technologies, National University of Singapore, Singapore}
\affiliation{Department of Electrical and Computer Engineering, National University of Singapore, Singapore}

\author{Mario Berta}
\affiliation{Institute for Quantum Information, RWTH Aachen University,  Germany}

\author{Ludovico Lami}
\affiliation{Scuola Normale Superiore, Piazza dei Cavalieri 7, 56126 Pisa, Italy}
\affiliation{QuSoft, Science Park 123, 1098 XG Amsterdam, The Netherlands}
\affiliation{Korteweg--de Vries Institute for Mathematics, University of Amsterdam, Science Park 105-107, 1098 XG Amsterdam, The Netherlands}

\begin{abstract}
The sphere-packing bound quantifies the error exponent for noisy channel coding for rates above a critical value. Here, 
we study the zero-rate limit of the sphere-packing bound and show that it has an intriguing single-letter form, which we call the umlaut information of the channel, inspired by the lautum information introduced by Palomar and Verd\'u. Unlike the latter quantity, we show that the umlaut information is additive for parallel uses of channels. We show that it has a twofold operational interpretation: as the zero-rate error exponent of non-signalling--assisted coding on the one hand, and as the zero-rate error exponent of list decoding in the large list limit on the other.
\end{abstract}

\maketitle



\section{Introduction}

The sphere-packing bound $E_{\rm sp}(r, \pazocal{W})$~\cite{SHANNON196765} is a fundamental restriction on the error exponent (reliability function) of coding over a noisy channel $\pazocal{W}$ at a rate $r$. However, it only acquires a precise operational interpretation for rates above a certain critical rate, as, in general, it cannot be achieved for rates below this value~\cite{Gallager1965}. This may cast some doubts on the operational relevance of the sphere-packing bound in the low-rate regime. We observe, however, that in the limit $r \to 0$, the bound takes a rather curious form (see Sections~\ref{sec:prelims}--\ref{sec:channels} for detailed definitions and all derivations):
\begin{equation}\begin{aligned}\label{intro:sphere}
    \lim_{r\to0^+} E_{\rm sp}(r, \pazocal{W}) = \max_{P_X} \min_{Q_Y} D(P_XQ_Y\| P_{XY}),
\end{aligned}\end{equation}
where $P_{XY}(x,y) = \pazocal{W}(y|x) P_X(x)$ and $D$ denotes the relative entropy (Kullback--Leibler divergence). Interpreted as a correlation measure between the random variables $X$ and $Y$, the quantity on the right-hand side of~\eqref{intro:sphere} bears a certain resemblance to the mutual information
\begin{equation}\label{intro:mut-info}
I(X\!:\!Y) = D(P_{XY} \| P_X P_Y) = \min_{Q_Y} D(P_{XY}\| P_XQ_Y).
\end{equation}
However, since the order of the arguments in~\eqref{intro:sphere} is reversed in comparison to~\eqref{intro:mut-info}, this then suggests a possible connection with the lautum information~\cite{Lautum_08} --- a reversed variant of the mutual information defined as
\begin{equation}\label{intro:laut-info}
L(X\!:\!Y) = D(P_{X}P_Y \| P_{XY}).
\end{equation}
Importantly, however, Palomar and Verd\'{u}~\cite{Lautum_08} noticed that, unlike the mutual information, the lautum information cannot be expressed using a variational form optimised over all distributions $Q_Y$:  they showed through a specific example that
\bb\label{intro:first_variational}
    L(X\!:\!Y)  > \min_{Q_Y} D(P_XQ_Y\| P_{XY}).
\ee
Due to this key distinction, the zero-rate limit of the sphere-packing bound in~\eqref{intro:sphere} does not correspond to an information measure for channels induced by the lautum information itself, but rather by a different type of measure of correlations for joint probability distributions.

Motivated by this insight, in this work we study the correlation measure given by the right-hand side of~\eqref{intro:first_variational} and use it to define an information measure for channels. We show that it exhibits superior properties to the lautum information, in particular, that it is additive under parallel composition of channels. We also establish several direct operational interpretations of this channel information measure in the context of channel coding and list decoding.

\bigskip

\textbf{Overview of results.}\,---  For a joint probability distribution $P_{XY}$ on $\XX\times \YY$ with marginal $P_X$, the \deff{umlaut information} between $X$ and $Y$ is defined as 
\bb\label{intro:uml}
    U(X;Y) \coloneqq \min_{Q_Y}D(P_XQ_Y\| P_{XY}).
\ee
We use the name `umlaut' --- an anagram of `lautum', itself a reversed spelling of `mutual' --- to emphasise the difference from the lautum information as originally defined in~\cite{Lautum_08}.\footnote{The terminology seems to us unavoidable, as `umlaut' and `lautum' are the only two possible anagrams of `mutual' that also mean something in either English or Latin, if one is willing to exclude the third, more eerie alternative `tumula'.} Previously, closely related quantities have also appeared in Ref.~\cite{Nuradha2024}, which considered the `oveloh information' as a lautum-style reverse variant of the Holevo information~\cite{Holevo1973} for special cases of classical-quantum states, and in Ref.~\cite{ji_2024}, where $\alpha$-R\'enyi entropy variants of the umlaut information were defined for quantum states under the name of Petz–R\'enyi lautum information.

Leveraging the Gibbs variational principle, we can show that the minimiser is unique and leads to the following closed-form expression of the umlaut information: 
\bb
        U(X;Y)&=-H(P_X)-\log\sum_{y\in\YY}\exp\left(\sum_{x\in\XX}P_X(x)\log P_{XY}(x,y)\right).
    \ee
This formulation allows us to prove the additivity of the umlaut information under tensor product. In turn, this then implies an operational interpretation of the umlaut information in composite hypothesis testing. More precisely, for a joint probability distribution $P_{XY}$ over $\XX\times \YY$ with marginal $P_X$, we establish that the umlaut information $U(X;Y)$ governs the exponential decay rate of the type II error probability while maintaining a type I error probability smaller than a constant $\epsilon>0$ for the problem of testing
\begin{align}
\text{$H_0 : R_{X^nY^n}=P_X^{\times n}Q_{Y^n}$\quad versus\quad $H_1 : R_{X^nY^n}=P_{XY}^{\times n}$,}
\end{align}
where $Q_{Y^n}$ is an arbitrary probability distribution on $\mathcal{Y^n}$.

We can further show that the additivity property holds even in the channel setting where we optimise the umlaut information over input distribution $P_X$. More precisely, given a channel $\pazocal{W}_{Y|X}$, the \deff{channel umlaut information} is defined as 
\bb\label{intro:lautum_ch}
    U(\pazocal{W})&\coloneqq \max_{P_X}U(X;Y)=    \max_{P_X}\min_{Q_Y}D(P_{X}Q_Y\|P_{XY}),
\ee
where $P_{XY} \coloneqq  \pazocal{W}_{Y|X}P_X$. Then, for two channels $\pazocal{W}_1$ and $\pazocal{W}_2$, it holds that 
\bb\label{eq:8}
U(\pazocal{W}_1 \times \pazocal{W}_2) = U(\pazocal{W}_1) +U(\pazocal{W}_2). 
\ee
A similar property holds also for variants of the umlaut information defined in terms of Rényi $\alpha$-divergences, which we also introduce and study in detail.

This then allows for operational interpretations in the setting of noisy channel coding. We consider the problem of coding over a channel $\pazocal{W}_{Y|X}$ and in the block-length setting, the task is to send a number of messages $M=\exp({rn})$ with rate $r\ge 0$ through a certain number $n\in \mathbb{N}$ of copies of the channel $\pazocal{W}$, i.e. $\pazocal{W}^{\times n}$. As is well-known, for rates strictly below  the channel capacity $r< C(\pazocal{W})$, the (average) error probability $\epsilon(\exp(rn), \pazocal{W}^{\times n})$ vanishes exponentially fast as 
\bb
\epsilon(\exp(rn), \pazocal{W}^{\times n}) \simeq  \exp(-n E(r, \pazocal{W})+o(n)), 
\ee
where $E(r, \pazocal{W})$ is known as the error exponent (sometimes also termed reliability function). Although the standard setting of communication relies on encoders and decoders without additional assistance, many converse results can be proved for more general, assisted coding schemes. In particular, Polyanskiy, Poor and Verd\'u~\cite{PPV} introduced the meta-converse bound that implies many well-known converse results in the literature, and Matthews~\cite{Matthews2012} later showed that this bound in fact exactly corresponds to the setting of coding assisted by non-signalling (NS) correlations. Now, in the zero-rate regime where  $\frac{1}{n}\log M \rightarrow 0$, we establish that the non-signalling--assisted error exponent is quantified by the channel umlaut information. That is, we find
\bb
\errns{\pazocal{W}}=E_{\mathrm{sp}}(0^+,\pazocal{W}) = U(\pazocal{W}),
\ee
where $E_{\mathrm{sp}}(0^+, \pazocal{W})$ denotes the zero-rate limit of the sphere-packing bound $E_{\mathrm{sp}}(r, \pazocal{W})$ from \cite{SHANNON196765}.

Returning to the unassisted setting, we can relax the decoding requirement to the list decoding where the error occurs only if the original message does not belong to the list of messages (of size $L\ge 1$) returned by the decoder. In this setting, the zero-rate unassisted reliability function is~\cite{Blinovsky2001Oct} (see also \cite{bondaschi2021} for further discussions)
\bb
        \errL{\pazocal{W}} = \max_{P_X}\sum_{x_1,\dots,x_{L+1}}P_X(x_1)\cdots P_X(x_{L+1})\left(-\log \sum_{y\in\YY}\sqrt[L+1]{\pazocal{W}(y|x_1)\cdots\pazocal{W}(y|x_{L+1})}\right).
    \ee
As the second operational interpretation of the channel lautum information, we show that for all $L\ge 1$, the channel umlaut information $U(\pazocal{W})$ provides an upper bound to the unassisted zero-rate error exponent $E_L(0^+, \pazocal{W})$ and establish that the rate $U(\pazocal{W})$ is in fact achieved in the large list limit $L\to\infty$. 
This is similar to, but different from, prior literature that studied an asymptotic limit of exponentially large list sizes and also connected it with the sphere packing bound~\cite{Merhav2017}.

The remainder of our manuscript is structured as follows. In Section~\ref{sec:prelims} we fix our notation; in Section~\ref{sec:states} we present the umlaut information and its operational interpretations for probability distributions; in Section~\ref{sec:channels} we discuss the channel umlaut information and its operational interpretations; finally, in Section~\ref{sec:outlook} we discuss some open questions.

\bigskip


\section{Notation and Preliminaries}\label{sec:prelims}
 
Given a measurable space $( \XX, \Sigma_{\XX} )$ and probability measures $P, Q$ on this space, the Kullback--Leibler divergence, or relative entropy, is defined as
\bb
    D(P\|Q)\coloneqq \int_{\XX}\,  \log\frac{\mathrm{d} P}{\mathrm{d} Q} \, \mathrm{d} P
\ee
if $P \ll Q$ and $+\infty$ otherwise. In the former case the Radon-Nikodym derivative $\frac{\mathrm{d} P}{\mathrm{d} Q}$ is well-defined $P$-everywhere. We will mostly work with 
finite spaces $\XX$, with $\Sigma_{\XX}$ the power set of $\XX$, in which case the relative entropy takes the form 
\begin{align}
    D(P\|Q) = \sum_{x \in \XX} P(x) \log \frac{P(x)}{Q(x)}\, ,
\end{align}
and $P$, $Q$ are the probability mass functions defined on $\XX$. The discussion of the case of a continuous probability density function is
relegated to Appendix~\ref{app:cv}.
We use standard notation to describe random variables, as summarised in Table~\ref{tbl:notation}.

\begin{table}[h!]
\setlength\tabcolsep{5pt}
\begin{tabular}{r|l}
    $\XX$ & finite alphabet of symbols $x\in\XX$ with cardinality $|\XX|$ \\
    $x^n$ & $x^n=(x_1,\dots,x_n)$ is an element of $\XX^n$ \\
    $x_i$  & i-th component of the vector $x^n$ \\
    $\mathcal{P}(\XX)$ & set of probability distributions on $\XX$ \\
    $P_X$ & probability distribution on $\XX$ \\
    $X$ & random variable taking values in $\XX$ with probability distribution $P_X$ \\
    $P_{X^n}$ & a probability distribution on $\XX^n$ \\
    $X^n$ & $X^n=(X_1,\dots,X_n)$ is a random variable taking values in $\XX^n$ with probability distribution $P_{X^n}$ \\
    $P_{X_i}$ & i-th marginal of $P_{X^n}$, i.e. the probability distribution of $X_i$ \\
\end{tabular}
\caption{Summary of our notation.} \label{tbl:notation}
\end{table}

The Shannon entropy of $P$ is defined as
\begin{align}
    H(P)\coloneqq -\sum_{x\in\XX} P(x) \log P(x) \,.
\end{align}
We will later refer to several properties of the relative entropy, which we are going to quickly summarise here.

\begin{description}
\item [Positive definiteness] This is the property that $D(P\|Q) \geq 0$ with equality if and only if $P = Q$ as measures. 
\end{description}

Positive definiteness is closely related to the following well known inequality.\footnote{To prove Lemma~\ref{thm:Gibbs} from positive definiteness, it suffices to write the left-hand side of~\eqref{Gibbs} as $D(P\|P_0) - \log \sum_{x\in \XX} \exp[-A(x)]$, where $P_0$ is the probability distribution defined by the right-hand side of~\eqref{Gibbs_optimiser}, and then observe that $D(P\|P_0)\geq 0$, with equality if and only if $P=P_0$.}
\begin{lemma}[(Gibbs variational principle)]\label{thm:Gibbs}
Let $A:\XX\to\mathbb{R}$ be a function. For any $P\in\mathcal{P}(\XX)$,
    \begin{equation}
       -H(P)+\mathbb{E}_P[A(X)]\geq - \log\sum_{x\in\XX} \exp[-A(x)]\, ,
    \label{Gibbs}
    \end{equation}
    with equality if and only if
    \begin{equation}
        P(x)=\frac{\exp[-A(x)]}{\sum_{x'\in\XX}\exp[-A(x')]},
    \label{Gibbs_optimiser}
    \end{equation}
    where $\mathbb{E}_P[A(X)]\coloneqq\sum_{x\in\XX}P(x)A(x)$ is the expectation value of $A(X)$ when $X$ is distributed according to $P$.
\end{lemma}

\begin{description}
\item[Joint convexity] The map $(P, Q) \mapsto D(P\|Q)$ is jointly convex.

\item[Monotonicity] The relative entropy is monotone under simultanous application of stochastic maps (channels, in the following) to both arguments. 

\item[Additivity] The relative entropy is additive when both arguments are tensor products, i.e., $D(P_1 \times P_2\| Q_1 \times Q_2) = D(P_1\|Q_1) + D (P_2\|Q_2)$.
\end{description}


\section{Umlaut information}\label{sec:states}

\subsection{Definition and basic properties}

The mutual information between two random variables $X$ and $Y$ with joint probability distribution $P_{XY}$ and marginals $P_X$ and $P_Y$ is defined as
\bb
I(X\!:\!Y)\coloneqq D(P_{XY}\| P_XP_Y) \,.
\ee
Inspired by the this definition, Palomar and Verd\'{u}~\cite{Lautum_08} defined the lautum information as
\begin{equation}
L(X\!:\!Y)\coloneqq D(P_XP_Y\|P_{XY})
\end{equation}
and established some fundamental properties of this quantity. Both the mutual information and the lautum information can measure the independence between the random variables $X$ and $Y$, as they both vanish only when the random variables are independent. Additionally, they provide upper bounds on the total variation between the joint probability distributions and  the product its marginals. They have also been used to provide generalization bounds for learning algorithms~\cite{Issa}. 

Despite their superficial similarity, there are many differences between mutual and lautum information: for instance, $I(X\!:\!Y)$ can be expressed as a linear combination of entropic measures~\cite{Verdu2019Jul}, while the same property is not true for $L(X\!:\!Y)$~\cite{Lautum_08}. Furthermore, it is well known that the mutual information can be written in the following variational forms:
 \bb
    I(X\!:Y\!)= D(P_{XY}\| P_XP_Y)=\min_{Q_Y}D(P_{XY}\| P_XQ_Y)=\min_{Q_X}\min_{Q_Y}D(P_{XY}\| Q_XQ_Y).
 \ee
Palomar and Verd\'{u} pointed out that adopting such an optimised definition for the lautum information would produce a different result: they showed that $ L(X\!:\!Y)>\min_{Q_Y}D(P_XQ_Y\| P_{XY})$
for the counterexample
\bb
    P_{XY}(x,y)=\begin{cases}
        0 &x=y=0\\
        \frac{1}{3}&\text{otherwise,}
    \end{cases}
\ee
where $\XX=\YY=\{0,1\}$.
The aim of this paper is to reconsider the variational formulation and to study its properties and operational interpretation in communication.

\begin{Def}[(Umlaut information)]\label{def:lautum}
Given two random variables $X$ and $Y$ taking values in $\XX$ and $\YY$ with joint probability distribution $P_{XY}$, the umlaut information is defined as
\bb
U(X;Y)&\coloneqq \min_{Q_Y}D(P_XQ_Y\|P_{XY})
\ee
where $P_X$ is the marginal of $P_{XY}$ on $\XX$ and $Q_Y\in \mathcal{P}(\YY)$.
\end{Def}

While the lautum information $L$ is symmetric under exchange of $X$ and $Y$, in general the umlaut information $U$ is not: because of this reason, for the latter we use the semicolon instead of the colon. This asymmetry will turn out to be particularly meaningful in the setting of communication theory, as the role of the sender and the receiver is not symmetric in general. Let us first establish some basic properties of the umlaut information.

\begin{prop} \label{prop:lautum_prop}
    The umlaut information satisfies the following properties:
    \begin{enumerate}
        \item[(1)] Positive definiteness: $U(X;Y) \geq 0$ with equality if and only if $P_{XY} = P_X P_Y$.
        \item[(2)] Boundedness: $U(X;Y) < \infty \iff \exists y \in \YY$ such that $\forall x \in \XX$ either $P_{Y|X}(y|x) > 0$ or $P_X(x) = 0$, i.e., if there is a symbol in $\YY$ that cannot exclude any symbol in $\XX$.
        \item[(3)] Data-processing inequality: Consider a Markov chain $X - Y - Z$. Then,
        \begin{align}
            U(X;Z) \leq U(X;Y) \quad \textnormal{and} \quad U(X;Z) \leq U(Y;Z) \,.
        \end{align}
    \end{enumerate}    
\end{prop}

\begin{proof}
    For (1), note that from the positive definiteness of the relative entropy we gather that $U(X;Y) = 0 \iff P_{XY} = P_X Q_Y$ for some $Q_Y$, but this implies that $Q_Y = P_Y$. 
    
    For (2), to show ``$\impliedby$'', we first observe that the candidate distribution with $Q(y) = 1$ makes $U(X; Y)$ finite. To show the contrapositive, we note that if $\exists x \in \XX$ such that $P_{Y|X}(y|x) = 0$ and $P_X(x) > 0$ then we must have $Q(y) = 0$ in order for the condition $P_X \times Q_Y \ll P_{XY}$ to be met. But since this must be so for all $y \in \YY$, $Q_Y$ cannot be a probability distribution.

    For (3), it suffices to show the first inequality since $X - Y - Z$ implies $Z - Y - X$. The Markov condition than implies the existence of a channel
    $\pazocal{W}_{Z|Y}$ such that $P_{XZ}(x,z) = \sum_{y \in \YY} P_{XY}(x,y) \pazocal{W}(z|y)$. Using the monotonicity of relative entropy under channels we arrive at
    \begin{align}
        U(X;Y) = \min_{Q_Y} D(P_X Q_Y \| P_{XY}) \geq \min_{Q_Y} D(P_X \tilde{Q}_Z \| P_{XZ}) \geq \min_{Q_Z} D(P_X Q_Z \| P_{XZ}) = U(X; Z) \,.
    \end{align}
    where we used that $\tilde{Q}_Z(z) = \sum_{y \in \YY} Q_Y(y) \pazocal{W}(z|y)$ is a candidate for the minimization over $\mathcal{P}(Z)$.
\end{proof}


\subsection{Closed form and optimal marginal}

In this subsection we find the minimiser appearing in our variational definition of the umlaut information. As a consequence, we can write down a closed formula for $U(X;Y)$ and prove additivity. The same computations can be done in the case of the Rényi $\alpha$-umlaut information; also in this case additivity holds. 

\begin{Def}[(Umlaut-marginal of a joint distribution)]\label{def:lautum_marginal}
Given a joint probability distribution $P_{XY}\in\mathcal{P}(\XX\times\YY)$, the umlaut-marginal $\uml{P}_Y$ of $P_{XY}$ on the alphabet $\YY$ is defined as
    \bb
        \uml{P}_Y(y)\coloneqq \frac{1}{Z[P_{XY}]}\exp\left(\sum_{x\in\XX}P_X(x)\log P_{XY}(x,y)\right)\qquad\forall y\in\YY,
    \ee
    where $Z[P_{XY}]$ is the constant making $\uml{P}_Y(y)$ a probability distribution on $\YY$ and $P_X$ is the marginal of $P_{XY}$ on $\XX$.
\end{Def}

In the following theorem we are going to prove that the umlaut-marginal of a joint distribution is the optimiser of the variational problem appearing in the definition of $U$.

\begin{prop}[(A closed formula for the umlaut information)]\label{thm:closed_lautum}
Given a joint probability distribution $P_{XY}\in\mathcal{P}(\XX\times\YY)$, let $P_X$ be the marginal on $\XX$ and $\uml{P}_Y$ be the umlaut-marginal on $\YY$. Then the variational formulation of the umlaut information provided in Definition~\ref{def:lautum} admits the following explicit representation:
    \bb
        U(X;Y)&=D(P_X\uml{P}_Y\|P_{XY})\\
        &=-H(P_X)-\log\sum_{y\in\YY}\exp\left(\sum_{x\in\XX}P_X(x)\log P_{XY}(x,y)\right).
    \ee
    Furthermore, $\uml{P}_Y$ is the unique minimiser in the definition of $U$.
\end{prop}

\begin{proof}
Let us compute
\bb
    U(X;Y)&= \min_{Q_Y}D(P_XQ_Y\|P_{XY})\\
    &=-H(P_X)+\min_{Q_Y} \sum_{y\in\YY}Q_Y(y)\left(\log Q_Y(y)-\sum_{x\in\XX}P_X(x)\log P_{XY}(x,y) \right)\\
    &= -H(P_X)-\log\sum_{y\in\YY}\exp\left(\sum_{x\in\XX}P_X(x)\log P_{XY}(x,y)\right),
\ee
where in the last line we have used the Gibbs variational principle (Lemma~\ref{thm:Gibbs}) to solve the minimization problem over $Q_Y$. The unique minimiser $\uml{P}_Y$ turns out to be the Gibbs state corresponding to 
\bb
A(y) = -\sum_{x\in\XX}P_X(x)\log P_{XY}(x,y)
\ee
in the statement of Lemma~\ref{thm:Gibbs}, which is the umlaut-marginal of $P_{XY}$ on $\YY$.
\end{proof}

An alternative argument to derive the closed form of the umlaut information presented in the above Proposition~\ref{thm:closed_lautum} can be obtained by carefully taking the limit $\alpha\to 1$ in~\cite[Lemma~38]{ji_2024}. Using the explicit formula for $U$, it is elementary to prove its additivity.

\begin{cor}[(Additivity of $U$)]\label{cor:additivity}
Given two joint probability distributions $P_{XY}\in\mathcal{P}(\XX\times\YY)$ and $P'_{X'Y'}\in\mathcal{P}(\XX'\times\YY')$, let us consider the pairs of random variables $(X,X')$ and $(Y,Y')$ distributed according the product $P_{XY}P'_{X'Y'}$. Then their umlaut information is given by the sum
    \begin{equation}
        U(XX';YY')=U(X;Y)+U(X';Y').
    \end{equation}
\end{cor}

\begin{proof}
    The umlaut-marginal $\uml{P}_{YY'}$ of $P_{XY}P'_{X'Y'}$ factorises due to the additivity of the logarithm:
\bb
        \uml{P}_{YY'}(y,y')&= \frac{1}{Z[P_{XY}P'_{X'Y'}]}\exp\left(\sum_{\substack{x\in\XX\\x'\in\mathcal{X'}}}P_X(x)P_{X'}(x')\log \left(P_{XY}(x,y)P'_{X'Y'}(x',y')\right)\right)\\
        &= \frac{1}{Z[P_{XY}P'_{X'Y'}]}\exp\left(\sum_{x\in\XX}P_X(x)\log P_{XY}(x,y)+\sum_{x'\in\XX'}P_{X'}(x')\log P_{X'Y'}(x',y')\right)\\
        &=\uml{P}_{Y}(y)\uml{P}'_{Y'}(y')
    \ee
whence, by the closed formula of Proposition~\ref{thm:closed_lautum},
\bb
    U(XX';YY')&=D(P_{X}P_{X'}\uml{P}_{YY'}\|P_{XY}P'_{X'Y'})\\
    &=D(P_{X}P_{X'}\uml{P}_{Y}\uml{P}'_{Y'}\|P_{XY}P'_{X'Y'})\\
    &=D(P_{X}\uml{P}_{Y}\|P_{XY})+D(P'_{X}\uml{P}'_{Y'}\|P'_{X'Y'})\\
    &=U(X;Y)+U(X';Y'),
\ee
which proves the additivity of $U$.
\end{proof}

\begin{rem}
    We can manipulate the closed-form expression for the umlaut information in order to find an alternative form:
    \bb\label{eq:alternative_um}
        U(X;Y)&=-H(P_X)-\log\sum_{y\in\YY}\exp\left(\sum_{x\in\XX}P_X(x)\log P_{XY}(x,y)\right)\\
        &=-\log\sum_{y\in\YY}P_Y(y)\exp\left(\sum_{x\in\XX}P_X(x)\log \frac{P_{XY}(x,y)}{P_X(x)P_Y(y)}\right)\\
        &=-\log \mathbb{E}_{Z\sim P_Y}\left[  \exp\left(-D(P_X\|P_{X|Y=Z})\right)\right].
 \ee
\end{rem}

\subsection{Rényi divergence variant}

The definition the umlaut information (Definition~\ref{def:lautum}) can be generalised in terms of the Rényi $\alpha$-relative entropies. These quantities previously appeared in~\cite[Appendix~E]{ji_2024} as a tool used to obtain bounds on the error probability of (quantum) state exclusion.

\begin{Def}[(R\'enyi $\alpha$-umlaut information)] \label{def:renyi-umlaut}
Let $\alpha\in(0,1)\cup (1,\infty)$. Given two random variables $X$ and $Y$ taking values in $\XX$ and $\YY$ with joint probability distribution $P_{XY}$, the R\'enyi $\alpha$-umlaut information is defined as
\bb
U_\alpha(X;Y)&\coloneqq \min_{Q_Y}D_\alpha(P_XQ_Y\|P_{XY})
\ee
where $P_X$ is the marginal of $P_{XY}$ on $\XX$ and $Q_Y\in \mathcal{P}(\YY)$; $D_\alpha$ is the R\'enyi $\alpha$-relative entropy: given $P,Q\in\mathcal{P}(\XX)$,
\bb\label{eq:Renyi}
    D_{\alpha}(P\|Q)\coloneqq \frac{1}{\alpha-1}\log\sum_{x\in \XX}P^\alpha(x)Q^{1-\alpha}(x)\, .
\ee
\end{Def}

\begin{lemma}\label{lem:alpha_to_1}
     Given two random variables $X$ and $Y$, it holds that
    \bb
        \lim_{\alpha\to 1^-} U_\alpha(X;Y)=U(X;Y)
    \ee
    and
    \bb
        \lim_{\alpha\to 1^+} U_\alpha(X;Y)=U(X;Y).
    \ee
\end{lemma}

\begin{proof}
    Let $P_{XY}$ be the joint probability distribution of $X$ and $Y$. Since $\alpha\mapsto D_\alpha(p\|q)$ is a monotonically increasing function, for $0<\alpha<1$ we can write
    \bb
        \lim_{\alpha\to 1^-} U_\alpha(X;Y)=\sup_{\alpha<1}\min_{Q_Y}D_\alpha(P_XQ_Y\|P_{XY})
    \ee
    By the Mosonyi--Hiai minimax theorem~\cite[Corollary A2]{MosonyiHiai}, we can rewrite
    \bb
        \lim_{\alpha\to 1^-} U_\alpha(X;Y)=\min_{Q_Y}\sup_{\alpha<1}D_\alpha(P_XQ_Y\|P_{XY})\eqt{(i)}\min_{Q_Y}D(P_XQ_Y\|P_{XY})=U(X;Y),
    \ee
    where in~(i) we have used that the R\'{e}nyi relative entropies converge to the Kullback--Leibler divergence as $\alpha\to 1^-$. For $\alpha>1$, leveraging again on the monotonicity of $\alpha\mapsto D_\alpha(p\|q)$, we can easily compute
    \bb
        \lim_{\alpha\to 1^+}U_\alpha(X;Y)&=\inf_{\alpha<1}\min_{Q_Y}D_\alpha(P_XQ_Y\|P_{XY})=\min_{Q_Y}\inf_{\alpha<1}D_\alpha(P_XQ_Y\|P_{XY})\\
        &\eqt{(iii)}\min_{Q_Y}D(P_XQ_Y\|P_{XY})=U(X;Y),
    \ee
    where in~(iii) we have used that the Rényi relative entropies converge to the Kullback--Leibler divergence as $\alpha\to 1^+$.
\end{proof}

Using the approach of Sibson~\cite{sibson_1969}, for the Rényi $\alpha$-umlaut information a similar closed formula can be shown \cite[Lemma~38]{ji_2024}. Remarkably, it also turns out to be additive.

\begin{prop}[{(A closed formula and additivity for $U_\alpha$ \cite[Lemma~38]{ji_2024})}]\label{thm:closed_alpha}
    For any $\alpha\in(0,1)\cup (1,\infty)$, given a joint probability distribution $P_{XY}\in\mathcal{P}(\XX\times\YY)$, it holds that
    \bb\label{eq:closed1}
        U_\alpha(X;Y)= D_\alpha\big(P_X\uml{P}^{(\alpha)}_Y\,\big\|\,P_{XY}\big) = -\log \sum_{y \in \YY} P_Y(y)  \left( \sum_{x\in\XX}P_X^\alpha(x)P_{X|Y}^{1-\alpha}(x|y) \right)^{\frac{1}{1-\alpha}},
    \ee
    where $P_X$ is the marginal of $P_{XY}$ on $\XX$ and, assuming $U_\alpha(X;Y)<+\infty$,
    \bb\label{eq:closed2}
        \uml{P}_Y^{(\alpha)}(y)\coloneqq \frac{1}{Z_\alpha[P_{XY}]}\left(\sum_{x\in\XX}P_X^\alpha(x)P_{XY}^{1-\alpha}(x,y)\right)^{\frac{1}{1-\alpha}},
    \ee
    where $Z_\alpha[P_{XY}]$ is the normalisation constant making $\uml{P}_Y^{(\alpha)}$ a probability distribution on $\YY$, with the conventions $0^{-1}=+\infty$ and $+\infty^{-1}=0$. If $U_\alpha(X;Y)=+\infty$, then $\uml{P}_Y^{(\alpha)}$ can be taken to be any distribution in $\mathcal{P}(\YY)$. In particular, $U_\alpha$ is additive for $\alpha\in(0,1)\cup (1,\infty)$. 
\end{prop}

\begin{proof}
    Let us consider the case $\alpha\in(0,1)$ and rewrite
\bb
    U_\alpha(X;Y)&=\min_{Q_Y}\frac{1}{\alpha-1}\log\sum_{\substack{x\in\XX\\ y\in\YY}}P_X^\alpha(x)Q_Y^\alpha(y)P_{XY}^{1-\alpha}(x,y)\\
    &=\frac{1}{\alpha-1}\log \max_{Q_Y}\sum_{y\in\YY}Q_Y^\alpha(y)\sum_{x\in\XX}P_X^\alpha(x)P_{XY}^{1-\alpha}(x,y).
\ee
Let $f,g:\YY\to \R$ be defined as $f(y)\coloneqq Q_Y^\alpha(y)$ and $g\coloneqq \sum_{x\in\XX}P_X^\alpha(x)P_{XY}^{1-\alpha}(x,y)$, for all $y\in \YY$. Then
\bb
f,g\geq 0, \qquad \|f\|_{1/\alpha}=1\, ,
\ee
where, for an arbitrary $p>0$, we defined the norm $\|f\|_p\coloneqq \left(\sum_{y\in \YY} |f(y)|^p\right)^{1/p}$. Using this observation, we can alternatively cast the optimisation over $Q_Y$ as an optimisation over vectors $f\geq 0$ such that $\|f\|_{1/\alpha} =1$. Then
\bb\label{eq:U_a}
    U_\alpha(X;Y)
    &=\frac{1}{\alpha-1}\log \max_{\substack{f\geq 0,\\ \|f\|_{1/\alpha}=1}}\langle f, g\rangle = \frac{1}{\alpha-1}\log \|g\|_{\frac{1}{1-\alpha}}=-\log \left\|\sum_{x\in\XX}P_X^\alpha(x)P_{XY}^{1-\alpha}(x,\,\cdot\,)\right\|_{\frac{1}{1-\alpha}}^{\frac{1}{1-\alpha}},
\ee
where in the second equality we observed that the maximum is achieved by $f = g^{\frac{\alpha}{1-\alpha}}\big/\|g\|_{\frac{1}{1-\alpha}}^{\frac{\alpha}{1-\alpha}}$, according to H\"older's inequality; in other words,
\bb \label{def:alpha-marginal}
f=\frac{\left(\sum_{x\in\XX}P_X^\alpha(x)P_{XY}^{1-\alpha}(x,\,\cdot\,)\right)^{\frac{\alpha}{1-\alpha}}}{\left[
\sum_{y\in\YY}\left(\sum_{x\in\XX} P_X^\alpha(x)P_{XY}^{1-\alpha}(x,y)\right)^{\frac{1}{1-\alpha}}\right]^{\alpha} } \eqcolon \left(\uml{P}_{Y}^{(\alpha)}\right)^{\alpha}.
\ee
We can rewrite~\eqref{eq:U_a} further:
\bb\label{eq:as_later}
 U_\alpha(X;Y)&=-\log \sum_{y \in \YY} P_Y(y)  \left( \sum_{x\in\XX}P_X^\alpha(x)P_{X|Y}^{1-\alpha}(x|y) \right)^{\frac{1}{1-\alpha}}.
\ee
Let us call in general $\uml{P}_{Y}^{(\alpha)}$ the $\alpha$-umlaut marginal of an arbitrary distribution $P_{XY}$, as in~\eqref{def:alpha-marginal}.
If we consider any product distribution $P_{XY}P'_{X'Y'}$, it is easy to verify that its $\alpha$-umlaut reduced state is the tensor product of $\uml{P}_{Y}$ and $\uml{P}'_{Y'}$. Since the Rényi quantum relative entropies are additive, this concludes the proof of the additivity of $U_\alpha$.

We now consider the case of $\alpha>1$. Let us assume $U_\alpha(X;Y)<\infty$.
\bb\label{eq:alphag1}
    U_\alpha(X;Y)&=\min_{Q_Y}\frac{1}{\alpha-1}\log\sum_{\substack{x\in\XX\\ y\in\YY}}P_X^\alpha(x)Q_Y^\alpha(y)P_{XY}^{1-\alpha}(x,y)\\
    &\eqt{(i)}\min_{Q_Y\in\mathcal{P}(\YY^\ast)}\frac{1}{\alpha-1}\log\sum_{\substack{x\in\mathrm{supp}(P_X)\\ y\in\YY^\ast}}P_X^\alpha(x)Q_Y^\alpha(y)P_{XY}^{1-\alpha}(x,y)\\
    &\eqt{(ii)}\min_{Q_Y\in\mathcal{P}(\YY^\ast)}\frac{1}{\alpha-1}\log \sum_{y\in\YY^\ast}Q_Y^\alpha(y)\left(\frac{1}{Z_\alpha[P_{XY}]}\left(\sum_{x\in\XX}P_X^\alpha(x)P_{XY}^{1-\alpha}(x,y)\right)^{\frac{1}{1-\alpha}}\right)^{1-\alpha}-\log Z_\alpha[P_{XY}]\\
    &\eqt{(iii)}\min_{Q_Y\in\mathcal{P}(\YY^\ast)}D_{\alpha}(Q_Y\|\uml{P}_Y^{(\alpha)})-\log Z_{\alpha}[P_{XY}]
\ee
where, starting from~(i), we denote $\YY^\ast\coloneqq\{y\in\YY: P(x,y)>0\ \ \forall x\in\mathrm{supp}(P_X)\}$; we can indeed restrict the optimization to probability distributions supported on $\YY^\ast$ since, if $Q_Y(y)>0$ for $y\in\YY\setminus\YY^\ast$, then we would have
\bb
    \sum_{\substack{x\in\XX\\ y\in\YY}}P_X^\alpha(x)Q_Y^\alpha(y)P_{XY}^{1-\alpha}(x,y)=+\infty.
\ee
Clearly, $\YY^\ast=\emptyset$ if and only if $U_\alpha(X;Y)=+\infty$. In (ii) we have introduced
\bb
    Z_\alpha[P_{XY}]\coloneqq \sum_{y\in\YY^\ast}\left(\sum_{x\in\XX}P_X^\alpha(x)P_{XY}^{1-\alpha}(x,y)\right)^{\frac{1}{1-\alpha}}
\ee
and in (iii) we have defined
\bb
    \uml{P}_Y^{(\alpha)}(y)\coloneqq\frac{1}{Z_\alpha[P_{XY}]}\left(\sum_{x\in\XX}P_X^\alpha(x)P_{XY}^{1-\alpha}(x,y)\right)^{\frac{1}{1-\alpha}}\qquad y\in\YY^\ast,
\ee
which is a probability distribution on $\YY^\ast$ and which can naturally be extended to a probability distribution on $\YY$ by setting $\uml{P}_Y^{(\alpha)}(y)=0$ for all $y\in\YY\setminus\YY^\ast$. Using the conventions $0^{-1}=+\infty$ $+\infty^{-1}=0$, since $\alpha>1$ we could formally rewrite
\bb
    Z_\alpha[P_{XY}]=\sum_{y\in\YY}\left(\sum_{x\in\XX}P_X^\alpha(x)P_{XY}^{1-\alpha}(x,y)\right)^{\frac{1}{1-\alpha}},
\ee
extending the sum to $\YY$, and similarly
\bb
    \uml{P}_Y^{(\alpha)}(y)\coloneqq\frac{1}{Z_\alpha[P_{XY}]}\left(\sum_{x\in\XX}P_X^\alpha(x)P_{XY}^{1-\alpha}(x,y)\right)^{\frac{1}{1-\alpha}}, \qquad y\in\YY,
\ee
Looking at the last line of~\eqref{eq:alphag1}, we notice that
\bb
     U_\alpha(X;Y)=-\log Z_\alpha[P_{XY}]
\ee
and that the minimisation problem is saturated for $Q_Y=\uml{P}_Y^{(\alpha)}$, which concludes the proof of~\eqref{eq:closed1} and~\eqref{eq:closed2}. If we consider $P_{XY}P'_{X'Y'}$ and we construct
    \bb
        (\YY\times \YY')^\ast = \{(y,y')\in \YY\times \YY': P_{XY}(x,y)P'_{X'Y'}(x',y')>0\ \ \forall (x,x')\in\mathrm{supp}(P_XP'_{X'})\},
    \ee
    then $(\YY\times \YY')^\ast = \YY^\ast \times (\YY')^\ast$. Indeed,
    \bb
    P_{XY}(x,y)P'_{X'Y'}(x',y')>0\quad&\iff\quad P_{XY}(x,y)>0 \text{ and } P'_{X'Y'}(x',y')>0,\\
    (x,x')\in\mathrm{supp}(P_XP'_{X'})\quad&\iff\quad x\in\mathrm{supp}(P_X) \text{ and } x'\in\mathrm{supp}(P'_{X'}).
    \ee
    So, the optimiser in $U_\alpha(XX';YY')$ is proportional to
    \bb
        \id_{(\YY\times\YY')^\ast}(y,y')&\left(\sum_{\substack{x\in\XX\\x'\in\XX'}}P_X^\alpha(x)P_{X'}^\alpha(x')P_{XY}^{1-\alpha}(x,y){P'}_{X'Y'}^{1-\alpha}(x',y')\right)^{\frac{1}{1-\alpha}}\\
        &=\id_{\YY^\ast}(y)\id_{(\YY')^\ast}(y')\left(\sum_{x\in\XX}P_X^\alpha(x)P_{XY}^{1-\alpha}(x,y)\right)^{\frac{1}{1-\alpha}}\left(\sum_{x'\in\XX'}{P'}_{X'}^\alpha(x')P_{XY}^{1-\alpha}(x',y')\right)^{\frac{1}{1-\alpha}},
    \ee
    whence
    \bb
        U_\alpha(XX';YY')&=D_\alpha(P_X\uml{P}_Y^{(\alpha)}P'_{X'}\uml{P'}_{Y'}^{(\alpha)}\|P_{XY}P'_{X'Y'})\\
        &=D_\alpha(P_X\uml{P}_Y^{(\alpha)}\|P_{XY})+D_\alpha(P'_{X'}\uml{P'}_{Y'}^{(\alpha)}\|P'_{X'Y'})=U_\alpha(X;Y)+U_\alpha(X';Y').
    \ee
 This concludes the proof of the additivity.

\end{proof}

\begin{rem}
    For $\alpha\in(0,1)\cup(1,\infty)$, we can rewrite, similarly to~\eqref{eq:alternative_um},
    \bb\label{eq:new_form}
        U_\alpha(X;Y)&=-\log \sum_{y \in \YY} P_Y(y)  \left( \sum_{x\in\XX}P_X^\alpha(x)P_{X|Y}^{1-\alpha}(x|y) \right)^{\frac{1}{1-\alpha}}\\
        &=-\log \sum_{y \in \YY} P_Y(y)  \exp\left(-D_\alpha(P_X\|P_{X|Y=y})\right)\\
        &=-\log \mathbb{E}_{Z\sim P_Y}\left[  \exp\left(-D_\alpha(P_X\|P_{X|Y=Z})\right)\right].
 \ee
\end{rem}
This form suggests another proof of Lemma~\ref{lem:alpha_to_1}, where we showed that the Rényi $\alpha$-umlaut information converges to the umlaut information in the limit $\alpha \to 1$.
\begin{proof}[Alternative proof of Lemma~\ref{lem:alpha_to_1}]
    By the closed-form expression in~\eqref{eq:new_form}, we can compute
    \bb
        \lim_{\alpha\to 1^{\pm}}U_\alpha(X;Y) 
        &\eqt{(i)}-\log \sum_{y \in \YY} P_Y(y)  \exp\left(-\lim_{\alpha\to 1^\pm}D_\alpha(P_X\|P_{X|Y=y})\right)\\
        &=-\log \sum_{y \in \YY} P_Y(y)  \exp\left(-D(P_X\|P_{X|Y=y})\right)\\
        &=U(X;Y).
    \ee
    where in (i) we have used that the sum is over a finite set.
\end{proof}


\subsection{Operational interpretation in composite hypothesis testing}

Let $\XX$ be a finite set of symbols and let $\mathcal{P}(\XX)$ be the set of probability distributions on $\XX$. Let us suppose that a random variable $X$ is distributed according to $Q$ (null hypothesis $H_0$) or $P$ (alternative hypothesis $H_1$). 
The task of simple asymmetric hypothesis testing consists in sampling $n$ i.i.d.\ copies of $X$ in order to decide whether $H_0$ or $H_1$ holds. The asymmetry stems from the role of the two hypotheses: for instance, we may want to detect $H_1$ as efficiently as possible when it is the case, provided that under $H_0$ the probability of a false alarm is under a fixed threshold, or vice-versa. A test guessing between $H_0$ and $H_1$ starting from $n$ samples of $X$ is a map $\mathcal{A}_n:\XX^n\to \{H_0,H_1\}$, called acceptance function, which we do not necessarily assume to be deterministic. Let $X^n=(X_1,\dots,X_n)$ be the vector of $n$ i.i.d.\ copies of $X$ and let $Q^{\times n}$ and $P^{\times n}$ the corresponding probability distributions according to $H_0$ and $H_1$, respectively. Two kinds of error could occur in the guess of $\mathcal{A}_n$:
\begin{itemize}
    \item $H_0$ holds, but $\mathcal{A}_n$ guesses $H_1$ (error of type I, or false positive);
    \item $H_1$ holds, but $\mathcal{A}_n$ guesses $H_0$ (error of type II, or false negative).
\end{itemize}
The corresponding error probabilities are
\begin{align}
\alpha(\mathcal{A}_n)&\coloneqq \mathbb{P}_{X^n\sim Q^{\times n}}\left(\mathcal{A}_n(X_1,\dots,X_n)=H_1\right)&\text{type-I error}, \label{type_I_err_iid} \\
\beta(\mathcal{A}_n)&\coloneqq \mathbb{P}_{X^n\sim P^{\times n}}\left(\mathcal{A}_n(X_1,\dots,X_n)=H_0\right)&\text{type-II error}. \label{type_II_err_iid}
\end{align}
Given any $\epsilon\in (0,1)$, the hypothesis testing relative entropy is defined as
\bb
D^{\epsilon}_H(Q^{\times n}\|P^{\times n})\coloneqq -\log\min\left\{\beta(\mathcal{A}_n)\,:\, \alpha(\mathcal{A}_n)\leq \epsilon\right\}.
\label{D_H_iid}
\ee
The type II error exponent will asymptotically decay according to the error exponent
\bb
\text{Stein}(Q\|P)\coloneqq\lim_{\epsilon\to 0}\liminf_{n\to \infty}\frac{1}{n}D^\epsilon_H(Q^{\times n}\|P^{\times n})\, ,
\ee
which we refer to as 
Stein exponent. The Chernoff--Stein 
lemma~\cite{stein_unpublished, chernoff_1956} states that
\bb
\text{Stein}(Q\|P)=D(Q\|P)\, ,
\ee
where $D(Q\|P)$ is the relative entropy.
This means that, when the type-I error probability is below a fixed threshold, the optimal type-II error probability asymptotically decays as $\exp(-nD(Q\|P))$, i.e.\ the distinguishability between $Q$ and $P$ in this asymmetric setting is quantified by their relative entropy. 

Here, we are interested in the generalisation of the above setting where the null hypothesis is composite, i.e.\ is given by a set of non-i.i.d distributions. Namely, for all $n\in \N^+$, let $\FF_n\subseteq \mathcal{P}(\XX^n)$ be a given family of distributions; we collect the sets $\FF_n$ in the sequence $\FF=(\FF_n)_{n\in\N^+}$. The new task is testing whether $X^n$ is distributed according to $Q_n$, for some $Q_n\in \FF_n$, or according to the i.i.d.\ distribution $P^{\times n}$. 
The error probabilities are 
\begin{align}
\alpha(\mathcal{A}_n)&\coloneqq \sup_{Q_n\in \FF_n} \mathbb{P}_{X^n\sim Q_n}\left(\mathcal{A}_n(X_1,\dots,X_n)=H_1\right)&\text{type-I error}, \label{type_I_err_pr_composite} \\
\beta(\mathcal{A}_n)&\coloneqq \mathbb{P}_{X^n\sim P^{\times n}}\left(\mathcal{A}_n(X_1,\dots,X_n)=H_0\right)&\text{type-II error}. \label{type_II_err_pr_composite} 
\end{align}
Minimising the type II error with the usual constraint on the type I error yields a decay rate for the former equal to
\bb\label{eq:stein_def}
    \mathrm{Stein}(\FF\|P)\coloneqq \lim_{\epsilon\to 0}\liminf_{n\to \infty}\frac{1}{n} D^\epsilon_H(\FF_n\|P^{\times n}),
\ee
where $P^{\times n}$ is the product distribution corresponding to $H_1$, and
\bb
D^\epsilon_H(\FF_n\|P^{\times n}) \coloneqq -\log\min\left\{\beta(\mathcal{A}_n)\,:\, \alpha(\mathcal{A}_n)\leq \epsilon\right\},
\label{D_H_composite}
\ee
with $\beta(\mathcal{A}_n)$ now given by~\eqref{type_II_err_pr_composite} (cf.~\eqref{D_H_iid}). Note that if $\FF_n$ is a compact
convex set for every $n$, von Neumann's minimax theorem~\cite{vN-minimax} allows us to write
\bb
D^\epsilon_H(\FF_n\|P^{\times n}) = \min_{Q_n\in \FF_n} D^\epsilon_H(Q_n\|P^{\times n})\, ,
\label{D_H_minimax}
\ee
where $D^\epsilon_H(Q_n\|P^{\times n})$ is given by a formula analogous to~\eqref{D_H_iid}, the only difference being that $Q^{\times n}$ is replaced by $Q_n$ in the definition of the type-I error probability~\eqref{type_I_err_iid}.

We can also consider the strong converse Stein exponent, given by a modified version of~\eqref{eq:stein_def} in which we allow the type I error to be arbitrarily close to $1$, instead of arbitrarily small, and we replace the limit inferior in $n$ with a limit superior:
\bb\label{eq:stein_sc_def}
    \mathrm{Stein}^\dag(\FF\|P)\coloneqq \lim_{\epsilon\to 1^-}\limsup_{n\to \infty}\frac{1}{n} D^\epsilon_H(\FF_n\|P^{\times n}) .
\ee
Clearly, in general $\mathrm{Stein}^\dag(\FF\|P) \geq \mathrm{Stein}(\FF\|P)$. In many interesting cases, however, equality holds; equivalently, $\lim_{n\to \infty}\frac{1}{n} D^\epsilon_H(\FF_n\|P^{\times n})$ exists for all $\e\in (0,1)$, and its value is independent of $\e$. This holds, for example, when both hypotheses are simple and i.i.d., in which case we have indeed~\cite{stein_unpublished, chernoff_1956}
\bb
\label{eq:strong_converse_property_iid}
D(Q\|P) = \mathrm{Stein}(Q \| P) = \mathrm{Stein}^\dagger(Q\|P ) = \lim_{n\to \infty}\frac{1}{n}D^\epsilon_H(Q^{\times n}\|P^{\times n})\qquad \forall\ \e\in (0,1)\, .
\ee

The operational interpretation of the umlaut information emerges in the composite hypothesis testing problem where the underlying alphabet is bipartite, i.e.\ of the form $\XX\times \YY$, the alternative hypothesis is i.i.d.\ with single-copy distribution $P \mapsto P_{XY}$, and the null hypothesis is composite and of the form $Q_n \mapsto P_X^{\times n} Q_{Y^n}$. Here, $P_X$ is the marginal of $P_{XY}$ to the $X$ variable, and $Q_{Y^n}$ denotes an arbitrary probability distribution on $\YY^n$. With the role of the two hypotheses interchanged, this problem previously appeared in~\cite{Polyanskiy13,hayashi_2016-1,tomamichel_2018}.

More explicitly, given a bipartite probability distribution $P_{XY}$ on $\XX\times\YY$, the random variables $X^n$ and $Y^n$ taking values in $\XX^n$ and $\YY^n$, respectively, are distributed according to one of the following hypotheses:
\begin{itemize}
    \item $H_0$: the probability of observing $X_1=x_1,\dots, X_n=x_n$ and $Y_1=y_1, \dots, Y_n=y_n$ is given by $P_{X}(x_1)\cdots P_{X}(x_n)Q_{Y^n}(y_1,\dots, y_n)$, where $P_X$ is the marginal of $P_{XY}$ on $\XX$, and $Q_{Y^n}$ could be any probability distribution in $\mathcal{P}(\YY^n)$;
    \item $H_1$: the probability of observing $X_1=x_1,\dots, X_n=x_n$ and $Y_1=y_1, \dots, Y_n=y_n$ is given by $P_{XY}(x_1,y_1)\cdots P_{XY}(x_n,y_n)$.
\end{itemize}

To phrase the above hypothesis testing task in the composite hypothesis testing framework described earlier, it suffices to define $\FF^{P_X} \coloneqq \big(\FF^{P_X}_n\big)_n$, with
\bb
\FF^{P_X}_n &\coloneqq \left\{P_X^{\times n} Q_{Y^n}:\ Q_{Y^n}\in \mathcal{P}(\mathcal{Y}^n)\right\}.
\ee
Note that $\FF^{P_X}_n$ is a compact convex set, hence we can apply~\eqref{D_H_minimax} and write
\bb
D^\epsilon_H\big(\FF_n^{P_X}\,\big\|\,P_{XY}^{\times n}\big) = \min_{Q_{Y^n}\in \mathcal{P}(\mathcal{Y^n})} D^\epsilon_H\big(P_X^{\times n} Q_{Y^n}\,\big\|\,P_{XY}^{\times n}\big)\, .
\label{D_H_minimax_product_testing}
\ee
The following theorem states that the corresponding Stein exponent coincides with the umlaut information between $X$ and $Y$, and that this equality holds also in the strong converse regime.

\begin{boxed}{}
\begin{thm}[(Operational interpretation of the umlaut information)] \label{thm:interp_L2}
Given a joint probability distribution $P_{XY}\in\mathcal{P}(\XX\times\YY)$, let $P_X$ be the marginal on $\XX$. Then it holds that
\bb
U(X;Y)=\lim_{n\to \infty}\frac{1}{n} D^\epsilon_H\big(\FF^{P_X}\,\big\|\, P_{XY}\big) \qquad \forall\ \e\in (0,1)\, ;
\ee
equivalently,
\bb
U(X;Y) = \mathrm{Stein}\big(\FF^{P_X}\,\big\|\, P_{XY}\big) = \mathrm{Stein}^\dag\big(\FF^{P_X}\,\big\|\, P_{XY}\big)\, .
\ee
\end{thm}
\end{boxed}

\begin{proof} 
Let $\alpha\in(0,1)$, and, as usual, let $P_{XY}^{\times n}\in\mathcal{P}(\XX^n\times\YY^n)$ be the i.i.d.\ distribution 
\bb
P^{\times n}_{XY}(x_1,\dots,x_n,y_1,\dots,y_n)=P_{XY}(x_1,y_1)\cdots P_{XY}(x_n,y_n)\, .
\ee
We denote as $P^{\times n}_X$ its marginal on $\XX^n$. A standard argument shows that (see, e.g., \cite{Hayashi_2007,Audenaert2012_quantum})
\bb\label{eq:Dalpha}
D^\epsilon_H(p\|q)\geq D_\alpha(p\|q)+\frac{\alpha}{1-\alpha}\log\frac{1}{\epsilon}.
\ee
We can use this inequality as follows:
\bb\label{proof_interp_L2_eq3}
\mathrm{Stein}\big(\FF^{P_X}\,\big\|\,P_{XY}\big) &\eqt{(i)} \lim_{\epsilon\to 0}\liminf_{n\to \infty}\frac{1}{n}\min_{Q_{Y^n}\in\mathcal{P}(\YY^n)}D^\epsilon_H(P_X^{\times n} Q_{Y^n}\|P_{XY}^{\times n})\\
&\geq \lim_{\epsilon\to 0}\liminf_{n\to \infty}\frac{1}{n}\min_{Q_{Y^n}\in\mathcal{P}(\YY^n)}\left(D_\alpha(P_X^{\times n} Q_{Y^n}\|P_{XY}^{\times n})+\frac{\alpha}{1-\alpha}\log\frac{1}{\epsilon}\right)\\
&=\liminf_{n\to \infty}\frac{1}{n}U_\alpha(X^n;Y^n)\\
&\eqt{(ii)} U_\alpha(X;Y),
\ee
where we used~\eqref{D_H_minimax_product_testing} in~(i), and the additivity of $U_\alpha$ (as in Theorem~\ref{thm:closed_alpha}) in~(ii). In particular,
\bb
\mathrm{Stein}\big(\FF^{P_X}\,\big\|\,P_{XY}\big) &\geq \lim_{\alpha\to 1^-} U_\alpha(X;Y) \\
&\eqt{(iii)} U(X;Y),
\ee
where in~(iii) we employed Lemma~\ref{lem:alpha_to_1}. For the upper bound we consider the ansatz 
\bb
Q_{Y^n}(y_1,\dots,y_n) = Q_Y^{\times n}(y_1,\ldots,y_n) = Q_Y(y_1)\cdots Q_Y(y_1),
\ee 
where $Q_Y\in\mathcal{P}(\YY)$ is an arbitrary fixed distribution. Then, 
\bb\label{eq:65}
\mathrm{Stein}^\dag\big(\FF^{P_X}\,\big\|\,P_{XY}\big) &\leq \lim_{\epsilon\to 1^-}\limsup_{n\to \infty}\frac{1}{n}D^\epsilon_H(P_X^nQ_Y^n\|P_{XY}^n)\\
&\eqt{(iv)}D(P_XQ_Y\|P_{XY})
\ee
where (iv)~follows from the strong converse of the Chernoff--Stein lemma~\cite{stein_unpublished, chernoff_1956}, here reported as~\eqref{eq:strong_converse_property_iid}. Minimising over $Q_Y\in\mathcal{P}(\YY)$ yields
\bb
\mathrm{Stein}^\dag\big(\FF^{P_X}\,\big\|\,P_{XY}\big) \leq \min_{Q_Y} D(P_XQ_Y\|P_{XY}) = U(X;Y)\, ,
\ee
concluding the proof.
\end{proof}

\subsection{Example: Joint Gaussian distributions}

The definition of the umlaut information (Definition~\ref{def:lautum}) also applies to continuous variables (see Appendix~\ref{app:cv}), and an instructive example is as follows.

\begin{prop}[(Umlaut information of joint Gaussian distributions)]
\label{prop:Gaussian-umlaut}
Let $x,m_X\in\XX=\mathbb{R}^{n}$ and $y,m_Y\in\YY=\mathbb{R}^{k}$; let $V\in\mathbb{R}^{(n+k)\times (n+k)}$ such that $V>0$. Let us introduce $r=(x,y)$, $m=(m_X,m_Y)$ and let us rewrite
\bb
    V=\begin{pmatrix}V_{XX} & V_{XY} \\ V_{XY}^\intercal & V_{YY}\end{pmatrix},
\ee
where $V_{XX}\in \mathbb{R}^{n\times n}$. Let $(X,Y)$ be a pair of random variables taking values in $\XX\times\YY$ with Gaussian probability distribution
\bb
    P_{XY}(x,y)=\frac{e^{-\frac{1}{2}(r-m)^\intercal V^{-1}(r-m)}}{(2\pi)^{(n+k)/2}\sqrt{\det V}}.
\ee
Then the umlaut-marginal of $P_{XY}$ on $\YY$ is given by
\bb
    \uml{P}_Y(y) = \frac{1}{(2\pi)^{k/2}\sqrt{\det (V/V_{XX})}}\,e^{-\frac{1}{2}(y-m_Y)^\intercal (V/V_{XX})^{-1}(y-m_Y)},
\ee
where
\bb
    V/V_{XX} &\coloneqq V_{YY}-ZV_{XX}^{-1}Z^\intercal
\ee
and the umlaut information between $X$ and $Y$ is
\bb
    U(X;Y)=\frac12 \log \frac{\det V}{\det \left(V_{XX}\oplus (V/V_{XX})\right)} + \frac12 \Tr \left[ V^{-1}\left(V_{XX}\oplus (V/V_{XX})\right)\right]-\frac{n+k}{2}.
\ee
\end{prop}

It is interesting to notice that, for Gaussian distributions, the marginal of $P_{XY}$ on $\YY$ is a Gaussian distribution $P_Y=\mathcal{G}(m_B,V_{YY})$ having as a covariance matrix the reduced covariance matrix $V_{XX}$. The umlaut-marginal $\uml{P}_{Y}=\mathcal{G}(m_B,(V/V_{XX})^{-1})$ is a Gaussian distribution too.

\begin{proof}[Proof of Proposition~\ref{prop:Gaussian-umlaut}]
Let us start by inverting $V$ using the Schur complement:
\bb
    V^{-1}=\begin{pmatrix}(V/V_{YY})^{-1} & Z \\ Z^\intercal & (V/V_{XX})^{-1}\end{pmatrix},
\ee
where $Z$ is a matrix that we do not need to specify and
\bb
    V/V_{YY} &\coloneqq V_{XX}-ZV_{YY}^{-1}Z^\intercal,\\
    V/V_{XX} &\coloneqq V_{YY}-ZV_{XX}^{-1}Z^\intercal.
\ee
It is known that the marginal of $P_{XY}$ on $\XX$ is
\bb
    P_{X}(x)=\frac{e^{-\frac{1}{2}(x-m_X)^\intercal V_{XX}^{-1}(x-m_X)}}{(2\pi)^{n/2}\sqrt{\det V_{XX}}}.
\ee
Now we can compute
\bb
        \uml{P}_Y(y)&\propto \exp\left(\sum_{x\in\XX}P_X(x)\log P_{XY}(x,y)\right)\\
        &\propto \exp\left(-\sum_{x\in\XX}\frac{1}{2}(r-m)^\intercal V^{-1}(r-m)\frac{e^{-\frac{1}{2}(x-m_X)^\intercal V_{XX}^{-1}(x-m_X)}}{(2\pi)^{n/2}\sqrt{\det V_{XX}}}\right)\\
    \ee
In particular,
\bb
    \sum_{x\in\XX}\frac{1}{2}(r-m)^\intercal &V^{-1}(r-m)\frac{e^{-\frac{1}{2}(x-m_X)^\intercal V_{XX}^{-1}(x-m_X)}}{(2\pi)^{n/2}\sqrt{\det V_{XX}}}\\
    &=\frac{1}{2}\Tr\left[(V/V_{YY})^{-1}\,\mathbb{E}_{P_X}\left[(x-m_X)(x-m_X)^\intercal\right]\right]\\
    &\qquad +\Tr\left[Z^\intercal \mathbb{E}_{P_X}\left[(x-m_X)\right](y-m_Y)^\intercal\right]\\
    &\qquad +\frac{1}{2}\Tr\left[(V/V_{XX})^{-1}\,(y-m_Y)(y-m_Y)^\intercal\right]\\
    &=\frac{1}{2}\Tr\left[(V/V_{YY})^{-1}\,V_{XX}\right]+\frac{1}{2}(y-m_Y)^\intercal (V/V_{XX})^{-1} (y-m_Y),
\ee
whence
\bb
    \uml{P}_Y(y) = \frac{1}{(2\pi)^{k/2} \sqrt{\det (V/V_{XX})}}\,e^{-\frac{1}{2}(y-m_Y)^\intercal (V/V_{XX})^{-1}(y-m_Y)}.
\ee
The umlaut information is therefore given by
\bb
    U(X;Y)=D(P_X\uml{P}_Y\|P_{XY})
\ee
according to Proposition~\ref{thm:closed_lautum}. We notice that
\begin{itemize}
    \item $P_X\uml{P}_Y$ has mean $m$ and covariance $V_{XX}\oplus (V/V_{XX})$,
    \item $P_{XY}$ has mean $m$ and covariance $V$,
\end{itemize}
whence
\bb
    U(X;Y)=\frac12 \log \frac{\det V}{\det \left(V_{XX}\oplus (V/V_{XX})\right)} + \frac12 \Tr \left[ V^{-1}\left(V_{XX}\oplus (V/V_{XX})\right)\right]-\frac{n+k}{2},
\ee
and this concludes the proof.
\end{proof}


\section{Channel umlaut information}\label{sec:channels}

\subsection{Definition and basic properties}

\begin{figure}[h]
    \centering
    \includegraphics[width=0.3\linewidth]{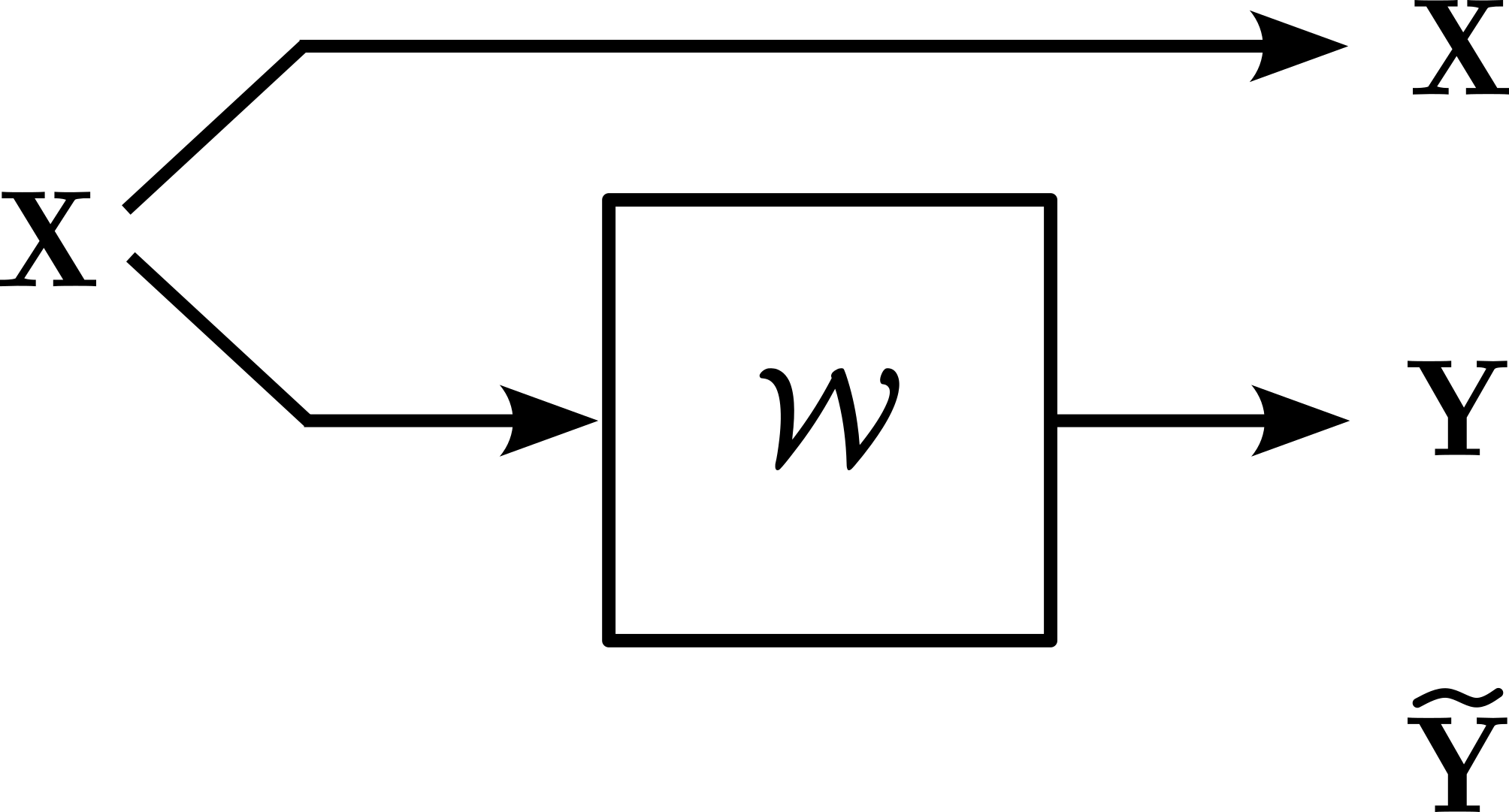}
    \caption{$X$ is the random variable representing a random input in the channel $\pazocal{W}$; $Y$ is the corresponding output, which is correlated to $X$; $\tilde Y$ is a random variable taking values in $\YY$ which is independent of $X$.}
    \label{fig:XYY}
\end{figure}

Let us introduce the notion of umlaut information of a discrete memoryless channel $\pazocal{W}$ from $\XX$ to $\YY$, i.e. a stochastic matrix $[\pazocal{W}(y|x)]_{x\in\XX,y\in\YY}$. Let us consider, as the input of the channel, a random variable $X$ taking values in $\XX$ and with distribution $P_X$. The corresponding output will therefore be a random variable $Y$ taking values in $\YY$ with distribution
\bb
    P_{Y}(y)=\sum_{x\in\XX}\pazocal{W}(y|x)P_{X}(x).
\ee
Let $P_{XY}$ be the joint distribution of $X$ and $Y$, namely,
\bb\label{eq:inp_outp}
    P_{XY}(x,y)=\pazocal{W}(y|x)P_{X}(x)\, .
\ee
Now we are interested in studying the relative entropy between the pair $(X,\tilde Y)$ and the pair $(X,Y)$, where $\tilde Y$ is a random variable (independent of $X$)  taking values in $\YY$ that minimises such divergence (see Figure~\ref{fig:XYY}). Then we maximise the result among the possible probability distributions of the input, as in the following definition.

\begin{Def}[(Umlaut information of a channel)]
\label{def:channel_lautum}
   Let $\pazocal{W}$ be a discrete memoryless channel $\pazocal{W}$ from $\XX$ to $\YY$ and let $X$ and $Y$ be random variables taking values in $\XX$ and $\YY$ with joint distribution
    \bb\label{eq:joint_prob_ch}
        P_{XY}(x,y)=\pazocal{W}(y|x)P_{X}(x).
    \ee
    Then we define the umlaut information of the channel $\pazocal{W}$ as 
\bb\label{eq:lautum_ch}
    U(\pazocal{W}) &\coloneqq \max_{P_X}U(X;Y)=    \max_{P_X}\min_{Q_Y}D(P_{X}Q_Y\|P_{XY}).
    \ee
\end{Def}

In the case where $P_X$ is restricted to be the uniform probability distribution, this quantity was considered for the classical-quantum setting in~\cite[Eq.~(5.133)]{Nuradha2024}, where it was called `oveloh information'. The maximisation on $P_X$, however, will prove important in order to obtain operational interpretations of $U(\pazocal{W})$.

\begin{prop}
    The functional $(P_X, Q_Y, \pazocal{W}) \mapsto D(P_X Q_Y\| P_{XY})$ where $P_{XY}(x,y) = P_X(x) \pazocal{W}(y|x)$ is linear in $P_X$ and jointly convex in $Q_Y$ and $\pazocal{W}$. Moreover,  $P_X \mapsto \min_{Q_Y} D(P_X Q_Y\| P_{XY})$ is concave.
\end{prop}

\begin{proof}
    Joint convexity in $Q_Y$ and $\pazocal{W}$ follows from joint convexity of the relative entropy. Linearity in $P_X$ is evident once one
    rewrites
    $D(P_X Q_Y\| P_{XY}) = \sum_{x \in \XX} P_X(x)\, D(Q_Y \| \pazocal{W}(\cdot|x) )$. Concavity of the second functional follows from the fact that a pointwise minimum of linear functions is concave.
\end{proof}
As a consequence of this, we can exchange the maximisation and minimisation using Sion's minimax theorem to find
\begin{align}
    U(\pazocal{W}) = \min_{Q_Y} \max_{x \in \XX} D( Q_Y \| \pazocal{W}(\cdot|x) ) \,. \label{eq:minimax2}
\end{align}

\begin{prop}
    The umlaut information of the channel $\pazocal{W}$ satisfies
    \begin{itemize}
        \item[(1)] Positive definiteness: $U(\pazocal{W}) \geq 0$ with equality if and only if $\pazocal{W}(\cdot|x)$ is independent of $x \in \XX$. 
        \item[(2)] Boundedness: $U(\pazocal{W}) < \infty \iff \exists y \in \YY$ such that $\forall x \in \XX: \pazocal{W}(y|x) > 0$, i.e., there exists an output symbol that cannot exclude any input symbol.
        \item[(3)] Convexity: The map $\pazocal{W} \mapsto U(\pazocal{W})$ is convex.
    \end{itemize}
\end{prop}

\begin{proof}
    For (1), the implication ``$\impliedby$'' follows by choosing $Q_Y = \pazocal{W}(\cdot|x)$. On the other hand, if $\pazocal{W}(\cdot|x)$ depends on $x \in\XX$ then the property follows form positive definiteness of the relative entropy.
    For (2), we note that this is essentially Property (2) of~Proposition~\ref{prop:lautum_prop}, except that now restricting to $P_X(x) > 0$ is no longer necessary as we are maximising over input distributions.

    For (3), we observe that the functional
    \begin{align}
        (Q_Y, \pazocal{W}) \mapsto \max_{x \in \XX} D( Q_Y \| \pazocal{W}(\cdot|x) )
    \end{align}
    in~\eqref{eq:minimax2} is a maximum of jointly convex function and thus inherits this property. Taking the minimum over $Q_Y$ then leaves $\pazocal{W} \mapsto U(\pazocal{W})$ convex.
\end{proof}

\subsection{An explicit formula and additivity}

\begin{conv}\label{conv:infty} According to the usual conventions $0\log 0 = 0$ and $\exp(-\infty)= 0$, we introduce the following notation. Let $\pazocal{W}$ be a discrete memoryless channel from $\XX$ to $\YY$ and let $P_X\in\mathcal{P}(\XX)$. Then we define
\bb\label{eq:notation}
\sum_{y\in\YY}\exp\left(\sum_{x\in\XX}P_X(x)\log\pazocal{W}(y|x)\right)\coloneqq \sum_{y\in\YY_{\pazocal{W},P_X}}\exp\left(\sum_{x\in{\rm supp}(P_X)}P_X(x)\log\pazocal{W}(y|x)\right),
\ee
where ${\rm supp}(P_X)\coloneqq\{x\in\XX:\, P_X(x)>0\}$ and $\YY_{\pazocal{W},P_X}\coloneqq \{y\in\YY:\ \pazocal{W}(y|x)>0\ \, \forall\, x\in {\rm supp}(P_X)\}$. We will always interpret the left-hand side of~\eqref{eq:notation} as specified on the right-hand side.
\end{conv}

\begin{prop}\label{prop:lautum_channel_formula} Let $\pazocal{W}$ be a discrete memoryless channel from $\XX$ to $\YY$. Then, its umlaut information can be written as
\bb\label{eq:looks_additive}
U(\pazocal{W})=-\log\min_{P_X}\sum_{y\in\YY}\exp\left(\sum_{x\in\XX}P_X(x)\log\pazocal{W}(y|x)\right)
\ee
according to Convention~\ref{conv:infty}, or, equivalently,
\bb\label{eq:alternative_W_0}
\exp\left[-U(\pazocal{W})\right] &=\min_{P_X}\sum_{y\in\YY}\prod_{x\in\XX}\pazocal{W}(y|x)^{P_X(x)},
\ee
where we set $0^0= 1$. It also holds that
\bb
U(\pazocal{W})=\min_{Q_Y}\max_{x\in\XX}D\left(Q_Y\middle\|\pazocal{W}(\,\cdot\,|x)\right) = \min_{Q_Y}\max_{x\in\XX} \sum_{y\in \YY} Q_Y(y) \log \frac{Q_Y(y)}{\pazocal{W}(y|x)}.
\ee
\end{prop}

Before proving the proposition, we 
recall a minimax lemma.

\begin{lemma}[{\cite[Theorem~5.2]{Farkas2006}}]\label{lem:minmax} Let $A$ be a compact, convex subset of a Hausdorff topological vector space $U$, and let $B$ be a convex subset of the linear space $V$. Let $f: A \times B \to (-\infty, +\infty]$ be lower semicontinuous on $A$ for all fixed $y\in B$, concave in the first variable, and convex in the second. Then
\bb
    \sup_{x\in A}\inf_{y\in B} f(x,y)=\inf_{y\in B}\sup_{x\in A} f(x,y).
\ee
\end{lemma}

\begin{proof}[Proof of Proposition~\ref{prop:lautum_channel_formula}]
By plugging the joint probability~\eqref{eq:joint_prob_ch} of the pair input-output for the channel $\pazocal{W}$ in the Definition~\ref{def:channel_lautum} of channel umlaut information, and by means of the closed formula of Theorem~\ref{thm:closed_lautum}, we have
\bb\label{eq:conto_sopra}
    U(\pazocal{W})&=\max_{P_X}\left(-H(P_X)-\log\sum_{y\in\YY}\exp\left(\sum_{x\in\XX}P_X(x)\log\left(P_X(x)\pazocal{W}(y|x)\right)\right)\right)\\
    &=\max_{P_X}\left(-H(P_X)-\log\sum_{y\in\YY}\exp\left(-H(P_X)+\sum_{x\in\XX}P_X(x)\log\pazocal{W}(y|x)\right)\right)\\
    &=-\log\min_{P_X}\sum_{y\in\YY}\exp\left(\sum_{x\in\XX}P_X(x)\log\pazocal{W}(y|x)\right),
\ee
and~\eqref{eq:alternative_W_0} can be obtained by exponentiating~\eqref{eq:looks_additive}. Without using the closed formula for the umlaut information, let us now directly compute
\bb
    U(\pazocal{W})&=\max_{P_X}\min_{Q_Y}D(P_{X}Q_Y\|P_{XY})\\
    &=\max_{P_X}\min_{Q_Y}\left(-H(Q_Y)-\sum_{\substack{x\in\XX\\y\in\YY}}P_X(x)Q_Y(y)\log\pazocal{W}(y|x)\right).
\ee

The expression to be optimised is linear in $P_X$ and convex in $Q_Y$, since it is the sum of two convex functions. Furthermore, $Q_Y \mapsto D(P_{X}Q_Y\|P_{XY})$ is lower semicontinuous \cite[Theorem~15]{vanErven2014}. We can therefore apply the minimax result in Lemma~\ref{lem:minmax} to get
\bb
    U(\pazocal{W})
    &=\min_{Q_Y}\max_{P_X}\sum_{\substack{x\in\XX\\y\in\YY}}P_X(x)Q_Y(y)\log\frac{Q_Y(y)}{\pazocal{W}(y|x)}\\
    &=\min_{Q_Y}\max_{P_X}\sum_{x\in\XX}P_X(x)D\left(Q_Y\middle\|\pazocal{W}(\,\cdot\,|x)\right)\\
    &=\min_{Q_Y}\max_{x\in\XX}D\left(Q_Y\middle\|\pazocal{W}(\,\cdot\,|x)\right),
\ee   
which concludes the proof.
\end{proof}

\begin{cor}[(Additivity of the channel umlaut information)]\label{cor:add_chan}
Let $\pazocal{W}_1$ be a channel from $\XX_1$ to $\YY_1$ and let $\pazocal{W}_2$ be a channel from $\XX_2$ to $\YY_2$. Let $\pazocal{W}_1 \times \pazocal{W}_2$  be the product channel, defined as
    \bb
        \big(\pazocal{W}_1 \times \pazocal{W}_2)(y_1,y_2|x_1,x_2)\coloneqq \pazocal{W}_{1}(y_1|x_1)\pazocal{W}_{2}(y_2|x_2)
    \ee
     for any $x_1\in\XX_1$, $x_2\in\XX_2$, $y_1\in\YY_1$ and $y_2\in\YY_2$.
    Then, we have
    \bb
        U(\pazocal{W}_1 \times \pazocal{W}_2) = U(\pazocal{W}_{1})+U(\pazocal{W}_{2}).
    \ee
\end{cor}

\begin{proof}
    By~\eqref{eq:looks_additive}, we immediately see that
    \bb
        U(\pazocal{W}_1 \times \pazocal{W}_2)&=-\min_{P_{X_1X_2}}\log\sum_{\substack{y_1\in\YY_1\\ y_2\in\YY_2}}\exp\left(\sum_{\substack{x_1\in\XX_1\\ x_2\in\XX_2}}P_{X_1X_2}(x_1,x_2)\left(\log\pazocal{W}_1(y_1|x_1)+\log\pazocal{W}_2(y_2|x_2)\right)\right)\\
        &=-\min_{P_{X_1X_2}}\log\sum_{\substack{y_1\in\YY_1\\ y_2\in\YY_2}}\exp\left(\sum_{x_1\in\XX_1}P_{X_1}(x_1)\log\pazocal{W}_1(y_1|x_1)\right)\exp\left(\sum_{x_2\in\XX_2}P_{X_2}(x_2)\log\pazocal{W}_2(y_2|x_2)\right)\\
        &\eqt{(i)}-\min_{P_{X_1}}\log\sum_{y_1\in\YY_1}\exp\left(\sum_{x_1\in\XX_1}P_{X_1}(x_1)\log\pazocal{W}_1(y_1|x_1)\right)\\
        &\qquad -\min_{P_{X_2}}\log\sum_{y_2\in\YY_2}\exp\left(\sum_{x_2\in\XX_2}P_{X_2}(x_2)\log\pazocal{W}_2(y_2|x_2)\right)\\
        &=U(\pazocal{W}_1)+U(\pazocal{W}_2)
    \ee
    where in (i) we have noticed that, since each of the two terms depends only on the respective marginal, we can reduce the minimisation over the joint probability distribution $P_{X_1X_2}$ to two independent minimisations over the marginals. This concludes the proof. 
\end{proof}

\subsection{Rényi $\alpha$-umlaut information of a channel}

It is also natural to extend the definition of the R\'enyi $\alpha$-umlaut information to channels.
\begin{Def}[(R\'enyi $\alpha$-umlaut information of a channel)] \label{def:renyi-umlaut_channel}
Let $\alpha\in(0,1)\cup (1,\infty)$, let $\pazocal{W}$ be a discrete memoryless channel $\pazocal{W}$ from $\XX$ to $\YY$, and let $X$ and $Y$ be random variables taking values in $\XX$ and $\YY$ with joint distribution $P_{XY}(x,y)=\pazocal{W}(y|x)P_{X}(x)$.
    We define the R\'enyi $\alpha$-umlaut information of the channel $\pazocal{W}$ as 
\bb\label{eq:renyi_umlaut_ch}
    U_\alpha(\pazocal{W}) &\coloneqq \max_{P_X}U_\alpha(X;Y)=    \max_{P_X}\min_{Q_Y}D_\alpha(P_{X}Q_Y\|P_{XY}).
    \ee
\end{Def}

\begin{prop}\label{prop:variational_Renyi} The minimax variational form for the channel umlaut information in~\eqref{eq:minimax2} can be generalised to the family of R\'enyi $\alpha$-umlaut information of channels for $\alpha\in(0,1)\cup (1,\infty)$. More precisely, let $\pazocal{W}$ be a discrete memoryless channel from $\XX$ to $\YY$. Then
\begin{align}
    U_\alpha(\pazocal{W}) = \min_{Q_Y} \max_{x \in \XX} D_\alpha( Q_Y \| \pazocal{W}(\,\cdot\,|x) ) \,. 
\end{align}
\end{prop}

\begin{proof} We start by explicitly writing
    \bb
        U_\alpha(\pazocal{W})&=\max_{P_X}\min_{Q_Y}\frac{1}{\alpha-1}\log\sum_{\substack{x\in\XX\\ y\in\YY}}P_{X}^\alpha(x)Q_Y^\alpha(y)\pazocal{W}^{1-\alpha}(y|x)P_X^{1-\alpha}(x)\\
        &=\max_{P_X}\min_{Q_Y}\frac{1}{\alpha-1}\log\sum_{\substack{x\in\XX\\ y\in\YY}}P_{X}(x)Q_Y^\alpha(y)\pazocal{W}^{1-\alpha}(y|x).
    \ee
    Let us first consider the case $\alpha\in(0,1)$. By monotonicity of the logarithm, we have
    \bb
        U_\alpha(\pazocal{W})&=\frac{1}{\alpha-1}\log\min_{P_X}\max_{Q_Y}\sum_{\substack{x\in\XX\\ y\in\YY}}P_{X}(x)Q_Y^\alpha(y)\pazocal{W}^{1-\alpha}(y|x).
    \ee
    By concavity of $x\mapsto x^\alpha$, the map $Q_Y\mapsto \sum_{x,y}P_{X}(x)Q_Y^\alpha(y)\pazocal{W}^{1-\alpha}$ is concave, and by linearity in $P_X$ we can apply Sion's minimax theorem to switch maximum and minimum:
    \bb
        U_\alpha(\pazocal{W})&=\frac{1}{\alpha-1}\log\max_{Q_Y}\min_{P_X}\sum_{\substack{x\in\XX\\ y\in\YY}}P_{X}(x)Q_Y^\alpha(y)\pazocal{W}^{1-\alpha}(y|x)\\
        &=\frac{1}{\alpha-1}\log\max_{Q_Y}\min_{x\in\XX}\sum_{y\in\YY}Q_Y^\alpha(y)\pazocal{W}^{1-\alpha}(y|x)\\
        &=\min_{Q_Y}\max_{x\in\XX}\frac{1}{\alpha-1}\log\sum_{y\in\YY}Q_Y^\alpha(y)\pazocal{W}^{1-\alpha}(y|x),
    \ee
    and this proves the claim for $\alpha\in(0,1)$. Analogously, for $\alpha>1$ we have
    \bb
        U_\alpha(\pazocal{W})&=\frac{1}{\alpha-1}\log\max_{P_X}\min_{Q_Y}\sum_{\substack{x\in\XX\\ y\in\YY}}P_{X}(x)Q_Y^\alpha(y)\pazocal{W}^{1-\alpha}(y|x)\\
        &\eqt{(i)}\frac{1}{\alpha-1}\log\min_{Q_Y}\max_{P_X}\sum_{\substack{x\in\XX\\ y\in\YY}}P_{X}(x)Q_Y^\alpha(y)\pazocal{W}^{1-\alpha}(y|x)\\
        &=\frac{1}{\alpha-1}\log\min_{Q_Y}\max_{x\in\XX}\sum_{y\in\YY}Q_Y^\alpha(y)\pazocal{W}^{1-\alpha}(y|x)\\
        &=\min_{Q_Y}\max_{x\in\XX}\frac{1}{\alpha-1}\log\sum_{y\in\YY}Q_Y^\alpha(y)\pazocal{W}^{1-\alpha}(y|x),
    \ee
    where in (i) we noticed that $x\mapsto x^\alpha$ is convex. This concludes the proof.
\end{proof}

\begin{cor}\label{cor:add_Renyi_channel} For all $\alpha\in(0,1)\cup (1,\infty)$, the R\'enyi $\alpha$-umlaut information is additive under the tensor product of channels, namely
\bb
    U_\alpha(\pazocal{W}_1\times\pazocal{W}_2)= U_\alpha(\pazocal{W}_1)+U_\alpha(\pazocal{W}_2) 
\ee
\end{cor}

\begin{proof}
    Using the variational form established in Proposition~\ref{prop:variational_Renyi}, we have
    \bb
        U_\alpha(\pazocal{W}_1\times\pazocal{W}_2) &= \min_{Q_{Y_1Y_2}} \max_{\substack{x_1\in \XX_1\\x_2\in\XX_2}} D_\alpha( Q_{Y_1Y_2} \| \pazocal{W}_1(\cdot|x_1)\times \pazocal{W}_2(\cdot|x_2) )\\
        &\leqt{(i)} \min_{Q_{Y_1}\times Q_{Y_2}} \max_{\substack{x_1\in \XX_1\\x_2\in\XX_2}} D_\alpha( Q_{Y_1}\times Q_{Y_2} \| \pazocal{W}_1(\cdot|x_1)\times \pazocal{W}_2(\cdot|x_2) )\\
        &\eqt{(ii)}\min_{Q_{Y_1}} \max_{x_1\in\XX_1} D_\alpha( Q_{Y_1} \| \pazocal{W}_1(\cdot|x_1))+\min_{Q_{Y_2}} \max_{x_2\in\XX_2} D_\alpha( Q_{Y_2} \| \pazocal{W}_2(\cdot|x_2))\\
        &=U_\alpha(\pazocal{W}_1)+U_\alpha(\pazocal{W}_2)\,, 
    \ee
    where in (i) we have chosen the ansatz $Q_{Y_1Y_2}=Q_{Y_1}Q_{Y_2}$ and in (ii) we have leveraged the additivity of the  R\'enyi $\alpha$-relative entropy. Conversely, 
    \bb
    U_\alpha(\pazocal{W}_1\times\pazocal{W}_2) &=    \max_{P_{X_1X_2}}U_\alpha(X_1,X_2 \,;Y_1,Y_2)_{(\pazocal{W}_1)_{Y_1|X_1}(\pazocal{W}_2)_{Y_2|X_2}P_{X_1X_2}}\\
    &\geqt{(iii)}\max_{P_{X_1}\times P_{X_2}}U_\alpha(X_1,X_2 \,;Y_1,Y_2)_{(\pazocal{W}_1)_{Y_1|X_1}P_{X_1}\times (\pazocal{W}_2)_{Y_2|X_2}P_{X_2}}\\
    &\eqt{(vi)} \max_{P_{X_1}\times P_{X_2}}\left(U_\alpha(X_1 \,;Y_1)_{(\pazocal{W}_1)_{Y_1|X_1}P_{X_1}}+U_\alpha(X_2 \,;Y_2)_{(\pazocal{W}_2)_{Y_2|X_2}P_{X_2}}\right)\\
    &=U_\alpha(\pazocal{W}_1)+U_\alpha(\pazocal{W}_2),
    \ee
    where in (iii) we have chosen the ansatz $P_{X_1X_2}=P_{X_1}P_{X_2}$ and in (vi) we have used the additivity of the  R\'enyi $\alpha$-umlaut information for probability distributions (Proposition~\ref{thm:closed_alpha}). 
\end{proof}

Finally, as an immediate consequence of Lemma~\ref{lem:alpha_to_1} we can see that the $\alpha$-R\'enyi umlaut information of a channel converges to the umlaut information.
\begin{cor}\label{cor:aloha_to_1_channel}
For any discrete memoryless channel $\pazocal{W}$ it holds that

\bb
    \lim_{\alpha\to 1} U_\alpha(\pazocal{W})=U(\pazocal{W})\, .
\ee
\end{cor}

\begin{proof} 
Using Proposition~\ref{prop:variational_Renyi}, we have
   \bb
        \lim_{\alpha\to 1^+} U_\alpha(\pazocal{W})=\min_{Q_Y}\inf_{\alpha>1} \max_{x \in \XX} D_\alpha( Q_Y \| \pazocal{W}(\,\cdot\,|x) )\eqt{(i)}\min_{Q_Y} \max_{x \in \XX}\inf_{\alpha>1} D_\alpha( Q_Y \| \pazocal{W}(\,\cdot\,|x) )\eqt{(ii)} U(\pazocal{W}),
   \ee
   where in (i) we noticed that $\XX$ is finite and that, by the monotonicity in $\alpha$ of the R\'enyi divergences and by Lemma~\ref{lem:alpha_to_1}, we have
   \bb
   \inf_{\alpha>1}D_\alpha( Q_Y \| \pazocal{W}(\,\cdot\,|x) )=\lim_{\alpha\to 1^+}D_\alpha( Q_Y \| \pazocal{W}(\,\cdot\,|x) )=D( Q_Y \| \pazocal{W}(\,\cdot\,|x) ),
   \ee and in (ii) we used~\eqref{eq:minimax2}. Analogously, using Lemma~\ref{lem:alpha_to_1} we 
obtain the complementary claim, i.e.
\bb
\lim_{\alpha\to 1^-}    U_\alpha(\pazocal{W})=\sup_{\alpha<1}\max_{P_X}U_\alpha(X;Y)_{\pazocal{W}_{Y|X}P_X}=\max_{P_X}\sup_{\alpha<1}U_\alpha(X;Y)_{\pazocal{W}_{Y|X}P_X}=U(\pazocal{W}).
\ee
\end{proof}


\subsection{Zero-rate limit of the sphere-packing bound}

For a discrete memoryless channel $\pazocal{W}$ from $\XX$ to $\YY$ and $r>0$ the sphere-packing bound can be written as
\begin{align}
E_{\mathrm{sp}}(r,\pazocal{W}) & =\sup_{\alpha\in(0,1]}\max_{P_X}\min_{Q_Y}\;\frac{1-\alpha}{\alpha}\left(D_{\alpha}\left(P_{XY}\middle\|P_X\times Q_Y\right)-r\right) \\
&= \sup_{\alpha\in(0,1]} \Big(  U_{1-\alpha}(\pazocal{W}) - \frac{1-\alpha}{\alpha} r \Big), \label{eq:sp}
\end{align}
with $P_{XY}(x,y)=\pazocal{W}(y|x)P_X(x)$ and the Rényi divergences and R\'enyi umlaut information as featured in~\eqref{eq:Renyi} and Definition~\ref{def:renyi-umlaut}. While $E_{\mathrm{sp}}(r,\pazocal{W})$ for $r$ larger than some critical value has an operational interpretation in noisy channel coding first discussed in \cite{SHANNON196765}, we are here interested in computing the zero-rate limit $r\to0$, denoted by $E_{\mathrm{sp}}(0^+,\pazocal{W})$.

\begin{prop}\label{prop:sphere-packing-limit}
For a discrete memoryless channel $\pazocal{W}$ from $\XX$ to $\YY$, we have
\begin{align}
E_{\mathrm{sp}}(0^+,\pazocal{W})=U(\pazocal{W}). 
\end{align}
\end{prop}

\begin{proof}
From~\eqref{eq:sp} we immediately see that
$E_{\mathrm{sp}}(r,\pazocal{W})$ is antimonotone in $r$, so that we can write
\begin{align}
    E_{\mathrm{sp}}(0^+,\pazocal{W}) &= \sup_{r > 0} \sup_{\alpha\in(0,1]} \Big( U_{1-\alpha}(\pazocal{W}) - \frac{1-\alpha}{\alpha} r \Big) \\
    &= \sup_{\alpha\in(0,1]}  U_{1-\alpha}(\pazocal{W}) \,,
\end{align}
and the desired result then follows from Corollary~\ref{cor:aloha_to_1_channel}. 
\end{proof}

\subsection{Operational interpretation in non-signalling--assisted communication}\label{ns-assisted-interpretation}

\begin{figure}[h]
    \centering
    \includegraphics[width=0.45\linewidth]{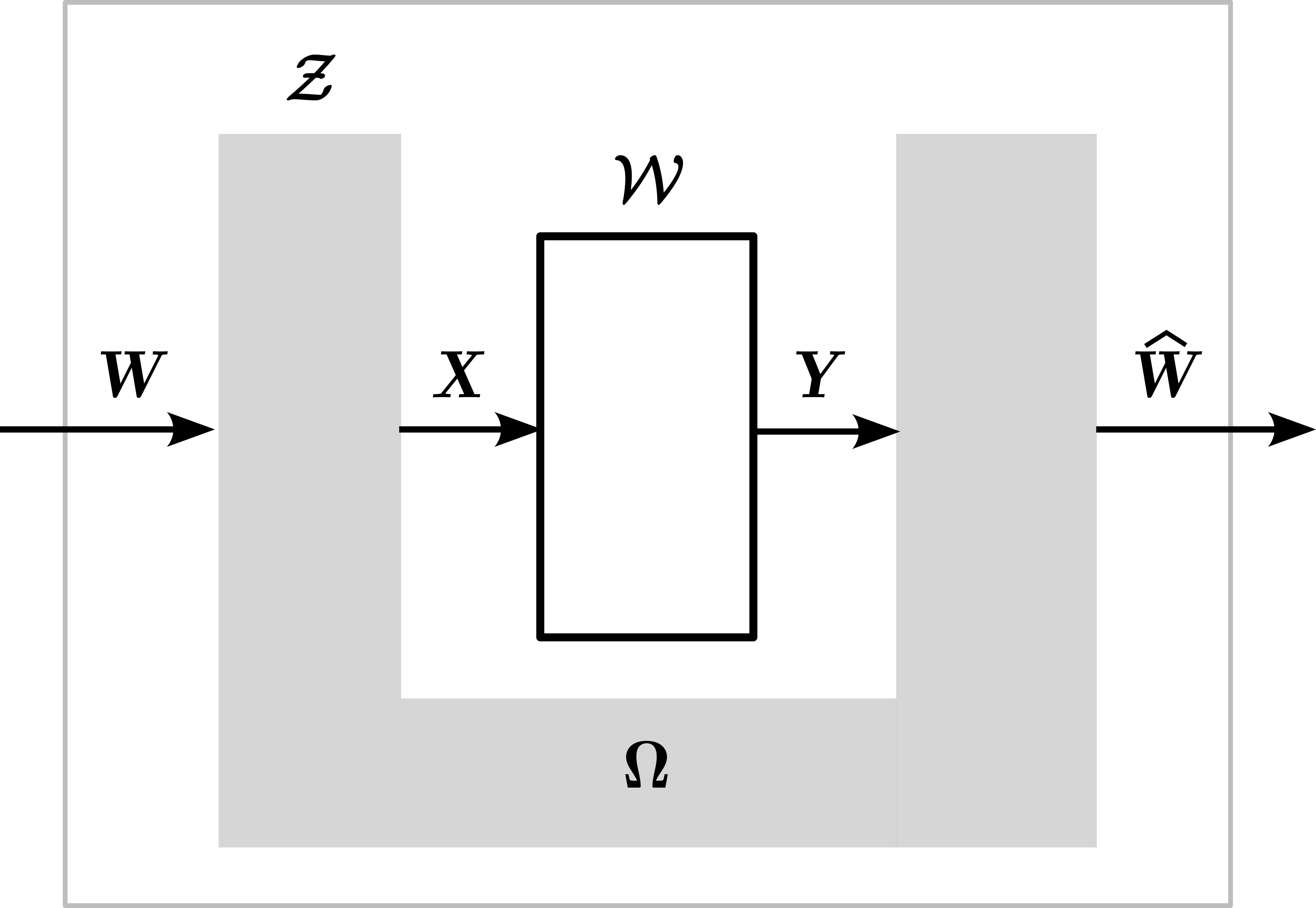}
    \caption{The composition of a $\Omega$-assisted code $\pazocal{Z}$ with a channel $\pazocal{W}$.}
    \label{fig:omega_assisted}
\end{figure}

Consider a source of messages represented by a random variable $W$ taking values in an alphabet $\mathcal{M}=\{1,\dots, M\}$ of size $M$. A code $\pazocal{Z}$ is defined by a joint probability distribution
\bb
    \pazocal{Z}(x,\hat w|y,w)=\mathbb{P}\left[X=x,\,\hat W=\hat w\,|\, Y=y,\,W=w\right],
\ee
where $\hat W$ is a random variable taking values in $\mathcal{M}$ interpreted as an estimator of the original message $W$, and $X$ and $Y$ are the input and the output of a channel $\pazocal{W}$ from $\XX$ to $\YY$, as in Figure~\ref{fig:omega_assisted}.

\begin{Def} A code $\pazocal{Z}$ is said to be $\Omega$-assisted, with $\Omega\in\{\emptyset, NS\}$, if
\begin{itemize}
    \item $\Omega=\emptyset$ (unassisted codes) $\pazocal{Z}(x,\hat w|y,w)=\pazocal{E}(x|w)\pazocal{D}(\hat w|y)$ for some conditional probability distributions $\pazocal{E}$ and $\pazocal{D}$: the code is made of an encoder $\pazocal{E}$ and a decoder $\pazocal{D}$ without correlations;
    \item $\Omega=NS$ (non-signalling--assisted codes) the estimator $\hat W$ is conditionally independent of the source $W$ and the input of the channel $X$ is conditionally independent of the output of the channel $Y$ as
    \bb
        \mathbb{P}\left[\hat W=\hat w\,|\, W=w, \, Y=y\right]&=\mathbb{P}\left[\hat W=\hat w\,|\, Y=y\right],\\
        \mathbb{P}\left[X= x\,|\, W=w, \, Y=y\right]&=\mathbb{P}\left[ X=x\,|\, W=w\right].
    \ee
\end{itemize}
\end{Def}

We refer to \cite{Matthews2012} for further interpretations of non-signalling--assisted codes for communication, originally studied in the setting of quantum non-locality \cite{Popescu94}. There might be intermediate settings, like the physically meaningful case of the entanglement-assisted codes. However, we will not discuss these here (cf.~Section~\ref{sec:outlook}), since the case of non-signalling assistance is the one providing an operati\"onal interpretation to the umlaut information.

\begin{Def}\label{def:assisted-codes}
Given a channel $\pazocal{W}$ and a source $W$ taking values in a set of messages of size $M$ with uniform probability, the minimum average error probability that can be achieved by a $\Omega$-assisted code is defined as
\bb
    \epsilon^\Omega(M, \pazocal{W})\coloneqq \min_{\pazocal{Z}\in\{\Omega\text{-assisted codes}\}}\left\{\mathbb{P}\left[ W\neq \hat W\right]\quad\text{ with }\quad Y|X \sim \pazocal{W},\quad X\hat{W}|YW\sim \pazocal{Z}\right\}
\ee
and the largest size $M$ of the set of messages that can be transmitted with error probability at most $\epsilon\in (0,1)$ using a $\Omega$-assisted code is
\bb
    M^\Omega(\epsilon,\pazocal{W}) \coloneqq \max_{\pazocal{Z}\in\{\Omega\text{-assisted codes}\}} \left\{M\,:\, \mathbb{P}[\hat W \neq W]\leq \epsilon\quad\text{ with }\quad Y|X \sim \pazocal{W},\quad X\hat{W}|YW\sim \pazocal{Z}\right\}.
\ee
\end{Def}

In particular, by the inclusion $\{\emptyset\text{-assisted codes}\}\subseteq \{NS\text{-assisted codes}\}$, it holds that
\bb\label{eq:NSvsPl}
    \epsilon^{\rm NS}(M, \pazocal{W})\leq \epsilon^\emptyset(M, \pazocal{W}),\qquad M^\emptyset(\epsilon,\pazocal{W})\leq M^{\rm NS}(\epsilon,\pazocal{W}).
\ee

The corresponding error exponents are then defined as follows.

\begin{Def}[($\Omega$-assisted zero-rate error exponent)]
Given a channel $\pazocal{W}$ from $\XX$ to $\YY$, its 
$\Omega$-assisted error exponent with communication rate $r$ is defined as
\bb\label{eq:defR}
E^\Omega(r,\pazocal{W})\coloneqq \liminf_{n\to\infty} - \frac{1}{n} \log \epsilon^\Omega(\, \exp({rn}) ,\pazocal{W}^{\times n})\, .
\ee
Furthermore, the $\Omega$-assisted zero-rate error exponent of $\pazocal{W}$ is defined as
    \bb\label{eq:def_thoroughness}
        \errom{\pazocal{W}}&\coloneqq \liminf_{r\to 0^+}E^\Omega(r,\pazocal{W})\, .
    \ee
\end{Def}

As an alternative quantifier for the $\Omega$-assisted zero-rate error exponent from~\eqref{eq:def_thoroughness} we also consider the quantity
\bb\label{eq:opom_def}
\opom{\pazocal{W}}\coloneqq\liminf_{M\to\infty}\liminf_{n\to \infty} -\frac{1}{n} \log \epsilon^\Omega(M,\pazocal{W}^{\times n}),
\ee
which is to be understood as the error exponent with constant, but arbitrarily large message size $M\to\infty$ (similar to what is considered in \cite[Equation (1)]{lami2024asymptotic}).
We will shortly see, in fact, that the asymptotic results do not depend on the message size $M$ and hence we do not need to take $M$ going to infinity: it will suffice to fix an arbitrary $M \geq 2$.
Note that for any fixed $M\geq 1$ and $r>0$,
\bb
     \epsilon^\Omega(\,\exp(rn) ,\pazocal{W}^{\times n})\geq  \epsilon^\Omega(\,M ,\pazocal{W}^{\times n})\qquad \forall\ n\geq \frac{\log M}{r}.
\ee
Therefore, it follows immediately from the definitions that
\bb
    \liminf_{n\to \infty} -\frac1n \log \epsilon^\Omega(\exp({rn}),\pazocal{W}^{\times n})\leq \liminf_{n\to \infty}  -\frac1n \epsilon^\Omega(\,M ,\pazocal{W}^{\times n}),
\ee
whence, by arbitrariness of $M\geq 1$ and $r>0$,
\bb\label{eq:Tast}
    \errom{\pazocal{W}}\leq \opom{\pazocal{W}}.
\ee
Further, it follows from \cite[Equation (1.55)]{SHANNON1967522} that in fact $E^{\emptyset}(0^+,\pazocal{W})=E^{\emptyset}_0(\pazocal{W})$, and we will prove as part of the forthcoming Theorem~\ref{thm:exact_tho} that also $E^{\mathrm{\rm NS}}(0^+,\pazocal{W})=E^{\mathrm{\rm NS}}_0(\pazocal{W})$. As such, $\errom{\pazocal{W}}$ and $\opom{\pazocal{W}}$ are identical quantifiers for the zero-rate error exponents considered in this work.

A one-shot upper bound on $M^\Omega_\epsilon$ was established by a seminal work of Polyanskiy, Poor and Verd\'u~\cite{PPV}, which was later shown to exactly correspond to NS-assisted codes~\cite{Matthews2012}. For our purposes this can be stated in terms of the error exponent as follows.

\begin{prop}[(Achievability of the meta-converse~\cite{Matthews2012})]\label{thm:rev_meta}
For any channel $\pazocal{W}$ from $\XX$ to $\YY$ and any $M \in \N$ we have
    \bb\label{eq:reversedmc1}
        -\log \epsilon^{\rm NS}(M,\pazocal{W})=\max_{P_X}\min_{Q_Y} D^{1/M}_H(P_XQ_Y\|P_{XY}),
    \ee
    where $P_{XY}(x,y)=\pazocal{W}(y|x)P_X(x)$.
    In particular,
    \bb\label{eq:reversedmc2}
        \opns{\pazocal{W}}&=\lim_{\delta\to 0}\liminf_{n\to \infty}\frac{1}{n}\max_{P_{X^n}}\min_{Q_{Y^n}} D^{\delta}_H(P_{X^n}Q_{Y^n}\|P_{X^nY^n})
    \ee
    where $P_{X^nY^n}(x_1,\dots,x_n,y_1,\dots,y_n)= P_{X^n}(x_1,\dots,x_n)\prod_{i=1}^n\pazocal{W}(y_i|x_i)$.
\end{prop}

This essentially follows from the considerations in \cite{Matthews2012} and the one-shot result has been previously stated as~\cite[Proposition B.1]{Siddharth18}. We give a self-contained proof in our notation in Appendix~\ref{app:missing}. The following theorem is our main result in this section.

\begin{boxed}{}
\begin{thm}[(Non-signalling--assisted zero-rate error exponent of $\pazocal{W}$ and umlaut information)]
\label{thm:exact_tho}
    Given a discrete, memoryless channel $\pazocal{W}$ from $\XX$ to $\YY$, it holds that
    \bb
         \opns{\pazocal{W}} = \errns{\pazocal{W}}=E_{\mathrm{sp}}(0^+,\pazocal{W})=U(\pazocal{W}).
    \ee
\end{thm}
\end{boxed}

\begin{proof}
The last equality $E_{\mathrm{sp}}(0^+,\pazocal{W})=U(\pazocal{W})$ was shown in Proposition~\ref{prop:sphere-packing-limit}. The second equality $\errns{\pazocal{W}}=E_{\mathrm{sp}}(0^+,\pazocal{W})$ can then also be deduced from taking the $r\to0$ limit of $E^{\mathrm{NS}}(r,\pazocal{W})=E_{\mathrm{sp}}(r,\pazocal{W})$, which is known from \cite[Theorem~24]{Polyanskiy13} (also see the discussion in \cite[Theorem~4.1]{Oufkir2024Oct}). We will show that $\opns{\pazocal{W}}$ and $\errns{\pazocal{W}}$ are both equal to $U(\pazocal{W})$.

In the rest of the proof, we will always refer to $P_{XY}$ (or similarly to $P_{X^nY^n}$) as the joint distribution of an input $X$ ($X^n$) and the corresponding output $Y$ ($Y^n$) using the channel $\pazocal{W}$ ($\pazocal{W}^{\times n}$), as in~\eqref{eq:inp_outp}. The proof is divided into a lower and upper bound.\\

\noindent\textbf{Lower bound.} Using Proposition~\ref{thm:rev_meta} and taking the particular ansatz $P_{X^n}=P^{\times n}_X$ defined as
    \bb
        P^{\times n}_X(x_1,\dots,x_n)=P_X(x_1)\cdots P_X(x_n)
    \ee
    for any $P_X\in\mathcal{P}(\XX)$, we get for any fixed $\alpha \in (0,1)$ that
    \bb
        \errns{\pazocal{W}}&=\lim_{r\to 0}\liminf_{n\to \infty}\frac{1}{n}\max_{P_{X^n}} \min_{Q_{Y^n}} D^{\exp(-rn)}_H(P_{X}^{\times n}Q_{Y^n}\|P_{X^nY^n})\\
        &\geq\lim_{r\to 0}\liminf_{n\to \infty}\frac{1}{n}\min_{Q_{Y^n}} D^{\exp(-rn)}_H(P_{X}^{\times n}Q_{Y^n}\|P_{XY}^n)\\
        &\geqt{(i)}\lim_{r\to 0}\liminf_{n\to \infty}\frac{1}{n}\left(\min_{Q_{Y^n}}  D_\alpha(P_{X}^{\times n}Q_{Y^n}\|P_{X^nY^n})+\frac{\alpha}{1-\alpha}\,rn\right)\\
        &\eqt{(ii)} \liminf_{n\to \infty}\frac{1}{n}U_\alpha(X^n;Y^n)\\
        &\eqt{(iii)}U_\alpha(X;Y)
    \ee
    where in~(i) we have used~\eqref{eq:Dalpha}, namely
    \bb
        D^\epsilon_H(P\|Q)\geq D_\alpha(P\|Q)+\frac{\alpha}{1-\alpha}\log\frac{1}{\epsilon}
    \ee
    for any $\alpha\in(0,1)$.
    In~(ii) we have introduced the umlaut information of $(X^n,Y^n)\sim P_{XY}^n$. In (iii) we have leveraged the additivity of $U_\alpha$, according to Theorem~\ref{thm:closed_alpha}. By arbitrariness of $P_X$ and $\alpha$,  we get
    \bb
        \errns{\pazocal{W}}&\geq \max_{P_X}\limsup_{\alpha<1}U_\alpha(X;Y)\, \eqt{(iv)}\, \max_{P_X}U(X;Y)=U(\pazocal{W})
    \ee
    where in (iv) we have used Lemma~\ref{lem:alpha_to_1}.\\

\noindent\textbf{Upper bound.} Now, we notice that
    \bb
        \opns{\pazocal{W}}&\leqt{(v)} \lim_{\delta\to 0}\liminf_{n\to \infty}\max_{P_{X^n}}\min_{Q_{Y^n}}\frac{1}{n}\frac{1}{1-\delta}\left(1+ D(P_{X^n}Q_{Y^n}\|P_{X^nY^n})\right)\\
        & =\liminf_{n\to\infty}\frac{1}{n}\max_{P_{X^n}}\min_{Q_{Y^n}}D(P_{X^n}Q_{Y^n}\|P_{X^nY^n})\\
        &=\liminf_{n\to\infty}\frac{1}{n}U(\pazocal{W}^{\times n})\\
        &\eqt{(vi)}U(\pazocal{W}),
    \ee
    where in~(v) we have again used the achievability of the meta-converse in Proposition~\ref{thm:rev_meta} as well as the upper bound (see e.g.~\cite[Equation (2.251)]{PolyanskiyPhD})
    \bb
        D_H^{\epsilon}(P\|Q)\leq \frac{1}{1-\epsilon}(1+D(P\|Q)),
    \ee
    and in~(vi) we have recalled that the channel umlaut information is additive (Corollary~\ref{cor:add_chan}).
Hence, we have found that
\bb
    U(\pazocal{W})\leq \errns{\pazocal{W}}\leqt{(vii)}\opns{\pazocal{W}}\leq U(\pazocal{W}),
\ee
where (vii) is~\eqref{eq:Tast}. This concludes the proof.
\end{proof}

A `strong converse' extension of the operational interpretation established in Theorem~\ref{thm:exact_tho} can be formulated as follows. For a fixed message size $M$, let
\bb
\opom{\pazocal{W},M}\coloneqq\liminf_{n\to \infty} -\frac{1}{n} \log \epsilon^\Omega(M,\pazocal{W}^{\times n}).
\ee
By definition of $\opom{\pazocal{W}}$ in~\eqref{eq:opom_def}, this means that
\bb
    \opom{\pazocal{W}}=\liminf_{M\to \infty}\opom{\pazocal{W},M}.
\ee
Then, we can prove the following statement.
\begin{thm}[(Strong converse for the non-signalling--assisted zero-rate error exponent)]\label{thm:strong_channels}
    Given a discrete, memoryless channel $\pazocal{W}$ from $\XX$ to $\YY$, for any fixed $M\geq 2$ it holds that
    \bb
         \opns{\pazocal{W},M} =U(\pazocal{W}).
    \ee
\end{thm}
\begin{proof}
    By the very definition of \opns{\pazocal{W},M}, we have that
    \bb
        \opns{\pazocal{W},M}\geq \opns{\pazocal{W}} \eqt{(i)} U(\pazocal{W}),
    \ee
    where in (i) we have used Theorem~\ref{thm:exact_tho}.
    Using the achievability of the meta-converse (Proposition~\ref{thm:rev_meta}), we 
    now see that
    \bb
        -\log \epsilon^{\rm NS}(M,\pazocal{W})=\max_{P_X}\min_{Q_Y} D^{1/M}_H(P_XQ_Y\|P_{XY}),
    \ee
    so that
    \bb
     \opns{\pazocal{W},M}=\liminf_{n\to \infty} \max_{P_{X^n}}\min_{Q_{Y^n}} \frac{1}{n} D^{1/M}_H(P_{X^n}Q_{Y^n}\|P_{X^nY^n}).
    \ee
    Hence
    \bb
        \opns{\pazocal{W},M}&\leqt{(ii)}\liminf_{n\to \infty} \max_{P_{X^n}}\min_{Q_{Y^n}} \frac{1}{n} \left(D_\alpha(P_{X^n}Q_{Y^n}\|P_{X^nY^n})+\frac{\alpha}{1-\alpha}\log\frac{1}{1-1/M}\right)\\
        &=\liminf_{n\to\infty}\frac{1}{n}U_\alpha(\pazocal{W}^{\times n})\\
        &\eqt{(iii)}U_\alpha(\pazocal{W}),
    \ee
    where in (ii) we have used the upper bound (see e.g.~\cite{mosonyi_2015})
    \bb
        D_H^{\epsilon}(p\|q)\leq D_\alpha(p\|q)+\frac{\alpha}{\alpha-1}\log\frac{1}{1-\epsilon},
    \ee
    which holds for $0<\epsilon<1$ and $\alpha\in(1,\infty)$,
    and in (iii) we used the additivity of the R\'enyi $\alpha$-umlaut information of a channel (Corollary~\ref{cor:add_Renyi_channel}). By arbitrariness of $\alpha>1$, we take the limit
    \bb
        \opns{\pazocal{W},M}&\leq 
        \lim_{\alpha\to 1^+} U_\alpha(\pazocal{W})\eqt{(iv)}U(\pazocal{W}),
    \ee
    where in (iv) we have used Corollary~\ref{cor:aloha_to_1_channel}. This concludes the proof.
\end{proof}


\subsection{Operational interpretation in unassisted communication}

The unassisted zero-rate error exponent for a discrete memoryless channel $\pazocal{W}$ is determined as~\cite[Theorem 3]{SHANNON196765}
\bb\label{eq:plain}
    \errpl{\pazocal{W}}=-\min_{P_X}\sum_{x,x'\in\XX}P_X(x)P_X(x')\log\sum_{y\in\YY}\sqrt{\pazocal{W}(y|x)\pazocal{W}(y|x')}.
\ee
where $P_X$ belongs to $\mathcal{P}(\XX)$ (also see~\cite{SHANNON1967522,Berlekamp1964BlockCW,Gallager1965} for further discussions). By the operational interpretation of  the unassisted and NS-assisted quantities, it follows that
\bb
    \errpl{\pazocal{W}}\leq \errns{\pazocal{W}}.
\ee
However, by comparing the two explicit expressions without having in mind their meaning, it is not immediately clear why one should be smaller than the other one. A direct proof will turn out to be insightful in order to provide an operational interpretation to the umlaut information in terms of list decoding (as discussed in the following).

\begin{prop}\label{prop:inequality}
For any $P_X\in\mathcal{P}(\XX)$, it holds that
    \bb
        -\sum_{x,x'\in\XX}P_X(x)P_X(x')\log\sum_{y\in\YY}\sqrt{\pazocal{W}(y|x)\pazocal{W}(y|x')}\leq -\log\sum_{y\in\YY}\exp\left(\sum_{x\in\XX}P_X(x)\log\pazocal{W}(y|x)\right).
    \ee
\end{prop}
\begin{proof}
    We will call, as usual, $P_{XY}(x,y)\coloneqq\pazocal{W}(y|x)P_X(x)$. By the Gibbs variational principle (Lemma~\ref{thm:Gibbs}), we reintroduce the minimization over $Q_Y\in\mathcal{P}(\XX)$:
    \bb\label{eq:first}
        -\log\sum_{y\in\YY}\exp\left(\sum_{x\in\XX}P_X(x)\log\pazocal{W}(y|x)\right)
        &=\min_{Q_Y}D(P_XQ_Y\|P_{XY})\\
        &=\min_{Q_Y}\sum_{\substack{x\in\XX\\y\in\YY}}P_X(x)Q_Y(y)\log\frac{P_{X}(x)Q_Y(y)}{P_X(x)\pazocal{W}(y|x)}\\
        &=\min_{Q_Y}\sum_{x\in\XX}P_X(x)D\big(Q_Y\|\pazocal{W}(\,\cdot\,|x)\big).
    \ee
    Now, we can fictitiously double the sum over $x\in\XX$
    \bb\label{eq:second}
        \min_{Q_Y}\sum_{x\in\XX}&P_X(x)D\big(Q_Y\|\pazocal{W}(\,\cdot\,|x)\big)\\
        &=\min_{Q_Y}\sum_{x,x'\in\XX}\frac{1}{2}\left(P_X(x)P_X(x')D\big(Q_Y\|\pazocal{W}(\,\cdot\,|x)\big)+P_X(x)P_X(x')D\big(Q_Y\|\pazocal{W}(\,\cdot\,|x')\big)\right)\\
        &= \min_{Q_Y}\sum_{x,x'\in\XX}P_X(x)P_X(x')\left(\frac{1}{2}D\big(Q_Y\|\pazocal{W}(\,\cdot\,|x)\big)+\frac{1}{2}D\big(Q_Y\|\pazocal{W}(\,\cdot\,|x')\big)\right)\\
        &\geq \sum_{x,x'\in\XX}P_X(x)P_X(x')\min_{Q_Y}\left(\frac{1}{2}D\big(Q_Y\|\pazocal{W}(\,\cdot\,|x)\big)+\frac{1}{2}D\big(Q_Y\|\pazocal{W}(\,\cdot\,|x')\big)\right)\\
    \ee
    The minimum can be explicitly computed by Gibbs variational principle:
    \bb\label{eq:third}
        \min_{Q_Y}\Bigg(\frac{1}{2}&D\big(Q_Y\|\pazocal{W}(\,\cdot\,|x)\big) + \frac{1}{2}D\big(Q_Y\|\pazocal{W}(\,\cdot\,|x')\big)\Bigg) \\
        &=\min_{Q_Y}\left(-H(Q_Y)-\sum_{y\in\YY}Q_Y(y)\left(\frac{1}{2}\log \pazocal{W}(y|x)+\frac{1}{2}\log\pazocal{W}(y|x')\right)\right)\\
        &=-\log\sum_{y\in\YY}\exp\left(\frac{1}{2}\log \pazocal{W}(y|x)+\frac{1}{2}\log\pazocal{W}(y|x')\right)\\
        &=-\log\sum_{y\in\YY}\sqrt{ \pazocal{W}(y|x)\pazocal{W}(y|x')}.
    \ee
    By concatenating~\eqref{eq:first} with~\eqref{eq:second} and by plugging~\eqref{eq:third} in the final expression, we have concluded the proof.
\end{proof}

We extend the notion of channel coding to list decoding in order to derive the second operational interpretation of the umlaut information. Given a set of messages $\mathcal{M}=\{1,\dots, M\}$ and a discrete memoryless channel $\pazocal{W}$ from $\XX$ to $\YY$, an  $L$-list unassisted code $\pazocal{Z}$ is given by the composition of an encoder $\pazocal{E}$ from $\mathcal{M}$ to $\XX$ with a decoder $\pazocal{D}$ from $\YY$ to $[\mathcal{M}]^L$, defined as the family of subsets of $\mathcal{M}$ with cardinality $L$. Calling $W$ the source of messages, i.e.~a random variable taking values in $\mathcal{M}$, which is assumed to be uniform, let $\hat W_L$ be the random output of the code taking values in $[\mathcal{M}]^L$. Then the transmission is considered to be successful if $W\in \hat W_L$.

The standard coding setting that we have discussed previously then corresponds to $L=1$ and the corresponding definitions can be suitably generalised to $L>1$ as follows. Given a channel $\pazocal{W}$ and a source $W$ taking values in a set of messages of size $M$ with uniform probability, the minimum average error probability that can be achieved by an $L$-list unassisted code is (borrowing notation from Definition~\ref{def:assisted-codes})
\bb
    \epsilon_L^{\emptyset}(M, \pazocal{W})\coloneqq \min_{\pazocal{Z}\in\{L\text{-list unassisted codes}\}}\left\{\mathbb{P}\left[ W\notin \hat W_L\right]\quad\text{ with }\quad Y|X \sim \pazocal{W},\quad X\hat{W}_L|YW\sim \pazocal{Z}\right\}. 
\ee
The $L$-list unassisted error exponent with communication rate $r$ is defined as
\bb
E_L^\emptyset(r,\pazocal{W})\coloneqq \liminf_{n\to\infty}-\frac{1}{n}\log\epsilon_L^{\emptyset}(L\exp(rn), \pazocal{W})\, .
\ee
Note that we require a message size of size $L\exp(rn)$ rather than simply $\exp(rn)$, as we want to interpret $r$ as a communication rate, and it is common to define rates as the logarithm of the message size divided by $L$ when considering list decoding schemes. See for example~\cite{SHANNON196765}. However, since we anyway keep the list size finite, this has no impact on the value of $E_L^\emptyset(r,\pazocal{W})$ (see also~\cite[footnote on p.~1]{bondaschi2021}).

Now, the $L$-list unassisted zero-rate error exponent of $\pazocal{W}$ is defined as
\bb\label{eq:zero-rate-list}
    \errL{\pazocal{W}}&\coloneqq \lim_{r\to 0^+} E_L^\emptyset(r,\pazocal{W}).
\ee
This expression in fact features a general, closed formula~\cite[Theorem]{Blinovsky2001Oct} (also see~\cite{bondaschi2021} for further discussions).

\begin{prop}[{\cite[Theorem on p.~279]{Blinovsky2001Oct}}]
\label{thm:L_thoroughness}
Given a discrete, memoryless channel $\pazocal{W}$ from $\XX$ to $\YY$, it holds that
    \bb
        \errL{\pazocal{W}} = \max_{P_X}\sum_{x_1,\dots,x_{L+1}}P_X(x_1)\cdots P_X(x_{L+1})\left(-\log \sum_{y\in\YY}\sqrt[L+1]{\pazocal{W}(y|x_1)\cdots\pazocal{W}(y|x_{L+1})}\right),
    \ee
    for any $L\geq 1$, where $P_X\in\mathcal{P}(\XX)$.
\end{prop}

The main result of the section is the following theorem, which shows that the umlaut information becomes equal to the zero-rate error exponent of list decoding in the large list limit.

\begin{boxed}{}
\begin{thm}[(Achievability of the umlaut information via $E^\emptyset_L(0^+,\pazocal{W})$)]
\label{thm:interpr_upper}
Given a discrete, memoryless channel $\pazocal{W}$ from $\XX$ to $\YY$, the umlaut information of $\pazocal{W}$ is an upper bound to its $L$-list unassisted zero-rate error exponent 
    \bb
        \errL{\pazocal{W}} \leq U(\pazocal{W}),
    \ee
    with equality in the large list limit:
   \bb
   \sup_{L\geq1}\errL{\pazocal{W}}=U(\pazocal{W}).
    \ee
\end{thm}
\end{boxed}
We remark that in our work, the size of the list is fixed when the asymptotic limit of a large number of uses of the channel is taken. In the literature, another setting was also considered: $L$ could be taken dependent on the block length $n$ as $L=\exp(\lambda n)$. This case is considered for instance in \cite[Theorem~3]{Merhav2017}, where it is proved that
\bb
    E_L^\emptyset(r,\pazocal{W})=E_{\rm sp}(r-\lambda,\pazocal{W}).
\ee
This result can be connected to the umlaut information of $\pazocal{W}$ as the error exponent converges to $U(\pazocal{W})$ if we consider the limit $r\to 0^+$ with the constraint $\lambda=\lambda_r<r$:
\bb
    E_L^\emptyset(0^+,\pazocal{W})=\lim_{r\to 0^+}E_{\rm sp}(r-\lambda_r,\pazocal{W})=E_{\rm sp}(0^+,\pazocal{W})=U(\pazocal{W}).
\ee
In our discussion, however, we will focus only on the fixed list size regime.
Random coding bounds on the list decoding error exponent in this setting have also been considered in~\cite{nakiboglu_2019}, from which connections with the sphere packing bound in the large list limit can also be deduced.

To work towards a proof of Theorem~\ref{thm:interpr_upper}, we first generalise the lower bounds from Proposition~\ref{prop:inequality} to the following quantities related to $E^\emptyset_L(0^+,\pazocal{W})$.

\begin{Def}\label{def:ell} Let $k\geq 1$ and let $q\in \mathbb{R}^k$ be any vector such that $\sum_{i=1}^kq_i=1$. Let $\pazocal{W}$ be a discrete memoryless channel from $\XX$ to $\YY$ and let $P_X\in\mathcal{P}(\XX)$. Then we define the $(k,q)$-\textit{lower umlaut information} of the channel $\pazocal{W}$ with input distribution $P_X$ as
\bb
    \ell_{k,q}(\pazocal{W},P_X)\coloneqq \sum_{x_1,\dots,x_k\in \XX}P_X(x_1)\cdots P_X(x_k)\left(-\log \sum_{y\in\YY}\exp\left(\sum_{i=1}^kq_i\log\pazocal{W}(y|x_i)\right)\right).
\ee
\end{Def}

\begin{prop}[(Lower bound to the umlaut information)]
\label{thm:inequality}
Let $k\geq 1$ and let $q\in \mathbb{R}^k$ be any vector such that $\sum_{i=1}^kq_i=1$. Let $\pazocal{W}$ be a discrete memoryless channel from $\XX$ to $\YY$ and let $P_X\in\mathcal{P}(\XX)$. Then, whenever we consider $P_{XY}(x,y)\coloneqq\pazocal{W}(y|x)P_X(x)$ as the joint distribution of $(X,Y)$, taking values in $\XX\times \YY$, we have
    \bb\label{eq:ineq_lautum}
        U(X;Y)\geq \ell_{k,q}(\pazocal{W},P_X),
    \ee
    and consequently
    \bb\label{eq:ineq_lautum_channel}
        U(\pazocal{W})\geq \max_{P_X}\ell_{k,q}(\pazocal{W},P_X).
    \ee
\end{prop}

\begin{proof} The strategy is very similar to the proof of Proposition~\ref{prop:inequality}. We estimate
{\allowdisplaybreaks
    \begin{align}
        U(X;Y)&=\min_{Q_Y}D(P_XQ_Y\|P_{XY})\nonumber\\
        &=\min_{Q_Y}\sum_{x\in\XX}P_X(x)D\big(Q_Y\|\pazocal{W}(\,\cdot\,|x)\big)\nonumber\\
        &=\min_{Q_Y}\sum_{x\in\XX}P_X(x)\sum_{i=1}^kq_iD\big(Q_Y\|\pazocal{W}(\,\cdot\,|x)\big)\nonumber\\
        &=\min_{Q_Y}\sum_{x_1,\dots,x_k\in\XX}P_X(x_1)\cdots P_X(x_k)\sum_{i=1}^kq_iD\big(Q_Y\|\pazocal{W}(\,\cdot\,|x_i)\big)\nonumber\\
        &\geq \sum_{x_1,\dots,x_k\in\XX}P_X(x_1)\cdots P_X(x_k)\min_{Q_Y}\sum_{i=1}^kq_iD\big(Q_Y\|\pazocal{W}(\,\cdot\,|x_i)\big)\\
        &=\sum_{x_1,\dots,x_k\in\XX}P_X(x_1)\cdots P_X(x_k)\min_{Q_Y}\left(-H(Q_Y)-\sum_{y\in\YY}Q_Y(y)\sum_{i=1}^kq_i\log \pazocal{W}(y|x_i)\right)\nonumber\\
        &\eqt{(i)}\sum_{x_1,\dots,x_k\in\XX}P_X(x_1)\cdots P_X(x_k)\left(-\log \sum_{y\in\YY}\exp\left(\sum_{i=1}^kq_i\log\pazocal{W}(y|x_i)\right)\right),\nonumber
    \end{align}
    }
    where in (i) we have used Gibbs variational principle (Lemma~\ref{thm:Gibbs}). This proves~\eqref{eq:ineq_lautum}. To get~\eqref{eq:ineq_lautum_channel}, we just take the maximum over $P_X\in\mathcal{P}(\XX)$.
\end{proof}

The main technical workhorse for the achievability proof of Theorem~\ref{thm:interpr_upper} is then given by the following proposition.

\begin{prop}[(Achievability of the umlaut information via $\ell_{k,q}$)]
\label{thm:achievability}
For any $k\geq 1$, let $\Delta_1^k$ be the set of vectors $q\in \mathbb{R}^k$ such that $\sum_{i=1}^kq_i=1$. Let $\pazocal{W}$ be a discrete memoryless channel from $\XX$ to $\YY$ and let $P_X\in\mathcal{P}(\XX)$. Then, whenever we consider $P_{XY}(x,y)\coloneqq\pazocal{W}(y|x)P_X(x)$ as the joint distribution of $(X,Y)$, taking values in $\XX\times \YY$, we have
    \bb
        U(X;Y)= \sup_{k\geq 1}\sup_{q\in\Delta_1^k}\ell_{k,q}(\pazocal{W},P_X),
    \ee
    and for the uniform vector $u_k=(1/k,\dots,1/k)$ that
    \bb\label{eq:achievability}
        U(X;Y)= \sup_{k\geq1}\ell_{k,u_k}(\pazocal{W},P_X)=\lim_{k\to \infty }\ell_{k,u_k}(\pazocal{W},P_X).
    \ee
\end{prop}

By inspection, Propositions~\ref{thm:inequality} and~\ref{thm:achievability} taken together immediately imply the sought-after Theorem~\ref{thm:interpr_upper}.

\begin{proof}[Proof of Theorem~\ref{thm:interpr_upper}]
    Let $u_{L+1}=(\frac{1}{L+1},\dots, \frac{1}{L+1})$ and $P_X\in\mathcal{P}(\XX)$. Then, by Definition~\ref{def:ell},
    \bb
        \ell_{L+1,u_{L+1}}(P_X,\pazocal{W})
        &=\sum_{x_1,\dots,x_{L+1}\in \XX}P_X(x_1)\cdots P_X(x_{L+1})\left(-\log \sum_{y\in\YY}\exp\left(\sum_{i=1}^{L+1}\frac{1}{L+1}\log\pazocal{W}(y|x_i)\right)\right)\\
        &=\sum_{x_1,\dots,x_{L+1}}P_X(x_1)\cdots P_X(x_{L+1})\left(-\log \sum_{y\in\YY}\sqrt[L+1]{\pazocal{W}(y|x_1)\cdots\pazocal{W}(y|x_{L+1})}\right).
    \ee
    Taking the maximum over $P_X\in\mathcal{P}(\XX)$ and by Propositions~\ref{thm:L_thoroughness} and~\ref{thm:inequality}, we have
    \bb
        \errL{\pazocal{W}} = \max_{P_X}\ell_{L+1,u_{L+1}}(P_X,\pazocal{W})\leq U(\pazocal{W}).
    \ee
    Taking the supremum over $L\geq1$ establishes via Proposition~\ref{thm:achievability} the equality in the large list limit.
\end{proof}

For the proof of the main technical Proposition~\ref{thm:achievability}, we need the following auxiliary lemma.

\begin{lemma}\label{lem:fv} Let us consider a fixed a discrete memoryless channel $\pazocal{W}$ from $\XX$ to $\YY$. Then the map $f:[0,1]^{\XX}\to \mathbb{R}\cup\{+\infty\}$, defined in terms of Convention~\ref{conv:infty} as
\bb\label{eq:f(v)}
    f(v)\coloneqq -\log \sum_{y\in\YY}\exp\left(\sum_{x\in\XX}v(x)\log\pazocal{W}(y|x)\right),
\ee 
is monotone increasing, in the sense that, if $v(x)\geq u(x)$ for all $x\in\XX$ 
, then $f(v)\geq f(u)$. Furthermore, it is lower semi-continuous and, if $\|v\|_1=1$, non-negative: 
\bb\label{eq:unif_low_b}
    f(v) \geq  0\, .
\ee
\end{lemma}

\begin{proof}
    According to Convention~\ref{conv:infty}, $f$ is meant to be defined as
    \bb
        f(v)= -\log \sum_{y\in\YY_{\pazocal{W},v}}\exp\left(\sum_{x\in{\rm supp}(v)}v(x)\log\pazocal{W}(y|x)\right),
    \ee
    where  ${\rm supp}(v)\coloneqq\{x\in\XX\,:\, v(x)>0\}$ and $\YY_{\pazocal{W},v}\coloneqq \{y\in\YY\,:\,\pazocal{W}(y|x)>0\;\forall\, x\in {\rm supp}(v)\}$. Let us suppose that $v(x)\geq u(x)$ for all $x\in\XX$. Then, since $ \log\pazocal{W}(y|x)\leq 0$ for any $x\in\XX$ and $y\in\YY$ and ${\rm supp}(v)\supseteq {\rm supp}(u)$, we have
    \bb\label{eq:ineq_u_v}
        \sum_{x\in{\rm supp}(v)}v(x)\log\pazocal{W}(y|x)\leq \sum_{x\in{\rm supp}(u)}u(x)\log\pazocal{W}(y|x)
    \ee
    for any $y\in\YY_{\pazocal{W},v}\cap \YY_{\pazocal{W},u}$. Using again that ${\rm supp}(v)\supseteq {\rm supp}(u)$, we see that $\YY_{\pazocal{W},v}\subseteq \YY_{\pazocal{W},u}$. By the monotonicity and positivity of the exponential, we get
    \bb
        \sum_{y\in\YY_{\pazocal{W},v}}\exp\left(\sum_{x\in{\rm supp}(v)}v(x)\log\pazocal{W}(y|x)\right)&\leqt{(i)}\sum_{y\in\YY_{\pazocal{W},v}}\exp\left(\sum_{x\in{\rm supp}(v)}u(x)\log\pazocal{W}(y|x)\right)\\
        &\leqt{(ii)} \sum_{y\in\YY_{\pazocal{W},u}}\exp\left(\sum_{x\in{\rm supp}(v)}u(x)\log\pazocal{W}(y|x)\right),
    \ee
    where in (i) we have used~\eqref{eq:ineq_u_v}; in (ii) we have noticed that extending the sum over $\YY_{\pazocal{W},v}$ to a sum over $\YY_{\pazocal{W},u}$ introduces further positive terms.
    By anti-monotonicity of $-\log(\,\cdot\,)$ we conclude that $f(v)\geq f(u)$. Let $v\in[0,1]^\XX$ such that $\|v\|_1=1$. Then 
    \bb
        f(v)=-H(v) -\log \sum_{y\in\YY_{\pazocal{W},v}}\exp\left(\sum_{x\in{\rm supp}(v)}v(x)\log\left(v(x)\pazocal{W}(y|x)\right)\right)=\min_{Q_Y} D(vQ_y\|P^{v}_{\pazocal{W}}),
    \ee
    where we have leveraged the Gibbs variational principle (Lemma~\ref{thm:Gibbs}) in the last identity to introduce the minimization over $Q_Y\in\mathcal{P}(\YY)$ and we have defined $P^{v}_{\pazocal{W}}(x,y)\coloneqq v(x)\pazocal{W}(y|x)$. Then~\eqref{eq:unif_low_b} follows from the positivity of the relative entropy between probability distributions.
    The remaining part to be proved is the lower semicontinuity. For any fixed nonempty subset $\XX^\ast\subseteq\XX$ let us define $V_{\XX^\ast}\coloneqq \{v\in [0,1]^{\XX}\,:\, {\rm supp}(v)=\XX^\ast \}$. Then we claim that the restriction $f|_{V_{\XX^\ast}}:V_{\XX^\ast}\to \mathbb{R}$ is continuous. Indeed, let us first rewrite
    \bb
        f|_{V_{\XX^\ast}}(v)= -\log \sum_{y\in\YY_{\pazocal{W},\XX^\ast}}\exp\left(\sum_{x\in\XX^\ast}v(x)\log\pazocal{W}(y|x)\right),
    \ee
    where we have noticed that the support of $v\in (0,1]^{\XX^\ast}$ is always $\XX^\ast$ and we have introduced $\YY_{\pazocal{W},\XX^\ast}\coloneqq \{y\in\YY\,:\,\pazocal{W}(y|x)>0,\;\forall\, x\in \XX^\ast\}$. By noticing that $\XX$ and $\YY$ are finite, $\YY_{\pazocal{W},\XX^\ast}$ does not depend on $v\in V_{\XX^\ast}$, and $f|_{V_{\XX^\ast}}(v)$ is written as a composition of continuous functions, we conclude that $v\mapsto f|_{V_{\XX^\ast}}(v)$ is a continuous function. Let us consider now a generic sequence $\{v_k\}_{k\in\mathbb{N}}$ in $[0,1]^\XX$ that converges to $v_\infty\in [0,1]^\XX$. Let $\{u_k\}_{k\in\mathbb{N}}$ be the sequence defined as
    \bb
        u_k(x)=\begin{cases}
            v_k(x) & x\in{\rm supp}(v_\infty)\\
            0 & x\notin{\rm supp}(v_\infty)
        \end{cases}
    \ee
    for any $x\in\XX$ and $k\in\mathbb{N}$. It is clear that $u_k\to v_\infty$ as $k\to \infty$. Since $v_k(\,\cdot\,)\geq u_k(\,\cdot\,)$, we have $f(v_k)\geq f(u_k)$ by monotonicity of $f$. Furthermore, by the very definition of the convergence $v_k\to v_\infty$, there exists a $k^\ast\in\mathbb{N}$ such that $v_k(x)>0$ for any $x\in{\rm supp}(v_\infty)$ and any $k\geq k^\ast$. This means that, for $k\geq k^\ast$, ${\rm supp}(u_k)={\rm supp}(v_\infty)$, whence
    \bb
        f(v_k)\geq f(u_k) = f|_{V_{{\rm supp}(v_\infty)}}(u_k)\qquad \forall\, k\geq k^\ast.
    \ee
    By taking the liminf, and leveraging the continuity of $f|_{V_{{\rm supp}(v_\infty)}}$, we get
    \bb
        \liminf_{k\to \infty} f(v_k)\geq \liminf_{k\to\infty} f|_{V_{{\rm supp}(v_\infty)}}(u_k) = f|_{V_{{\rm supp}(v_\infty)}}(v_\infty) = f(v_\infty),
    \ee
    and this concludes the proof.
\end{proof}

We now have all the ingredients available to prove Proposition~\ref{thm:achievability}.

\begin{proof}[Proof of Proposition~\ref{thm:achievability}]
    Due to the bound~\eqref{eq:ineq_lautum} of Proposition~\ref{thm:inequality}, it is sufficient to prove~\eqref{eq:achievability}. We can represent the lower umlaut information $\ell_{k,q}$ as the expectation value
     \bb
        \ell_{k,u_k}(\pazocal{W},P_X)= \mathbb{E}_{X^k\sim P_X^k}\left[-\log \sum_{y\in\YY}\exp\left(\sum_{i=1}^k\frac{1}{k}\log\pazocal{W}(y|X_i)\right)\right]
    \ee
   where $P_X^k(x_1,\dots,x_k)=P_X(x_1)\cdots P_X(x_1)$. It is easy to see that the result of the sum $\sum_{i=1}^k\frac{1}{k}\log\pazocal{W}(y|X_i)$ depends only on the number of occurrences of each symbol $x\in\XX$ in the string $X^n$, namely
    \bb
        \sum_{i=1}^k\frac{1}{k}\log\pazocal{W}(y|X_i)=\sum_{x\in\XX}\frac{N(x|X^k)}{k}\log\pazocal{W}(y|x)
    \ee
    where $N(x|X^k)\coloneqq\sum_{i=1}^k\id_{X_i=x}$ is the number of occurrences of $x$ in $X^k$. Let us introduce
    \bb\label{def:def_Nx}
        \hat N_k^x\coloneqq \frac{N(x|X^k)}{k}=\frac{1}{k}\sum_{i=1}^k\id_{X_i=x}
    \ee
    which, by the weak law of the large numbers, converges in probability to its expectation value $\mathbb{E}\left[\hat N_k^x\right]=P_X(x)$. More precisely, let us fix $0<\epsilon<\min\{P_X(x)\,:\,x\in{\rm supp}(P_X) \}$, and let us call $\mathcal{E}_k^{(\epsilon)}$ the event $\left\{\exists\, x\in\XX:\big|\hat{N}_k^x-P_X(x)\big|>\epsilon \right\}$ and $\big(\mathcal{E}_k^{(\epsilon)}\big)^c$ its complement. Then
    \bb
        \lim_{k\to \infty}\mathbb{P}\left(\mathcal{E}_k^{(\epsilon)}\right)=0
    \ee If $x\notin {\rm supp} (P_X)$, then $\mathbb{P}(\hat N^x_k\neq 0 )=0$. We can therefore rewrite
    \bb
        \ell_{k,u_k}(\pazocal{W},P_X)&= \mathbb{E}_{X^k\sim P_X^k}\left[-\log \sum_{y\in\YY}\exp\left(\sum_{x\in{\rm supp}(P_X)}\hat N_k^x\log\pazocal{W}(y|x)\right)\right]
    \ee
    and lower bound
    \bb\label{eq:quant_est}
        \ell_{k,u_k}(\pazocal{W},P_X)
        &\eqt{(i)}\mathbb{E}_{X^k\sim P_X^k}\left[-\log \sum_{y\in\YY}\exp\left(\sum_{x\in{\rm supp}(P_X)}\hat N_k^x\log\pazocal{W}(y|x)\right)\,\middle|\,\mathcal{E}_k^{(\epsilon)}\,\right]\mathbb{P}\left(\mathcal{E}_k^{(\epsilon)}\right)\\
        &\quad+\mathbb{E}_{X^k\sim P_X^k}\left[-\log \sum_{y\in\YY}\exp\left(\sum_{x\in{\rm supp}(P_X)}\hat N_k^x\log\pazocal{W}(y|x)\right)\,\middle|\,(\mathcal{E}_k^{(\epsilon)})^c\,\right]\mathbb{P}\left((\mathcal{E}_k^{(\epsilon)})^c\right)\\
        &\geqt{(ii)}\mathbb{P}\left(\mathcal{E}_k^{(\epsilon)}\right)\cdot 0+\mathbb{P}\left(\big(\mathcal{E}_k^{(\epsilon)}\big)^c\right)\left(-\log \sum_{y\in\YY}\exp\left(\sum_{x\in{\rm supp}(P_X)}(P_X(x)-\epsilon)\log\pazocal{W}(y|x)\right)\right),
    \ee
    where 
    \begin{itemize}
        \item in (i) we condition on $\mathcal{E}_k^{(\epsilon)}$ and on its complement;
        \item in (ii) we
        introduce two lower bounds: since $v(x)\coloneqq \hat N^x_k$ is a probability distribution on $\XX$ (i.e. $\|v\|_1=1$) we use~\eqref{eq:unif_low_b} of Lemma~\ref{lem:fv} in order to lower bound the expectation value conditioned on $\mathcal{E}_k^{(\epsilon)}$; then, we notice that not only $v(x)$, but also $u(x)\coloneqq P_X(x)-\epsilon$ belongs to $[0,1]^\XX$ since we chose $0<\epsilon<\min\{P_X(x)\,:\,x\in{\rm supp}(P_X) \}$, and $v(x)\geq u(x)$ in $\big(\mathcal{E}_k^{(\epsilon)}\big)^c$, as $\big|\hat{N}_k^x-P_X(x)\big|\leq \epsilon$, hence we can use the monotonicity property of Lemma~\ref{lem:fv}.
    \end{itemize}
    Taking the liminf, we get
    \bb
        \liminf_{k\to\infty}\ell_{k,u_k}(\pazocal{W},P_X)\geq -\log \sum_{y\in\YY}\exp\left(\sum_{x\in{\rm supp}(P_X)}(P_X(x)-\epsilon)\log\pazocal{W}(y|x)\right);
    \ee
    by arbitrariness of $0<\epsilon<\min\{P_X(x)\,:\,x\in{\rm supp}(P_X) \}$, we let $\epsilon\to 0^+$:
    \bb
        \liminf_{k\to\infty}\ell_{k,u_k}(\pazocal{W},P_X)&\geq \liminf_{\epsilon\to 0^+} \left(-\log \sum_{y\in\YY}\exp\left(\sum_{x\in{\rm supp}(P_X)}(P_X(x)-\epsilon)\log\pazocal{W}(y|x)\right)\right)\\
        &\geqt{(iii)} -\log \sum_{y\in\YY}\exp\left(\sum_{x\in\XX}P_X(x)\log\pazocal{W}(y|x)\right)\\
        &=U(X;Y),
    \ee
    where (iii) holds because of the lower semi-continuity property proved in Lemma~\ref{lem:fv}.
    Combining this result with Proposition~\ref{thm:inequality}, we conclude that
    \bb
        U(X;Y)\geq \liminf_{k\to\infty}\ell_{k,u_k}(\pazocal{W},P_X) \geq U(X;Y),
    \ee
    and this completes the proof.
\end{proof}

Having established the achievability of the umlaut information via $\errL{\pazocal{W}}$ in the large list limit (Theorem~\ref{thm:interpr_upper}), we can further upper bound the gap for finite list sizes $L$ as $O\big(\sqrt{\log L/L}\big)$.

\begin{prop}[(Quantitative estimate of the gap)]
\label{thm:quant_est} 
    Let $\pazocal{W}$ be a discrete memoryless channel from $\XX$ to $\YY$. Let $\bar{P}_X\in\mathcal{P}(\XX)$ be the probability distribution achieving the maximum in the definition of channel umlaut information, set $\bar{p}_{\min}\coloneqq\min_{x\in{\rm supp}\left(\bar{P}_X\right)}\bar{P}_X(x)$, define
    $\pazocal{W}_{\min}\coloneqq \min\{\pazocal{W}(y|x): x\in\XX,\, y\in\YY, \, \pazocal{W}(y|x)\neq 0\}$ and pick some $\epsilon\in(0,1)$.
Then, we have
    \bb\label{eq:quant_est2}
    0\leq U(\pazocal{W})-\errL{\pazocal{W}}\leq |\XX|\exp\left(-\frac{\epsilon^2}{2}\bar p_{\min}(L+1)\right)-\epsilon  \log\pazocal{W}_{\min}.
\ee
    In particular, calling
\bb
    \epsilon(L)=\sqrt{\frac{1}{\bar p_{\min}}\frac{\log(L+1)}{L+1}},
\ee
we can set $\epsilon=\min\{\epsilon(L),1/2\}$, getting 
\bb
|U(\pazocal{W})-\errL{\pazocal{W}}|=O\left(\left(\frac{\log L}{L}\right)^{1/2}\right)
\ee
for $L\to \infty$.
\end{prop}

\begin{proof}[Proof of Proposition~\ref{thm:quant_est}]
We proceed similarly to the proof of Theorem~\ref{thm:interpr_upper}. After fixing $P_X\in\mathcal{P}(\XX)$, let $\epsilon\in(0,1)$ and, for every $x\in\mathrm{supp}(P_X),$ let $\mathcal{F}_{k,x}^{(\epsilon)}$ be the event $\hat N_k^x\leq P_X(x)(1-\epsilon)$. Since $k\hat N^k_x$, by the very definition of $\hat N_k^x$, has a binomial distribution (see eq.~\eqref{def:def_Nx}), by the Chernoff bound on the lower tail, we have 
\bb
    \mathbb{P}\left(\mathcal{F}_{k,x}^{(\epsilon)}\right)=\mathbb{P}\left(\hat N_k^x\leq P_X(x)(1-\epsilon)\right)
    \leq \exp\left(-\frac{\epsilon^2}{2}P_X(x)k\right).
\ee
Therefore, the event
\bb
    \mathcal{F}_{k}^{(\epsilon)}=\left\{\exists x\in\mathrm{supp}(P_X)\,:\, \hat N_k^x\leq P_X(x)(1-\epsilon)\right\}=\bigcup_{x\in\mathrm{supp}(P_X)}\mathcal{F}_{k,x}^{(\epsilon)}
\ee
has probability at most
\bb
    \mathbb{P}\left(\mathcal{F}_{k}^{(\epsilon)}\right)\leq \sum_{x\in\mathrm{supp}(P_X)} \exp\left(-\frac{\epsilon^2}{2}P_X(x)k\right)\leq |\XX|\exp\left(-\frac{\epsilon^2}{2}p_{\min}k\right),
\ee
where $p_{\min}=\min_{x\in{\rm supp}(P_X)}P_X(x)$. We call $p_\epsilon\coloneqq \mathbb{P}\left(\mathcal{F}_k^{(\epsilon)}\right)$ to unburden the notation. 
Similarly to~\eqref{eq:quant_est}:
\bb
        \ell_{k,u_k}(\pazocal{W},P_X)&= \mathbb{E}_{X^k\sim P_X^k}\left[-\log \sum_{y\in\YY}\exp\left(\sum_{x\in{\rm supp}(P_X)}\hat N_k^x\log\pazocal{W}(y|x)\right)\right]\\
        &\eqt{(i)}\mathbb{E}_{X^k\sim P_X^k}\left[-\log \sum_{y\in\YY}\exp\left(\sum_{x\in{\rm supp}(P_X)}\hat N_k^x\log\pazocal{W}(y|x)\right)\,\middle|\,\mathcal{E}_k^{(\epsilon)}\,\right]p_\epsilon\\
        &\quad+\mathbb{E}_{X^k\sim P_X^k}\left[-\log \sum_{y\in\YY}\exp\left(\sum_{x\in{\rm supp}(P_X)}\hat N_k^x\log\pazocal{W}(y|x)\right)\,\middle|\,(\mathcal{E}_k^{(\epsilon)})^c\,\right](1-p_\epsilon)\\
        &\geqt{(ii)}\left(-\log \sum_{y\in\YY}\exp\left(\sum_{x\in{\rm supp}(P_X)}P_X(x)(1-\epsilon)\log\pazocal{W}(y|x)\right)\right)(1-p_\epsilon),
    \ee
    where 
    \begin{itemize}
        \item in (i) we condition on $\mathcal{F}_k^{(\epsilon)}$ and on its complement;
        \item in (ii) we make two lower bounds: since $v(x)\coloneqq \hat N^x_k$ is a probability distribution on $\XX$ (i.e. $\|v\|_1=1$) we use~\eqref{eq:unif_low_b} of Lemma~\ref{lem:fv} in order to lower bound with zero the expectation value conditioned on $\mathcal{F}_k^{(\epsilon)}$; then, we notice that $[0,1]^\XX$ $v(x)\geq u(x)\coloneqq P_X(x)(1-\epsilon)$ in $\big(\mathcal{F}_k^{(\epsilon)}\big)^c$, hence we can use the monotonicity property of Lemma~\ref{lem:fv}.
    \end{itemize}
Now, we can go further:
\bb\label{eq:first_estimate}
    \ell_{k,u_k}&(\pazocal{W},P_X)
    \\
    &\geq (1-p_\epsilon)\left(-\log \sum_{y\in\YY_{\pazocal{W},P_X}}\exp\left(\sum_{x\in{\rm supp}(P_X)}P_X(x)(1-\epsilon)\log\pazocal{W}(y|x)\right)\right)\\
    &=
    (1-p_\epsilon)\left(-\log \sum_{y\in\YY_{\pazocal{W},P_X}}\exp\left(-\epsilon\sum_{x\in {\rm supp}(P_X)}P_X(x)\log\pazocal{W}(y|x)\right)\exp\left(\sum_{x\in\XX}P_X(x)\log\pazocal{W}(y|x)\right)\right)\\
    &\geq
   \epsilon(1-p_\epsilon)\min_{y\in\YY_{\pazocal{W},P_X}}\sum_{x\in {\rm supp}(P_X)}P_X(x)\log\pazocal{W}(y|x)\\
    &\qquad+(1-p_\epsilon)\left(-\log \sum_{y\in\YY_{\pazocal{W},P_X}}\exp\left(\sum_{x\in\XX}P_X(x)\log\pazocal{W}(y|x)\right)\right)\\
    &\geq
    \epsilon\min_{y\in\YY_{\pazocal{W},P_X}}\sum_{x\in \XX}P_X(x)\log\pazocal{W}(y|x)+(1-p_\epsilon)U(X;Y),
\ee
where $X$ and $Y$ are random variables taking values in $\XX$ and $\YY$, respectively, with joint probability distribution $P_{XY}(x,y)=\pazocal{W}(y|x)P_X(x)$. Let $\bar{P}_{X}$ be the optimiser in the definition of channel umlaut information, namely $U(\pazocal{W})=U(\bar X;\bar Y)$, with $\mathbb{P}(\bar X=x, \bar Y=y)=\pazocal{W}(y|x)\bar{P}_X(x)$. Then, by~\eqref{eq:first_estimate},
\bb
    \max_{P_X}\ell_{k,u_k}(\pazocal{W},P_X)-U(\bar X;\bar Y)\geq \epsilon\min_{y\in\YY_{\pazocal{W},\bar{P}_X}}\sum_{x\in {\rm supp}(\bar{P}_X)}\bar P_X(x)\log\pazocal{W}(y|x)-p_\epsilon U(\bar X;\bar Y),
\ee
whence, setting $k=L+1$, we get
\bb\label{eq:bound_U_l}
    U(\pazocal{W})-\errL{\pazocal{W}}\leq p_\epsilon U(\pazocal{W})+\epsilon K_{\pazocal{W},\bar{P}_X},
\ee
where
\bb
    K_{\pazocal{W},\bar{P}_X}&\coloneqq \max_{y\in\YY_{\pazocal{W},\bar{P}_X}}\sum_{x\in {\rm supp}(\bar{P}_X)}-\bar P_X(x)\log\pazocal{W}(y|x)
    \leq -\log \pazocal{W}_{\min},
    \ee
    with $\pazocal{W}_{\min}\coloneqq \min\{\pazocal{W}(y|x): x\in\XX,\, y\in\YY, \, \pazocal{W}(y|x)\neq 0\}$.
Using the estimate on $p_\epsilon$ and calling $\bar{p}_{\min}\coloneqq \min_{x\in{\rm supp}(\bar{P}_X)}\bar{P}_X(x)$,~\eqref{eq:bound_U_l} can be upper bounded as 
\bb\label{eq:bound_U_l2}
    U(\pazocal{W})-\errL{\pazocal{W}}\leq |\XX|\exp\left(-\frac{\epsilon^2}{2}\bar p_{\min}(L+1)\right)-\epsilon\log\pazocal{W}_{\min},
\ee
and this concludes the proof.
\end{proof}


\subsection{An SDP bound on the unassisted zero-rate error exponent}

Let us rewrite~\eqref{eq:plain} as
\bb
    \errpl{\pazocal{W}}&=-\min_{P_X}\sum_{x,x'\in\XX}P_X(x)P_X(x')\log\sum_{y\in\YY}\sqrt{\pazocal{W}(y|x)\pazocal{W}(y|x')}\\
    &=\max_{P_X}\sum_{x,x'\in\XX}P_X(x)P_X(x')A_{xx'}=\max_{p}p^\intercal Ap=\max_{p}\Tr[B_pA]
\ee
where we have introduced $A_{xx'}\coloneqq -\log\sum_{y\in\YY}\sqrt{\pazocal{W}(y|x)\pazocal{W}(y|x')}$, $p$ as the vector of the probability distribution $P_X$ and $B_p\coloneqq pp^\intercal$.
$B$ belongs to the set of \textit{completely positive matrices} $\cp_{|\XX|}$, defined as
\bb
    \cp_d\coloneqq\left\{B\in\mathbb{R}^{d\times d}\,:\, B=\sum_{i=1}^kv_iv_i^\intercal,\,v_i\in\mathbb{R}^d_+,\, k\in\mathbb{N}\right\},
\ee
where $\mathbb{R}^d_+$ is the cone of vectors in $\mathbb{R}^d$ with non-negative entries. Let $u_d$ be the uniform vector in $\mathbb{R}^d$ given by $u_d=(1,\dots,1)^\intercal$. It holds that $u_{|\XX|}^\intercal B_p u_{|\XX|}=1$. Conversely, let us suppose that $B\in cp_{|\XX|}$ and $u_{|\XX|}^\intercal B u_{|\XX|}=1$. To unburden the notation, from now on the dimension of vectors and matrices will be implied. On the one hand, for some $k\in\mathbb{N}$,
\bb
    1 = u^\intercal B u=\sum_{i=1}^k (u^\intercal v_i)^2 \eqt{(i)} \sum_{i=1}^kq_i,
\ee
where in (i) we have introduced $q_i\coloneqq (u^\intercal v_i)^2$. On the other hand,
\bb
    B = \sum_{i=1}^k q_i\frac{v_i}{u^\intercal v_i}\frac{v_i^\intercal}{u^\intercal v_i} = \sum_{i=1}^k q_ip_ip_i^\intercal,
\ee
where $p_i\in\mathbb{R}^{|\XX|}$ are probability vectors,
whence
\bb
    \Tr[BA]=\sum_{i=1}^k q_i p_i^\intercal A p_i\leq  \max_p p^\intercal A p = \max_p \Tr[B_p A].
\ee
We can therefore rewrite
\bb
    \errpl{\pazocal{W}}=\max_{\substack{B\in \cp\\u^\intercal B u=1}}\Tr[AB].
\ee
Let $\mathrm{NN}_d$ be the cone of $d\times d$ matrices with non-negative entries and let $\mathrm{PSD}_d$ be the cone of $d\times d$ positive semi-definite matrices. It is easy to see that
\bb
    \cp_d\subseteq \mathrm{NN}_d \cap \mathrm{PSD}_d,
\ee
while it is non-trivial that
\bb
    \cp_d = \mathrm{NN}_d \cap \mathrm{PSD}_d,
\ee
if and only if $d\leq 4$~\cite[Theorem 2.4, Example 2.7]{compl_pos}.
We have therefore proved the following statement.
\begin{prop}[(An SDP for the plain zero-rate error exponent)] Let $\pazocal{W}$ be a discrete memoryless channel from $\XX$ to $\YY$. Then, the plain zero-rate error exponent can be written as
\bb
    \errpl{\pazocal{W}}=\max_{\substack{B\in \cp\\u^\intercal B u=1}}\Tr[A^{(\pazocal{W})}B],
\ee
    where $A^{(\pazocal{W})}$ is the $|\XX|\times|\XX|$ matrix
\bb
    A_{xx'}^{(\pazocal{W})}\coloneqq -\log\sum_{y\in\YY}\sqrt{\pazocal{W}(y|x)\pazocal{W}(y|x')}.
\ee
Furthermore, the plain zero-rate error exponent is upper bounded by the following SDP:
\bb
    \errpl{\pazocal{W}}\leq \max_{\substack{B\in \mathrm{PSD}\cap \mathrm{NN}\\u^\intercal B u=1}}\Tr[A^{(\pazocal{W})}B]
\ee
and the inequality becomes an equality for $|\XX|\leq 4$.
\end{prop}


\subsection{Examples}

In this subsection, we will compare the performance of some channel in the regime of zero rate with and without non-signalling assistance.

\bigskip

\textbf{Binary symmetric channel}. Let $q\in (0,1)$, $\XX=\YY=\{0,1\}$ and let
\bb
    \pazocal{W}_q(y|x)=\begin{cases}
        1-q & y=x\\
        q & y\neq x
    \end{cases}
\ee
be the binary symmetric channel. Then, by Proposition~\ref{prop:lautum_channel_formula}, we have
\bb
    U(\pazocal{W}_q)=-\log\min_{0\leq p\leq 1}\Big(\exp\left({p\log(1-q)+(1-p)\log q}\right)+\exp\left({p\log q+(1-p)\log (1-q)}\right)\Big),
\ee
which can be computed by solving
\bb
    \exp\big({p\log(1-q)+(1-p)\log q}\big)\log\frac{1-q}{q}+\exp\big(p\log q+(1-p)\log (1-q)\big)\log\frac{q}{1-q}=0,
\ee
i.e.
\bb
    p\log(1-q)+(1-p)\log q=p\log q+(1-p)\log (1-q),
\ee
which yields $p=1/2$ for every $q\in(0,1)$, whence
\bb\label{eq:BSC}
    U(\pazocal{W}_q)=-\log\left(2\exp\left(\frac{1}{2}\log q(1-q)\right)\right)=-\log \sqrt{q(1-q)}-\log 2.
\ee
Let us compare this result with the plain zero rate error exponent for the binary symmetric channel provided by~\eqref{eq:plain}.
\bb
    \errpl{\pazocal{W}_q}&=-\min_{0\leq p\leq 1}\left(p^2\log(1-q+q)+2p(1-p)\log(2\sqrt{q(1-q)})+(1-p)\log(q+1-q)\right)\\
    &\eqt{(i)}-2\min_{0\leq p\leq 1}p(1-p)\log(2\sqrt{q(1-q)})\\
    &=-\frac{1}{2}\left(\log \sqrt{q(1-q)}+\log 2\right),
\ee
where in (i) we have used that $-cp(1-p)$ with $c>0$ is minimised when $p=1/2$.
Hence,
\bb
    \errns{\pazocal{W}_q} = U(\pazocal{W}_q) = 2\errpl{\pazocal{W}_q}.
\ee

\bigskip

\textbf{Binary erasure channel}. Let $q\in (0,1)$, $\XX=\{0,1\}$ and $\YY=\{0,1,\epsilon\}$; let
\bb
    \pazocal{W}_q(y|x)=\begin{cases}
        1-q & y=x\\
        q & y= \epsilon\\
        0 &\text{otherwise}
    \end{cases}
\ee
be the binary erasure channel. Again by Proposition~\ref{prop:lautum_channel_formula}, using~\eqref{eq:alternative_W_0} we have
\bb
    \exp\left({-U(\pazocal{W}_q)}\right)&=\min_{0\leq p\leq 1}q^{p}q^{1-p}=q
\ee
whence
\bb
    U(\pazocal{W}_q)=-\log q
\ee
Comparing this umlaut information with the plain zero rate error exponent given by~\eqref{eq:plain} we get
\bb
    \errpl{\pazocal{W}_q}=-\min_{0\leq p\leq 1}2p(1-p)\log q=-\frac{1}{2}\log q,
\ee
so
\bb
    \errns{\pazocal{W}_q} = U(\pazocal{W}_q) = 2\errpl{\pazocal{W}_q}.
\ee

\bigskip

\textbf{The Gaussian channel.} Let $X$ be a generic random variable taking values in $\XX=\mathbb{R}^n$ having a fixed covariance matrix $C\in\mathbb{R}^{n\times n}$, and let $Y$ be the random output in $\YY=\mathbb{R}^k$ given by
\bb
Y=HX+N
\ee
where $H$ is a $k \times n$ deterministic matrix and $N$ is a random Gaussian vector in $\YY$ independent of $X$, with mean $m\in\mathbb{R}^k$ and covariance matrix $V\in\mathbb{R}^{k\times k}$. In this case, the conditional probability characterising the Gaussian channel $\pazocal{W}_{H,m,V}$ is
\bb
    \pazocal{W}_{H,m,V}(y|x) = \mathcal{G}(m+Hx,V)(y)\qquad\forall \,x\in\mathbb{R}^n,\,y\in\mathbb{R}^k.
\ee
Then, by a generalization of Proposition~\ref{prop:lautum_channel_formula} to the continuous variable setting, when $X$ is constrained to have covariance matrix $C$ the umlaut information is given by
\bb
        U(\pazocal{W}_{H,m,V})&=-\log\min_{p_X}\int_{\mathbb{R}^k}dy\,\exp\left(\int_{\mathbb{R}^n}dx\,p_X(x)\log\mathcal{G}(m+Hx,V)(y)\right)\\
        &=-\log\min_{p_X}\int_{\mathbb{R}^k}dy\,\exp\left(-\frac{1}{2}\int_{\mathbb{R}^n}dx\,p_X(x)(y-m-Hx)^\intercal V^{-1}(y-m-Hx) \right)\\
        &\qquad +\frac{1}{2}\log \left((2\pi)^n\det V\right).
\ee
We need to evaulate the following integral
\bb
    \int_{\mathbb{R}^n}dx&\,p_X(x)(y-m-Hx)^\intercal V^{-1}(y-m-Hx)\\
    &=\int_{\mathbb{R}^n}dx\,p_X(x)\Tr[V^{-1}(y-m-Hx)(y-m-Hx)^\intercal]\\
    &=(y-m-H\mu)^\intercal V^{-1}(y-m-H\mu)+\Tr[V^{-1}HCH^\intercal]
\ee
where $\mu \coloneqq \mathbb{E}_{p_{X}}[X]\in\mathbb{R}^n$ is the mean of $X$ and $C$ is the covariance of $X$, as defined above. Therefore
\bb
    U(\pazocal{W}_{H,m,V})
    &=-\log\min_{\mu}\frac{1}{\sqrt{(2\pi)^n\det V}}\int_{\mathbb{R}^k}dy\,\exp\left(-\frac{1}{2}(y-m-H\mu)^\intercal V^{-1}(y-m-H\mu)\right)\\
    &\qquad +\frac{1}{2}\Tr[CH^\intercal V^{-1}H]\\
    &=\frac{1}{2}\Tr[CH^\intercal V^{-1}H]\, ,
\ee
provided that $k\leq n$ and $H$ is of rank $k$. None of the previous expressions turned out to depend on the other degrees of freedom of $p_X$ apart from the covariance matrix $C$, i.e. the minimum appearing in the formula for the differential umlaut information for the Gaussian channel is achieved by any distribution with covariance matrix $C$ (and finite second moments).


\subsection{Comparison with channel lautum information}

Let $\pazocal{W}$ be a discrete memoryless channel $\pazocal{W}$ from $\XX$ to $\YY$ and let $X$ and $Y$ be random variables taking values in $\XX$ and $\YY$ with joint distribution
    \bb
        P_{XY}(x,y)=\pazocal{W}(y|x)P_{X}(x).
    \ee
Palomar and Verd\'{u} \cite{Lautum_08} considered the lautum information of the joint probability distribution of the input and the output of $\pazocal{W}$. They also study its maximization over the probability distributions of the input $X$, which we will denote by convenience
\bb
    L(\pazocal{W}) &\coloneqq \max_{P_X}L(X\!:\!Y)=    \max_{P_X}D(P_{X}P_Y\|P_{XY}),
\ee
where $P_Y$ is the marginal of $P_{XY}$ on $\YY$. As usual, when we consider $n$ uses of the channel $\pazocal{W}$, the joint probability distribution of $X^n$ and $Y^n$ is given by
\bb\label{eq:many_uses}
        P_{X^nY^n}(x_1,\dots,x_n,y_1,\dots,y_n)=P_{X}(x_1,\dots,x_n)\prod_{i=1}^n\pazocal{W}(y_i|x_i).
    \ee
for every $x_1,\dots,x_n\in\XX$ and $y_1,\dots,y_n\in\YY$

We are going to briefly recall a few results which provide some insights about the difference between $U(\pazocal{W})$ and $L(\pazocal{W})$.

\begin{prop}[({\cite[Theorem~3]{Lautum_08}})]\label{prop:ineqPV}
Let $\pazocal{W}$ be a discrete memoryless channel $\pazocal{W}$ from $\XX$ to $\YY$ and let $X^n$ and $Y^n$ as in~\eqref{eq:many_uses}. Then
    \bb
        L(X^n\!:\!Y^n)\geq \sum_{i=1}^nL(X_i\!:\!Y_i)
    \ee
    with equality if and only if $(Y_1,\dots, Y_n)$ are independent.
\end{prop}

On the contrary, since in the proof of~\eqref{eq:looks_additive} in Proposition~\ref{prop:lautum_channel_formula} the maximization over $P_X$ does not play any specific role --- namely, it is immediate to see that the same identity holds without the maximisation --- we can proceed as in Corollary~\ref{cor:add_chan} to immediately have the following statement.

\begin{prop} Let $\pazocal{W}$ be a discrete memoryless channel $\pazocal{W}$ from $\XX$ to $\YY$ and let $X^n$ and $Y^n$ as in~\eqref{eq:many_uses}. Then
    \bb
        U(X^n ; Y^n)= \sum_{i=1}^nU(X_i ; Y_i).
    \ee
\end{prop}

Let
\bb
    L^{\infty}(\pazocal{W}) &\coloneqq \liminf_{n\to\infty}\frac{1}{n}L(\pazocal{W}^{\times n}).
\ee
Proposition~\ref{prop:ineqPV} implies that $L^\infty(\pazocal{W})\geq L(\pazocal{W})$, as
\bb
    \max_{P_{X^n}}L(X^n:Y^n)\geq \max_{P_{X^n}}\sum_{i=1}^nL(X_i:Y_i) \geq \max_{P_{X}^n}\sum_{i=1}^nL(X_i:Y_i)=n \max_{P_{X_1}} L(X_1:Y_1),
\ee
where in the second inequality we have considered the ansatz of $n$ independent copies of a probability distribution $P_X\in\mathcal{P}(\XX)$.
Palomar and Verd\'{u} also provide the following operational interpretation of $L^\infty$ as a upper bound to the plain zero-rate error exponent.
\begin{prop}[(Theorem 10 in \cite{Lautum_08})] Let $\pazocal{W}$ be a discrete memoryless channel $\pazocal{W}$ from $\XX$ to $\YY$. Then
\bb
    \errpl{\pazocal{W}}\leq L^{\infty}(\pazocal{W}).
\ee 
\end{prop}
This result can be seen as a corollary of what we proved in this work, as
\bb
    L^{\infty}(\pazocal{W}) = \liminf_{n\to\infty}\frac{1}{n}L(\pazocal{W}^{\times n})\geqt{(i)}\liminf_{n\to\infty}\frac{1}{n}U(\pazocal{W}^{\times n})\eqt{(ii)} U(\pazocal{W})\eqt{(iii)}\errns{\pazocal{W}}\geq \errpl{\pazocal{W}},
\ee
where in (i) we have used that, by their very definition, $L\geq U$, in (ii) we have recalled the additivity of the channel umlaut information (Corollary~\ref{cor:add_chan}) and in (iii) we have leveraged the operational interpretation provided in Theorem~\ref{thm:exact_tho}. Palomar and Verd\'{u} provided an example that we can mention to prove that the inequality (i) can be strict.

\begin{figure}
    \centering
    \includegraphics[width=0.7\linewidth]{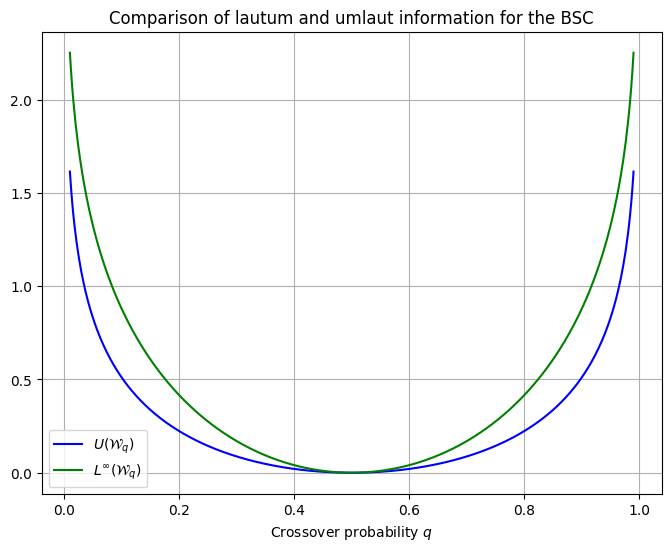}
    \caption{Regularised latum information and umlaut information for the binary symmetric channel (BSC) as functions of the crossover probability $q$.}
    \label{fig:LandU}
\end{figure}

\begin{prop}[(Theorem 13 in \cite{Lautum_08})] Let $q\in (0,1)$, $\XX=\YY=\{0,1\}$ and let
\bb
    \pazocal{W}_q(y|x)=\begin{cases}
        1-q & y=x\\
        q & y\neq x
    \end{cases}
\ee
be the binary symmetric channel. Then
\bb
    L(\pazocal{W}_q)&=\frac{1}{2}\log\frac{1}{4q(1-q)},
    L^{\infty}(\pazocal{W}_q)&=\left(\frac{1}{2}-q\right)\log\frac{1-q}{q}.
\ee
\end{prop}
We proved above -- see~\eqref{eq:BSC} -- that
\bb
    U(\pazocal{W}_q)=\frac{1}{2}\log\frac{1}{4q(1-q)},
\ee
which incidentally is equal to $L(\pazocal{W}_q)$ for this particular channel. Therefore, $L^{\infty}(\pazocal{W}_q)>U(\pazocal{W}_q)$ for $q\neq 0,1/2,1$, as we can see in Figure~\ref{fig:LandU}. Finally, in Figure~\ref{fig:recap} we summarise the hierarchy of the information measures and of the error exponents that were discussed in this paper.

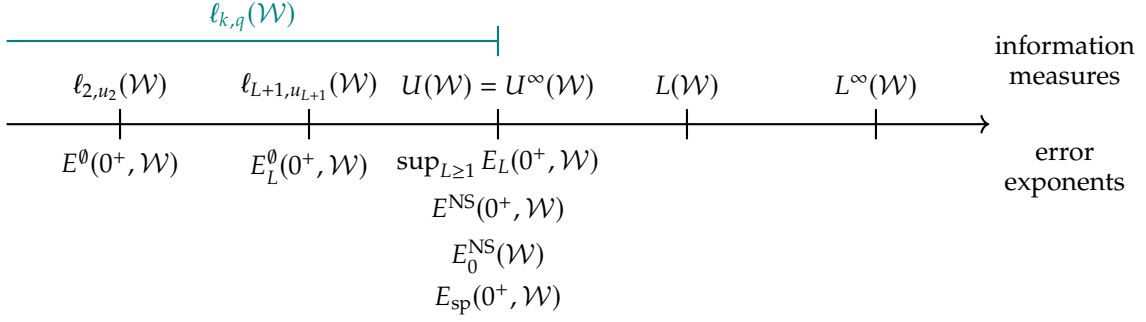
\begin{figure}
    \centering

\begin{tikzpicture}
    \draw[thick, ->] (2,0) -- (15,0);

    \draw[thick] (3.5,0.2) -- (3.5,-0.2);
    \draw[thick] (6,0.2) -- (6,-0.2); 
    \draw[thick] (8.5,0.2) -- (8.5,-0.2); 
    \draw[thick] (11,0.2) -- (11,-0.2); 
    \draw[thick] (13.5,0.2) -- (13.5,-0.2);

    \node[above] at (3.5,0.2) {$\ell_{2,u_{2}}(\pazocal{W})$};
    \node[above] at (6,0.2) {$\ell_{L+1,u_{L+1}}(\pazocal{W})$};
    \node[above] at (8.5,0.2) {$U(\pazocal{W})=U^{\infty}(\pazocal{W})$};
    \node[above] at (11,0.2) {$L(\pazocal{W})$};
    \node[above] at (13.5,0.2) {$L^\infty(\pazocal{W})$};
        \node[below] at (8.5,-0.2) {$\sup_{L\geq1}E_L(0^+,\pazocal{W})$};
    \node[below] at (8.5,-0.8) {\errns{\pazocal{W}}};
    \node[below] at (8.5,-1.4) {\opns{\pazocal{W}}};
    \node[below] at (8.5,-2.0) {$E_{\mathrm{sp}}(0^+,\pazocal{W})$};
    \node[below] at (6,-0.2) {\errL{\pazocal{W}}};
    \node[below] at (3.5,-0.2) {\errpl{\pazocal{W}}};

    \draw[thick, teal] (8.5,1.3) -- (8.5,0.9);
    \draw[thick, teal] (2,1.1) -- (8.5,1.1);
    \node[above, teal] at (5.25,1.1) {$\ell_{k,q}(\pazocal{W})$};


    \node[above] at (16,0.8) {information};
    \node[above] at (16,0.4) {measures};
    \node[above] at (16,-0.6) {error};
    \node[above] at (16,-1.1) {exponents};
\end{tikzpicture}

    \caption{A pictorial comparison of the information measures and of the error exponents that were discussed in this paper. $\pazocal{W}$ is a discrete memoryless channel from $\XX$ to $\YY$; for any $k\geq 1$, $q\in\mathbb{R}^k$ is a vector such that $\sum_{i=1}^kq_i=1$ and $u_k\in\mathbb{R}^k$ represents the uniform vector $u_k =(1/k,\dots, 1/k)$.}
    \label{fig:recap}
\end{figure}


\section{Outlook}
\label{sec:outlook}

As the umlaut information of channels quantifies the non-signalling--assisted zero-rate error exponent, which is different from the unassisted error exponent, there is the intriguing open question about how the exponent behaves under different types of assistance, and in particular about the entanglement-assisted zero-rate error exponent. Typically, quantum correlations can help in one-shot settings in (classical) Shannon theory for point-to-point problems, but become asymptotically useless in the sense that the unassisted, entanglement-assisted and non-signalling-assisted values become asymptotically equal \cite{Siddharth18}.\footnote{See, however, \cite{PhysRevA.95.052329} and follow-up works about network settings.} Here, however, the non-signalling--assisted value is different from the classical one, and so it is completely unclear where the entanglement-assisted value lands.  In that sense, the situation could be similar to zero-error Shannon theory, where the zero-error channel capacities are all different for the unassisted, entanglement-assisted, and non-signalling--assisted settings~\cite{Cubitt11}.

More generally, it would be interesting to explore more the idea of focusing more on quality in terms of small error probabilities instead of quantity in terms of optimal rates, which motivates the study of other types of zero-rate error exponents, both for classical and quantum Shannon theory~\cite{lami2024asymptotic}.


\acknowledgments

We thank Hao-Chung Cheng for enlightening correspondence on low-rate codes. MB and AO acknowledge funding from the European Research Council (ERC Grant Agreement No.~948139) and the Excellence Cluster Matter and Light for Quantum Computing (ML4Q). LL acknowledges financial support from the European Union under the European Research Council (ERC Grant Agreement No.~101165230) and from MIUR (Ministero dell'Istruzione, dell'Universit\`a e della Ricerca) through the project `Dipartimenti di Eccellenza 2023--2027' of the `Classe di Scienze' department at the Scuola Normale Superiore. MT is supported by the Ministry of Education through grant T2EP20124-0005 and by the National Research Foundation, Singapore through the National Quantum Office, hosted in A*STAR, under its Centre for Quantum Technologies Funding Initiative (S24Q2d0009).


\bibliography{biblio}


\appendix

\section{Auxiliary proofs}\label{app:missing}

{
\renewcommand{\thethm}{\ref{thm:rev_meta}}
\begin{prop}
For any channel $\pazocal{W}$ from $\XX$ to $\YY$ and any $M \in \N$ we have
    \bb
       -\log \epsilon^{\rm NS}(M,\pazocal{W})=\max_{P_X}\min_{Q_Y} D^{1/M}_H(P_XQ_Y\|P_{XY}),
    \ee
    where $P_{XY}(x,y)=\pazocal{W}(y|x)P_X(x)$.
    In particular,
    \bb
\opns{\pazocal{W}}&=\lim_{\delta\to 0}\liminf_{n\to \infty}\frac{1}{n}\max_{P_{X^n}}\min_{Q_{Y^n}} D^{\delta}_H(P_{X^n}Q_{Y^n}\|P_{X^nY^n})
    \ee
    where $P_{X^nY^n}(x_1,\dots,x_n,y_1,\dots,y_n)= P_{X^n}(x_1,\dots,x_n)\prod_{i=1}^n\pazocal{W}(y_i|x_i)$.
\end{prop}
}
\begin{proof}
    Using Matthews' linear programming formulation of the error of non-signalling coding~\cite[Proposition~13]{Matthews2012}, we get
    \bb
        \epsilon^{\rm NS}(M, \pazocal{W}) &=   1 - \max_{\substack{R\\P_X\in\mathcal{P}(\XX)}} \Bigg\{ \begin{aligned}[t] &\sum_{x\in\XX} \sum_{y\in\YY} \pazocal{W}(y|x)\, R_{xy} \,:\\
        & \sum_{x\in\XX} R_{xy} \leq \frac{1}{M} \quad \forall y \in \YY,\quad 0 \leq R_{xy} \leq P_X(x) \quad \forall x \in \XX, y \in \YY\Bigg\} \end{aligned}\\
        &= \min_{\substack{R'\\P_X\in\mathcal{P}(\XX)}} \Bigg\{ \begin{aligned}[t] &\sum_{x\in\XX} \sum_{y\in\YY} \pazocal{W}(y|x)\, R'_{xy} \,:\\
        & \sum_{x\in\XX} R'_{xy} \geq 1 - \frac{1}{M} \quad \forall y \in \YY,\quad 0 \leq R'_{xy} \leq P_X(x) \quad \forall x \in \XX, y \in \YY\Bigg\},
    \end{aligned}
    \ee
    where we defined $R'$ through $R'_{xy} = P_X(x) - R_{xy}$. Introducing a test $T$ such that $R'_{xy} = P_X(x)\, T_{xy}$, we can continue as
        \bb
        - \log &\epsilon^{\rm NS}(M, \pazocal{W}) \\
        &= -\log  \min_{\substack{T\\P_X\in\mathcal{P}(\XX)}} \Bigg\{ \begin{aligned}[t] &\sum_{x\in\XX} \sum_{y\in\YY} \pazocal{W}(y|x)\, P_X(x)\, T_{xy} \,:\\
        & \sum_{x\in\XX} P_X(x)\, T_{xy}  \geq 1 - \frac{1}{M} \quad \forall y \in \YY,\quad 0 \leq T_{xy} \leq 1 \quad \forall x \in \XX, y \in \YY \Bigg\} 
    \end{aligned}\\
    &= -\log  \min_{P_X \in \mathcal{P}(\XX)} \min_{T} \Bigg\{ \begin{aligned}[t] &\sum_{x\in\XX} \sum_{y\in\YY} P_{XY}(x,y)\, T_{xy} \,:\\
        & \min_{Q_Y \in \mathcal{P}(\YY)} \sum_{x\in\XX} P_X(x) \, Q_Y(y) \, T_{xy} \geq 1 - \frac{1}{M},\quad 0 \leq T_{xy} \leq 1 \quad \forall x \in \XX, y \in \YY\Bigg\} 
    \end{aligned}\\
        &= -\log  \min_{P_X \in \mathcal{P}(\XX)} \max_{Q_Y \in \mathcal{P}(\YY)} \min_{T} \Bigg\{ \begin{aligned}[t] &\sum_{x\in\XX} \sum_{y\in\YY} P_{XY}(x,y)\, T_{xy} \,:\\
        & \sum_{x\in\XX} P_X(x) \, Q_Y(y) \, T_{xy} \geq 1 - \frac{1}{M},\quad 0 \leq T_{xy} \leq 1 \quad \forall x \in \XX, y \in \YY\Bigg\} 
    \end{aligned}\\
&=  \max_{P_X \in \mathcal{P}(\XX)} \min_{Q_Y \in \mathcal{P}(\YY)} D^{1/M}_H(P_XQ_Y\|P_{XY}),
    \ee
where we used von Neumann's minimax theorem in the second-to-last line and recalled the definition of $D^{\epsilon}_H$~\eqref{D_H_iid} in the last one.

    This gives
    \bb
        \opns{\pazocal{W}}&= \lim_{M\to\infty}\liminf_{n\to \infty}-\frac{1}{n} \log \epsilon^{\rm NS}(M, \pazocal{W}^{\times n})\\
        &=\lim_{M\to\infty}\liminf_{n\to \infty}\frac{1}{n}\max_{P_{X^n}}\min_{Q_{Y^n}} D^{1/M}_H(P_{X^n}Q_{Y^n}\|P_{X^nY^n})\\
        &=\lim_{\delta\to 0}\liminf_{n\to \infty}\frac{1}{n}\max_{P_{X^n}}\min_{Q_{Y^n}} D^{\delta}_H(P_{X^n}Q_{Y^n}\|P_{X^nY^n}).
    \ee
     where $P_{X^nY^n}(x_1,\dots,x_n,y_1,\dots,y_n)\coloneqq P_{X^n}(x_1,\dots,x_n)\prod_{i=1}^n\pazocal{W}(y_i|x_i)$.
\end{proof}

\section{Continuous alphabets}\label{app:cv}

The aim of this section is the discussion of the generalization of the operational interpretation of the umlaut information for probability distributions to the setting of continuous alphabets.

\begin{Def}[(Umlaut information, continuous alphabets)]
Given two random variables $X$ and $Y$ taking values in $\XX$ and $\YY$. Let $\mu_{XY}$ be the law of $(X,Y)$ and let $\mu_X$ be the law of $X$. The umlaut information is defined as
\bb\label{eq:lautum_cv}
U(X;Y)&\coloneqq \inf_{\nu}D(\mu_X\times \nu\,\|\,\mu_{XY})
\ee
where $\nu$ is any probability measure on $\YY$.
\end{Def}

Let us suppose that $\XX=\mathbb{R}^{n_1}$ and $\YY=\mathbb{R}^{n_2}$, with $\mu_{XY}$ absolutely continuous with respect to the Lebesgue measure $\lambda_{\XX}\times\lambda_\YY$.
Then we can identify the law of $(X,Y)$ with a probability density function $p_{XY}(x,y)$. We can rewrite~\eqref{eq:lautum_cv} as
\bb\label{eq:lautum_cv2}
U(X;Y)&\coloneqq \inf_{q_Y}D(p_Xq_Y\|p_{XY})=\inf_{q_Y}\int_{\XX\times\YY}p_X(x)q_Y(y)\log\frac{p_X(x)q_Y(y)}{p_{XY}(x,y)}dx\,dy
\ee
where $q_Y$ is any probability distribution on $\YY$ and $p_X$ is the marginal of $p_{XY}$ on $\YY$. In principle,~\eqref{eq:lautum_cv2} does not account for all the measures $\nu$ appearing in the minimisation of~\eqref{eq:lautum_cv}. However, it is easy to see that the measures $\nu$ which are not absolutely continuous with respect to the Lebesgue measure on $\YY$ do not contribute to the minimization. Indeed, for any $\nu$ we have the following two possibilities.
\begin{itemize}
    \item $\mu_X\times \nu\not\ll \mu_{XY}$: this means that $D(\mu_X\times \nu\|\mu_{XY})=+\infty$, so $\nu$ does not contribute to the minimum;
    \item $\mu_X\times \nu\ll \mu_{XY}$: since $\mu_{XY}\ll \lambda$, we also have that $\mu_X\times\nu\ll\lambda$; given a measurable set\footnote{Like a ball centered in the origin of $\mathbb{R}^{n_1}$ with a suitably large radius.} $S\subseteq \mathcal{X}$ such that $0<\lambda_\mathcal{X}(S)<\infty$ and $\mu_X(S)>0$, we have
    \bb
        \mu_X(S)\nu(\,\cdot\,)=(\mu_X\otimes \nu)(S \times \,\cdot\,)\ll (\lambda_\mathcal{X}\times\lambda_\mathcal{Y})(S \times \,\cdot\,)=\lambda_\mathcal{X}(S)\lambda_\mathcal{Y}(\,\cdot\,),
    \ee
    and this proves that any measure $\nu$ that contributes to the minimum is absolutely continuous with respect to the Lebesgue measure on $\mathcal{Y}$.
\end{itemize}

\begin{lemma}[(Continuous variables Gibbs variational principle)]\label{lem:cv_Gibbs} Let $a:\mathbb{R}^n\to\mathbb{R}$ be a function such that the integral $\int_{\mathbb{R}^n} e^{-a(x)}dx$ is finite. For any probability density $p$, the integral $\int_{\mathbb{R}^n} p(x)\log\left(p(x)e^{a(x)}\right)dx$ is well defined and satisfies the inequality
    \begin{equation}\label{eq:gibbs_cv}
       \int_{\mathbb{R}^n} p(x)\log\left(p(x)e^{a(x)}\right)dx\geq - \log\int_{\mathbb{R}^n} e^{-a(x)}dx,
    \end{equation}
    with equality if and only if
    \begin{equation}
        p(x)=\frac{e^{-a(x)}}{\int_{\mathbb{R}^n} e^{-a(x)}dx},
    \end{equation}
    almost everywhere.
\end{lemma}


\begin{proof}
 Since $\log t\leq t-1$, we have
\bb
    I(x)\coloneqq p(x)\log\left(p(x)e^{a(x)}\right)= p(x)\left(-\log\frac{e^{-a(x)}}{p(x)}\right)\geq p(x)\left(1-\frac{e^{-a(x)}}{p(x)}\right)=p(x)-e^{-a(x)},
\ee
whence, if we denote by $t_\pm=\frac{1}{2}(t\pm |t|)$, then $I(x)=I_+(x)+I_-(x)$ and
\bb
    I_-(x)=(p(x)-e^{-a(x)})_-\geq -e^{-a(x)}.
\ee
The lower bound is integrable by hypothesis, so the integral $\int_{\R^n}I(x)dx\in (-\infty,+\infty]$ is well defined.
    We rewrite
    \bb
        \int_{\mathbb{R}^n} p(x)\log\left(p(x)e^{a(x)}\right)dx&=\int_{{\rm supp}(p)} p(x)\left(-\log \frac{e^{-a(x)}}{p(x)}\right)dx\\
        &\geqt{(i)}- \log\int_{{\rm supp}(p)} e^{-a(x)}dx\\
        &\geqt{(ii)}- \log\int_{\mathbb{R}^n} e^{-a(x)}dx
    \ee
    where (i) is Jensen inequality for the convex function $t\mapsto -\log t$ and (ii) follows from the monotonicity of the logarithm and the positivity of $e^{-a(x)}$. In particular, since the convexity is strict, (i) is an equality if and only if $e^{-a(x)}/p(x)$ is constant and (ii) is an equality if and only if $p$ has full support. This concludes the proof.
\end{proof}

Leveraging the Gibbs variational principle, we can rewrite~\eqref{eq:lautum_cv} as
\bb
    U(X;Y)&\coloneqq \inf_{q_Y}D(p_Xq_Y\|p_{XY})=-\log\int_\YY\exp\left(\int_{\XX}p_X(x)\log\frac{p_{XY}(x,y)}{p_X(x)}dx\right)dy
\ee
Indeed, let us consider
\bb
    a(y) \coloneqq \int_{\XX}p_X(x)\log\frac{p_X(x)}{p_{XY}(x,y)}dx.
\ee
Using Jensen inequality for the convex function $t\mapsto -\log t$, we get the lower bound
\bb
    a(y) \coloneqq \int_{\XX}p_X(x)\left(-\log\frac{p_{XY}(x,y)}{p_X(x)}\right)dx \geq -\log\int_{\XX}p_{XY}(x,y)=-\log p_Y(y),
\ee
which implies the upper bound
\bb
    0\leq \int_\YY e^{-a(y)}dy\leq \int_\YY p_Y(y) = 1, 
\ee
so the hypothesis of Lemma \ref{lem:cv_Gibbs} is satisfied and we can compute
\bb
    D(p_Xq_Y\|p_{XY}) &=\int_{\XX\times \YY} q_Y(y)p_X(x)\log\frac{p_X(x)q_Y(y)}{p_{XY}(x,y)}dx\,dy\\
    &=\int_{\YY} q_Y(y)\left(\int_{\XX}p_X(x)\log\frac{p_X(x)}{p_{XY}(x,y)}dx+\log q_Y(y)\right)dy\\
    &=\int_{\YY} q_Y(y)\log \left(q_Y(y)\exp\int_{\XX}p_X(x)\log\frac{p_X(x)}{p_{XY}(x,y)}dx\right)dy\\
    &\geq -\log\int_\YY\exp\left(\int_{\XX}p_X(x)\log\frac{p_{XY}(x,y)}{p_X(x)}dx\right)dy.
\ee
In particular, the lower bound is attained by $\ddot q_Y(y)\propto \exp\left(\int_\mathcal{X}p_X(x)\log p_{XY}(x,y)\right) $, so the infimum corresponds to a minimum.

Now, similarly to Definition~\ref{def:renyi-umlaut}, we can introduce the R\'enyi $\alpha$-umlaut information for continuous alphabets. Given $\alpha\in(0,1)\cup (1,\infty)$, we will call
\bb
U_\alpha(X;Y)&\coloneqq \inf_{Q_Y}D_\alpha(P_XQ_Y\|P_{XY})
\ee
where $D_\alpha$ is the R\'enyi $\alpha$-relative entropy for continuous probability distributions:
\bb
    D_{\alpha}(P\|Q)\coloneqq \frac{1}{\alpha-1}\log\int P^\alpha(x)Q^{1-\alpha}(x)dx\, .
\ee
The R\'enyi $\alpha$-umlaut information for continuous alphabets has the closed-form expression (see~\eqref{eq:U_a})
\bb\label{eq:U_a_c}
    U_\alpha(X;Y) 
        =-\log \int_\YY  P_Y(y)  \exp\left(-D_\alpha(P_X\|P_{X|Y=y})\right)dy
\ee
and it is additive with a proof identical to the one of Proposition~\ref{thm:closed_alpha}.

\begin{lemma}\label{lem:alpha_to_1_cont}
    Given two random variables $X$ and $Y$ taking values in $\XX=\mathbb{R}^{n_1}$ and $\YY=\mathbb{R}^{n_2}$ and having an absolutely continuous law with respect to the Lebesgue measure, it holds that
    \bb
        \lim_{\alpha\to 1^-}U_\alpha(X;Y)=U(X;Y).
    \ee
\end{lemma}

\begin{rem} In order to prove Lemma~\ref{lem:alpha_to_1_cont} for continuous alphabets, we cannot proceed as in the first proof of Lemma~\ref{lem:alpha_to_1}. Indeed, even if 
\begin{itemize}
    \item $\alpha\mapsto D_\alpha(P_XQ_Y\|P_{XY})$ is monotone \cite[Theorem 39]{vanErven2014},
    \item $Q_Y\mapsto D_\alpha(P_XQ_Y\|P_{XY})$ is lower semi-continuous with respect to the weak convergence \cite[Theorem 19]{vanErven2014},
\end{itemize}
we cannot use the Mosonyi--Hiai minimax theorem~\cite[Corollary A2]{MosonyiHiai} as $\mathcal{P}(\YY)$, the set of probability distributions on $\YY$, is not compact in general when endowed with the weak convergence (e.g. $\mathcal{P}(\mathbb{R})$ is not compact). However, the closed-form expression~\eqref{eq:U_a_c} is sufficient to prove the lemma.
\end{rem}

\begin{proof}[Proof of Lemma~\ref{lem:alpha_to_1_cont}]
    Leveraging the closed-form expression~\eqref{eq:U_a_c}, we have
    \bb
        \lim_{\alpha\to 1^-}U_\alpha(X;Y) 
        &=-\log \lim_{\alpha\to 1^-}\int_\YY  P_Y(y)  \exp\left(-D_\alpha(P_X\|P_{X|Y=y})\right)dy\\
        &\eqt{(i)}-\log \int_\YY P_Y(y)  \exp\left(-\lim_{\alpha\to 1^-}D_\alpha(P_X\|P_{X|Y=y})\right)dy\\
        &\eqt{(ii)}-\log \int_\YY P_Y(y)  \exp\left(-D(P_X\|P_{X|Y=y})\right)dy\\
        &=-\log \int_\YY P_Y(y)  \exp\left(\int_\XX P_X(x)\log P_{XY}(x,y)dx-\log P_{Y}(y)+H(P_X)\right)dy\\
        &=-H(P_X)-\log \int_\YY\exp\left(\int_\XX P_X(x)\log P_{XY}(x,y)dx\right)dy\\
        &=U(X;Y)
    \ee
    where in (ii) we have used that $\lim_{\alpha\to 1^-}D_\alpha(p\|q)=D(p\|q)$ \cite{vanerven2010renyi} and in (i) we have used Lebesgue's dominated convergence theorem: since for all $\alpha\in(0,1)$ $D_\alpha(p\|q)$ is non-negative \cite{vanerven2010renyi}, we have the domination
    \bb
        0 \leq P_Y(y)  \exp\left(-D_\alpha(P_X\|P_{X|Y=y})\right)\leq P_Y(y) \in L^1(\YY),
    \ee
    which ensures that we can commute the limit with the integral.
\end{proof}

The operational interpretation of the umlaut information is valid also in the case of continuous alphabets. In order to prove this, we are going to leverage the following Lemma, which requires a preliminary definition: given a probability distribution $p$ on $\XX=\mathbb{R}^k$ and a finite (Lebesgue)-measurable partition $\pazocal{P}=\{A_1,\dots,A_n\}$ of cardinality $n$, we define the probability distribution $p|_\pazocal{P}$ on $[n]$ as $p|_\pazocal{P}(i)\coloneqq \int_{A_i} p(x)dx$ for $i=1,\dots, n$.

\begin{lemma}[({\cite[Theorem 10]{vanErven2014}})]\label{lem:partitions}  Let $p,q$ two probability distributions. For any $\alpha\in[0,\infty]$, we have
\bb
    D_\alpha(p\|q)=\sup_\pazocal{P}D_\alpha(p|_\pazocal{P}\|q|_\pazocal{P})
\ee
    where the supremum is over all finite measurable partitions $\pazocal{P}$ of $\XX$.
\end{lemma}

\begin{thm}[(Operational interpretation of the umlaut information, continuous alphabets)]\label{thm:op_int_cont_alph}
Given a joint probability distribution $P_{XY}\in\mathcal{P}(\XX\times\YY)$, let $P_X$ be the marginal on $\XX$. Then it holds that
\bb
U(X;Y)=\mathrm{Stein}\big(\FF^{P_X}\,\big\|\, P_{XY}\big)\, .
\ee
\end{thm}

\begin{proof} 
Let $\alpha\in(0,1)$, and, as usual, let $P_{XY}^{\times n}\in\mathcal{P}(\XX^n\times\YY^n)$ be the i.i.d.\ distribution 
\bb
P^{\times n}_{XY}(x_1,\dots,x_n,y_1,\dots,y_n)=P_{XY}(x_1,y_1)\cdots P_{XY}(x_n,y_n)\, .
\ee
We denote as $P^{\times n}_X$ its marginal on $\XX^n$. Eq.~\eqref{eq:Dalpha} still holds for continuous alphabets. Indeed, let $p$ and $q$ be two probability distributions on $\mathbb{R}^k$. Then
\bb
D^\epsilon_H(p\|q)\geqt{(a)} D^\epsilon_H(p|_\pazocal{P}\|q|_\pazocal{P})\geqt{(b)} D_\alpha(p|_\pazocal{P}\|q_\pazocal{P})+\frac{\alpha}{1-\alpha}\log\frac{1}{\epsilon}
\ee
where in (a) we have used data-processing inequality for $D^\epsilon_H$, with $\pazocal{P}$ being a measurable finite partition of $\mathbb{R}^k$, and in (b) we have used~\eqref{eq:Dalpha} for probability distributions on finite sets. By arbitrariness of $\pazocal{P}$, we get
\bb
D^\epsilon_H(p\|q)\geq \sup_{\pazocal{P}}D_\alpha(p|_\pazocal{P}\|q_\pazocal{P})+\frac{\alpha}{1-\alpha}\log\frac{1}{\epsilon}=D_\alpha(p\|q)+\frac{\alpha}{1-\alpha}\log\frac{1}{\epsilon},
\ee
where the last equality follows from Lemma~\ref{lem:partitions}. Now,
\bb
\mathrm{Stein}\big(\FF^{P_X}\,\big\|\,P_{XY}\big) &\eqt{(i)} \lim_{\epsilon\to 0}\liminf_{n\to \infty}\frac{1}{n}\min_{Q_{Y^n}\in\mathcal{P}(\YY^n)}D^\epsilon_H(P_X^{\times n} Q_{Y^n}\|P_{XY}^{\times n})\\
&\geq \lim_{\epsilon\to 0}\liminf_{n\to \infty}\frac{1}{n}\min_{Q_{Y^n}\in\mathcal{P}(\YY^n)}\left(D_\alpha(P_X^{\times n} Q_{Y^n}\|P_{XY}^{\times n})+\frac{\alpha}{1-\alpha}\log\frac{1}{\epsilon}\right)\\
&=\liminf_{n\to \infty}\frac{1}{n}U_\alpha(X^n;Y^n)\\
&\eqt{(ii)} U_\alpha(X;Y),
\label{proof_interp_L2_eq3b}
\ee
where we used~\eqref{D_H_minimax_product_testing} in~(i), and the additivity of $U_\alpha$ in~(ii). In particular,
\bb
\mathrm{Stein}\big(\FF^{P_X}\,\big\|\,P_{XY}\big) &\geq \limsup_{\alpha\to 1^-}U_\alpha(X;Y) \\
&\eqt{(iii)} U(X;Y),
\ee
where in~(iii) we employed Lemma~\ref{lem:alpha_to_1_cont}.
For the upper bound we consider the ansatz 
\bb
Q_{Y^n}(y_1,\dots,y_n) = Q_Y^{\times n}(y_1,\ldots,y_n) = Q_Y(y_1)\cdots Q_Y(y_1),
\ee 
where $Q_Y$ is an arbitrary fixed probability density on $\YY$. Then, continuing from the first line of~\eqref{proof_interp_L2_eq3b}, we have
\bb\label{eq:stein_upper}
\mathrm{Stein}\big(\FF^{P_X}\,\big\|\,P_{XY}\big) &\leq \lim_{\epsilon\to 0}\liminf_{n\to \infty}\frac{1}{n}D^\epsilon_H(P_X^nQ_Y^n\|P_{XY}^n)\\
&\eqt{(iv)}D(P_XQ_Y\|P_{XY})
\ee
where (iv)~follows from the Stein lemma for continuous alphabets, for which we are going to provide a coincise proof at the end. Minimising~\eqref{eq:stein_upper} over $Q_Y$ yields
\bb
\mathrm{Stein}\big(\FF^{P_X}\,\big\|\,P_{XY}\big) \leq \min_{Q_Y} D(P_XQ_Y\|P_{XY}) = U(X;Y)\, ,
\ee
concluding the proof of Theorem~\ref{thm:op_int_cont_alph}.
For completeness' sake, we give a short proof of the Stein lemma for continuous alphabets.
\bb
\mathrm{Stein}(p\|q) &= \lim_{\epsilon\to 0}\liminf_{n\to \infty}\frac{1}{n}D^\epsilon_H(p^{\times n}\|q^{\times n})
\geqt{(e)}\lim_{\epsilon\to 0}\liminf_{n\to \infty}\frac{1}{n}D^\epsilon_H(p|_\pazocal{P}^{\times n}\|q|_\pazocal{P}^{\times n})\eqt{(f)}D(p|_\pazocal{P}\|q|_\pazocal{P}),
\ee
where (e) is data processing inequality and (f) is the Stein lemma for probability distribution on finite sets. By arbitrariness of $\pazocal{P}$, we get
\bb
    \mathrm{Stein}(p\|q)\geq \sup_\pazocal{P}D(p|_\pazocal{P}\|q|_\pazocal{P})\eqt{(g)}D(p\|q).
\ee
where (g) follows from Lemma~\ref{lem:partitions}. The proof of the weak converse is more standard, let us consider a generic test that maps a generic probability density $r(x)$ into a binary distribution\footnote{which can be interpreted as the probability of accepting the null hypothesis ($r_0$) or the alternative hypothesis ($r_1$).} $(r_0,r_1)$ according to an acceptance function $A:\XX\to [0,1]$ as follows
\bb
    r_0=\int_\XX A(x)r(x)dx, \qquad r_1=1-r_0.
\ee
Being $r\mapsto(r_0,r_1)$ a channel, by data processing inequality for the relative entropy, we have
\bb
    D(p\|q)\geq D\left((p_0,p_1)\|(q_0,q_1)\right)
\ee
If we constrain the type I error probability $p_1$ to match a certain threshold $\epsilon$, we get
\bb
    D(p\|q)&\geq (1-\epsilon)\log\frac{1-\epsilon}{q_0}+\epsilon\log\frac{\epsilon}{1-q_0}\\
    &\geqt{(h)}(1-\epsilon)\log\frac{1-\epsilon}{q_0}-\epsilon\left(\frac{1-q_0}{\epsilon}-1\right)\\
    &\geq (1-\epsilon)\log\frac{1-\epsilon}{q_0}-1,
\ee
where in (h) we have used the inequality $\log x\leq x-1$, i.e. $\log x\geq -(\frac{1}{x}-1)$, for $x>0$.
Minimising the type II error probability $q_0$ over all the possible acceptance functions constrained to $p_1=\epsilon$, we get
\bb
    D(p\|q)&\geq (1-\epsilon)\left(\log(1-\epsilon)+D_H^\epsilon(p\|q)\right)-1,
\ee
whence, by the additivity of the LHS,
\bb
    D(p\|q)&\geq \lim_{\epsilon\to 0}\liminf_{n\to\infty}\frac{1}{n} \left((1-\epsilon)\left(\log(1-\epsilon)+D_H^\epsilon(p^n\|q^n)\right)-1\right)=\mathrm{Stein}(p\|q).
\ee
This concludes the short argument for the Stein lemma for continuous alphabets.
\end{proof}


\end{document}

%% file: marcos.tex
\usepackage{newpxtext,newpxmath}

\let\coloneqq\relax

\usepackage[utf8]{inputenc}

\usepackage{amsthm}
\usepackage{amssymb}
\usepackage{amsmath}
\usepackage{bbm}
\usepackage[pdftex, backref=page]{hyperref}
\usepackage{braket}
\usepackage{mathdots}
\usepackage{mathtools}
\usepackage{enumerate}
\usepackage[shortlabels]{enumitem}
\usepackage{csquotes}
\usepackage[cal=boondox]{mathalfa}
\usepackage{graphicx}
\usepackage{stackengine}
\usepackage{scalerel}
\usepackage{tensor}       
\usepackage[dvipsnames,svgnames]{xcolor}
\usepackage{array}
\usepackage{makecell}
\newcolumntype{x}[1]{>{\centering\arraybackslash}p{#1}}
\usepackage{tikz}
\usepackage{pgfplots}
\usetikzlibrary{shapes.geometric, shapes.misc, positioning, arrows, arrows.meta, decorations.pathreplacing, decorations.pathmorphing, patterns, angles, quotes, calc}
\usepackage{booktabs}
\usepackage{xfrac}
\usepackage{siunitx}
\usepackage{centernot}
\usepackage{comment}
\usepackage{chngcntr}
\usepackage{caption}
\usepackage{subcaption}

\newtheorem{thm}{Theorem}
\newtheorem*{thm*}{Theorem}
\newtheorem{prop}[thm]{Proposition}
\newtheorem*{prop*}{Proposition}
\newtheorem{lemma}[thm]{Lemma}
\newtheorem*{lemma*}{Lemma}
\newtheorem{cor}[thm]{Corollary}
\newtheorem*{cor*}{Corollary}

\newtheorem*{cj*}{Conjecture}
\newtheorem{Def}[thm]{Definition}
\newtheorem*{Def*}{Definition}

\newtheorem*{question*}{Question}

\newtheorem*{problem*}{Problem}

\makeatletter
\def\thmhead@plain#1#2#3{%
  \thmname{#1}\thmnumber{\@ifnotempty{#1}{ }\@upn{#2}}%
  \thmnote{ {\the\thm@notefont#3}}}
\let\thmhead\thmhead@plain
\makeatother

\theoremstyle{definition}
\newtheorem{rem}[thm]{Remark}

\hypersetup{
    bookmarksnumbered=true, 
    breaklinks=true,
    unicode=false, 
    pdfstartview={FitH}, 
    pdfnewwindow=true, 
    colorlinks=true, 
    allcolors=MidnightBlue!70!black!70!TealBlue
}

\newcommand{\bb}{\begin{equation}\begin{aligned}\hspace{0pt}}
\newcommand{\bbb}{\begin{equation*}\begin{aligned}}
\newcommand{\ee}{\end{aligned}\end{equation}}
\newcommand{\eee}{\end{aligned}\end{equation*}}

\newcommand{\eqt}[1]{\stackrel{\mathclap{\mbox{\scriptsize #1}}}{=}}
\newcommand{\leqt}[1]{\stackrel{\mathclap{\mbox{\scriptsize #1}}}{\leq}}

\newcommand{\geqt}[1]{\stackrel{\mathclap{\mbox{\scriptsize #1}}}{\geq}}

\newcommand{\e}{\varepsilon}
\renewcommand{\epsilon}{\varepsilon}

\newcommand{\id}{\mathbbm{1}}
\newcommand{\R}{\mathbb{R}}
\newcommand{\N}{\mathbb{N}}

\DeclareMathOperator{\Tr}{Tr}

\DeclareMathAlphabet{\pazocal}{OMS}{zplm}{m}{n}

\newcommand{\XX}{\mathcal{X}}
\newcommand{\YY}{\mathcal{Y}}

\newcommand{\FF}{\mathcal{F}}

\newcommand{\lsmatrix}{\left(\begin{smallmatrix}}
\newcommand{\rsmatrix}{\end{smallmatrix}\right)}

\newcommand{\deff}[1]{\textbf{\emph{#1}}}

\stackMath

\stackMath

\makeatletter
\newcommand*\rel@kern[1]{\kern#1\dimexpr\macc@kerna}
\newcommand*\widebar[1]{%
  \begingroup
  \def\mathaccent##1##2{%
    \rel@kern{0.8}%
    \overline{\rel@kern{-0.8}\macc@nucleus\rel@kern{0.2}}%
    \rel@kern{-0.2}%
  }%
  \macc@depth\@ne
  \let\math@bgroup\@empty \let\math@egroup\macc@set@skewchar
  \mathsurround\z@ \frozen@everymath{\mathgroup\macc@group\relax}%
  \macc@set@skewchar\relax
  \let\mathaccentV\macc@nested@a
  \macc@nested@a\relax111{#1}%
  \endgroup
}

\counterwithin*{equation}{part}
\counterwithin*{thm}{part}
\counterwithin*{figure}{part}

\tikzset{meter/.append style={draw, inner sep=10, rectangle, font=\vphantom{A}, minimum width=30, line width=.8, path picture={\draw[black] ([shift={(.1,.3)}]path picture bounding box.south west) to[bend left=50] ([shift={(-.1,.3)}]path picture bounding box.south east);\draw[black,-latex] ([shift={(0,.1)}]path picture bounding box.south) -- ([shift={(.3,-.1)}]path picture bounding box.north);}}}
\tikzset{roundnode/.append style={circle, draw=black, fill=gray!20, thick, minimum size=10mm}}
\tikzset{squarenode/.style={rectangle, draw=black, fill=none, thick, minimum size=10mm}}

\definecolor{Blues5seq1}{RGB}{239,243,255}
\definecolor{Blues5seq2}{RGB}{189,215,231}
\definecolor{Blues5seq3}{RGB}{107,174,214}
\definecolor{Blues5seq4}{RGB}{49,130,189}
\definecolor{Blues5seq5}{RGB}{8,81,156}

\definecolor{Greens5seq1}{RGB}{237,248,233}
\definecolor{Greens5seq2}{RGB}{186,228,179}
\definecolor{Greens5seq3}{RGB}{116,196,118}
\definecolor{Greens5seq4}{RGB}{49,163,84}
\definecolor{Greens5seq5}{RGB}{0,109,44}

\definecolor{Reds5seq1}{RGB}{254,229,217}
\definecolor{Reds5seq2}{RGB}{252,174,145}
\definecolor{Reds5seq3}{RGB}{251,106,74}
\definecolor{Reds5seq4}{RGB}{222,45,38}
\definecolor{Reds5seq5}{RGB}{165,15,21}

\allowdisplaybreaks

\usepackage[most,breakable]{tcolorbox}
\renewenvironment{boxed}[1][white]%
{\expandafter\ifstrequal\expandafter{#1}{filled}{\begin{tcolorbox}[colback=gray!3,colframe=gray!20,breakable=false,enhanced,left=5.75pt,right=5.75pt,grow sidewards by=10pt]}{\begin{tcolorbox}[colback=MidnightBlue!70!black!70!TealBlue!4!white,colframe=MidnightBlue!70!black!70!TealBlue!50!white,breakable=false,enhanced,left=5.75pt,right=5.75pt,grow sidewards by=10pt]}}%
  {\end{tcolorbox}}

\let\cp\relax
\newcommand{\cp}{\mathrm{CP}}

\newcommand{\uml}[1]{\ddot{#1}}